\documentclass[12pt,times]{article}
\usepackage[margin=1in]{geometry}
\usepackage{multirow}
\usepackage{graphicx,float}
\usepackage[title]{appendix}    
\usepackage[table]{xcolor} 
\usepackage{tikz}
\usepackage[square, comma, numbers]{natbib}
\usepackage{setspace}
\usepackage{amsmath, amssymb, enumerate, amsthm,amsfonts,mathrsfs,mathtools,bbm}
\usepackage{booktabs}
\usepackage{algorithm, algpseudocode}
\usepackage{fancyhdr}
\usepackage{dirtytalk}
\MakeRobust{\say}
\usepackage{nameref}
\usepackage{comment}
\usepackage{titlesec}
\usepackage{subfig}

\usepackage{placeins}
\usepackage{etoolbox}
\preto\subsection{\FloatBarrier}

\usepackage[colorlinks=true,
            pdfpagemode=UseNone,
            urlcolor=blue,
            citecolor=blue,
            linkcolor=violet,
            bookmarks=true,
            backref=page
            ]{hyperref}
\usepackage{cleveref}          
\titleformat{\paragraph}
{\normalfont\normalsize\bfseries}{\theparagraph}{1em}{}
\titlespacing*{\paragraph}
{0pt}{3.25ex plus 1ex minus .2ex}{1.5ex plus .2ex}

\newcommand{\curly}{\mathrel{\leadsto}}

\newcommand{\cond}{\curly}
\newcommand{\re}{\mathbb{R}}

\newcommand{\GS}{\mathrm{GS}}
\newcommand{\cM}{\mathcal{M}}
\newcommand{\cD}{\mathcal{D}}

\newcommand{\cA}{\mathcal{A}}
\newcommand{\cN}{\mathcal{N}}

\newcommand{\cL}{\mathcal{L}}

\newcommand{\bfX}{\mathbf{X}}

\newcommand{\bfY}{\mathbf{Y}}
\newcommand{\bfx}{\mathbf{x}}
\newcommand{\bfz}{\mathbf{z}}
\newcommand{\bfy}{\mathbf{y}}

\newcommand{\bbN}{\mathbb{N}}

\newcommand{\reals}{{\mathbb R}}
\newcommand{\nats}{{\mathbb N}}

\newcommand{\ignore}[1]{{}}

\newcommand{\floor}[1]{{\lfloor #1 \rfloor}}
\newcommand{\ceil}[1]{{\lceil #1 \rceil}}
\newcommand{\seq}[1]{{$({#1}_n)$}}

\newcommand{\expect}{{\mathbb{E}}}
\newcommand{\eps}{\epsilon}
\newcommand{\eqd}{\overset{d}{=}}

\DeclareMathOperator{\Erf}{erf}
\newcommand{\erf}[1]{{\Erf \left(#1\right)}}
\newcommand{\erfinv}[1]{{\Erf^{-1} \left(#1\right)}}

\DeclareMathOperator*{\var}{Var}

\newtheorem{theorem}{Theorem}

\newtheorem{lemma}[theorem]{Lemma}
\newtheorem{proposition}[theorem]{Proposition}

\theoremstyle{remark}
\newtheorem{remark}[theorem]{Remark}

\theoremstyle{definition}

\newtheorem{definition}{Definition}

\newtheorem{condition}{Condition}

\numberwithin{equation}{section}
\numberwithin{theorem}{section}
\numberwithin{example}{section}
\numberwithin{definition}{section}

\title{Differentially private scale testing via rank transformations and percentile modifications}
\author{Joshua Levine and Kelly Ramsay}
\date{}

\begin{document}
\maketitle
\begin{abstract}
     We develop a class of differentially private two-sample scale tests, called the rank-transformed percentile-modified Siegel--Tukey tests, or RPST tests. These RPST tests are inspired both by recent differentially private extensions of some common rank tests and some older modifications to non-private rank tests. We present the asymptotic distribution of the RPST test statistic under the null hypothesis, under a very general condition on the rank transformation. We also prove RPST tests are differentially private, and that their type I error does not exceed the given level. 
We uncover that the growth rate of the rank transformation presents a tradeoff between power and sensitivity. 
 We do extensive simulations to investigate the effects of the tuning parameters and compare to a general private testing framework. Lastly, we show that our techniques can also be used to improve the differentially private signed-rank test.
 \noindent {\it Key words and phrases:}
Differential privacy, robust statistics, rank testing, nonparametric statistics, scale testing.
\end{abstract}
\section{Introduction}
Over the last two decades, differential privacy has emerged as one of the main modern privacy frameworks \citep{dwork2006calibrating}. One set of tools within this framework are differentially private hypothesis tests. Various differentially private hypothesis tests have been introduced. These include task specific hypothesis tests for common problems, such as ANOVA, regression, categorical data and more \citep{ cai2017priv,awan2018differentially,campbell2018differentially,couch2018differentially,barrientos2019differentially,couch2019differentially,Dette2022}, as well as methods to extend any non-private test to the private setting \citep{pena2022differentially,kazan2023test,kim2023differentially}. 
One common testing problem, scale testing, has not yet been carefully studied. 
To fill this gap, we develop a class of differentially private two-sample scale tests, called the \textit{rank-transformed percentile-modified Siegel--Tukey tests}, or RPST tests. 

These RPST tests are inspired both by recent differentially private extensions of some common rank tests \citep{couch2018differentially,couch2019differentially} and some older modifications to non-private rank tests \citep{van1952order, gastwirth1965percentile}. 
In particular, \citet{couch2018differentially,couch2019differentially} introduced private rank tests for testing for difference(s) in location parameters between groups. 
They introduced a technique to ensure privately estimating the sizes of each group does not inflate the type I error of the test, which we adapt to the setting of scale testing. 

Furthermore, our test incorporates rank transformations and percentile modifications, which allows us to increase the power of a ``naive'' private Siegel--Tukey test. 
Percentile modifications were first introduced by \citet{gastwirth1965percentile}, who noticed that extreme ranks determine if there are scale differences between groups. 
The rank transformation is inspired by the Van der Waerden test \citep{van1952order, van1953order}, a nonparametric alternative to ANOVA, which applies the quantile function of the standard normal distribution to the ranks of the data points.

Critically, we show that the transformation can actually reduce the sensitivity of the test statistic, and therefore reduce the additive noise required to ensure differential privacy. 
At a high level, the proposed test works as follows: (i) Combine both samples into one, and rank all points in a center-outward manner. (ii) Modify these ranks by setting central ranks to 0. (iii) Apply a non-negative, increasing transformation $\psi$ to the ranks. (iv) Sum the ranks and privatize the sum, privately estimate the asymptotic variance of the rank sum, and compare the private rank sum statistic to the appropriate null distribution.

Concrete contributions are as follows: (i)
We introduce a new class of nonparametric, robust private scale tests, RPST tests. We present their asymptotic distribution under the null hypothesis, under a very general condition on the rank transformation, see Theorem \ref{thm::main-result}. We also prove RPST tests are differentially private, Theorem~\ref{thm::privacy}, and that their type I error does not exceed the given level, $\alpha$, Lemma~\ref{lemm::no-type-1}. 
Theorem \ref{thm::main-result} uncovers that the growth rate of the transformation presents a tradeoff between power and sensitivity. 
(ii) We do extensive simulations to investigate the effects of the percentile modification and the rank transformation, under varying population distributions, sample sizes and privacy budgets. We show that our proposed test method outperforms tests constructed from general private testing frameworks. This justifies the development of private test specifically designed for scale testing. In addition, the previously mentioned tradeoff between power and sensitivity is verified empirically. (iii) We show that these techniques can also be used to improve the differentially private signed-rank test \citep{couch2019differentially}, see Section~\ref{sec::srtest}. 
(iv) Given that these tests did not exist in the non-private setting, some new rank tests are introduced as a byproduct of our work.


\section{Preliminaries}
\subsection{Problem of study}
We first introduce the problem we study in this work. 
Let $\cM_1(\re)$ be the set of Borel probability distributions on $\reals$. Going forward, we will identify members of $\cM_1(\re)$ with their cumulative distribution functions (CDFs).
We consider two independent random samples from distributions $F\in\cM_1(\re)$ and $G\in\cM_1(\re)$. The first one, which we will call ``group 1'', or $\bfX_{n_1}$, contains the observations $X_1, \dots, X_{n_1} \sim F$. The second one, which we will call ``group 2'', or $\bfY_{n_2}$, contains the observations $Y_1, \dots, Y_{n_2} \sim G$. 
We will further assume that both $F$ and $G$ are members of some \emph{location-scale family} of distributions, and both have the same location parameter. 
That is, there is some ``base'' distribution $F_0\in\cM_1(\re)$ and location $\mu\in\re$, which are fixed and unknown, and from which we can form the family of distributions given by $\cL=\{F_0(\sigma(\cdot-\mu)\colon \sigma\geq 0\}\subset\cM_1(\re)$. We then assume that $F$ and $G$ lie in $\cL$. For example, $\cL$ may be the family of Normal distributions with mean $\mu$ or the family of Cauchy distributions with median $\mu$. 
Our goal is to develop a differentially private test for a difference in scale, $\sigma$, between $F$ and $G$. 
That is, we wish to test the hypothesis
\begin{equation}\label{eqn::hyp}
    H_0: F=G \quad \text{vs.} \quad H_1: \text{there exists } \theta>0, \ \theta \neq 1, \text{ such that }  F(x) = G(\theta x)\ \forall x\in\reals,
\end{equation}
in a differentially private manner. 
Note that \eqref{eqn::hyp} implicitly assumes that $F$ and $G$ have the same location, e.g., mean or median. 
In practice, one can normalize the samples by a private estimate of the mean or median, or if one assumes normality, this can be achieved by differencing within each sample. 
For simplicity, we will also assume that $F$ and $G$ are absolutely continuous distributions. 
Relaxing the assumption of absolute continuity is discussed later in Remark~\ref{rem::ties}. 
These conditions can be summarized as follows:
\begin{condition}\label{cond::problem}
We have that $X_1, \dots, X_{n_1} \sim F$ and $Y_1, \dots, Y_{n_2} \sim G$ are two independent random samples for some $n_1,n_2\in\bbN$. Further, we have that $F,G\in\cL$ and $F,G$ are absolutely continuous. 
\end{condition}
Sometimes it will be convenient to consider the combined sample. In this way, let $n=n_1+n_2$, and we let $Z_1,Z_2,\ldots,Z_{n}$ be the combined sample, where we label the points in non-decreasing order $Z_1\leq\ldots\leq Z_n$. 
\subsection{Siegel--Tukey test}
Now that we understand the testing problem, we review the Siegel--Tukey test, a key ingredient in the proposed class of private scale tests. 
To test the hypothesis \eqref{eqn::hyp} in the nonprivate setting, \citet{siegel1960nonparametric} devised the following procedure, which is known as the Siegel--Tukey test. We first arrange the data points in the combined sample from least to greatest. We then assign ranks to the ordered data points in the following manner:
    \begin{enumerate}
        \item Assign rank 1 to the lowest data point.
        \item Assign rank 2 to the greatest data point, and assign rank 3 to the second-greatest data point.
        \item Assign rank 4 to the second-lowest data point, and assign rank 5 to the third-lowest data point.
        \item Continue alternating between the two highest unranked points and the two lowest unranked points in this manner until every data point has been assigned a rank.
    \end{enumerate}
The (non-private) test statistic is then computed as follows: the minimum of the sum of the ranks in group 1, say $\widehat U_1$ and the sum of the ranks in group 2, say $\widehat U_2$ is first computed. 
The final test statistic is given by $\widehat U=\min(\widehat U_1 -n_1(n_1+1)/2,\widehat U_2-n_2(n_2+1)/2)$, which has the same distribution under the null hypothesis as the Wilcoxon rank-sum statistic.  
Later, we introduce a new test based on these principles, which avoids using the minimum of two statistics. 

\subsection{Percentile modification}
It is helpful to review the percentile modification introduced by \citet{gastwirth1965percentile}. This is a technique used to improve the asymptotic relative efficiency of nonparametric rank-based hypothesis tests. This technique entails ignoring either the most extreme or the most central ranks of a dataset when calculating a test statistic. 
For instance, consider the Siegel--Tukey test. Given that both samples have the same mean, both samples will have some central observations. Therefore, both samples will have some larger (center-outward) ranks, and these will not be helpful in determining if there is a difference in dispersion between the two groups. 
Therefore, given $0<q<1$, we can instead set the $\floor{nq}$ largest ranks to be equal to $n-\floor{nq}$. This reduces unhelpful noise in the set of ranks, which improves the power of the test. 
Indeed, Gastwirth applied this technique to the Ansari-Bradley-Freund Test \citep{ansari1960rank}, another scale tests based on ranks which is very similar to the Siegel--Tukey Test. The result was an increase in the asymptotic relative efficiency and power for certain settings $q$. 
Subsequent works have utilized the technique of percentile modification to improve the power of other hypothesis tests \citep{chenouri2011data, Ramsay2023}. As mentioned previously, we utilize this technique in our proposed test. 

\subsection{Differential privacy}
Next, we introduce the neccessary background on differential privacy \citep{Dwork2006, Dwork2014}. 
Define a dataset of size $n\in \bbN$ to be a set of $n$ real numbers, and let $\cD_n$ be the set of datasets of size $n$. 
First, we say that a dataset $\bfx_n\in \cD_n$, is adjacent to another dataset $\bfy_n\in \cD_n$ if $\bfx_n$ and $\bfy_n$ differ by exactly one point. 
Let $\cA_n$ be the set of pairs of adjacent datasets of size $n$. 
Next, let $H_{\bfx_n}$ be a probability distribution over $\reals$, which depends on $\bfx_n$. 
That is, $H_{.}\colon \cD_n\to \cM_1(\reals)$. 
Assuming that $H_{\bfx_n}$ is absolutely continuous, let $h_{\bfx_n}$ be the associated density. 
We can now define differential privacy.  
\begin{definition}\label{dfn::pdp}
We say that the quantity $\theta\sim H_{\bfx_n}$ is $(\eps,\delta)$-differentially private if for all $(\bfy_n,\bfz_n)\in \cA_n$ and $x\in\reals$, it holds that
\begin{equation*}
h_{\bfy_n}(x)\leq e^\eps h_{\bfz_n}(x)+\delta .
\end{equation*}
\end{definition}
\noindent Here, $\theta$ is a differentially private quantity and the pair $(\eps,\delta)$ represent the privacy budget. 
The parameter $\eps>0$ controls the level of privacy assured to each individual, and $\delta>0$ is the probability of leakage, or, the probability that the privacy condition $h_{\bfy_n}(x)\leq e^\eps h_{\bfz_n}(x)$ fails to hold. 
Typically, $\eps$ is a small constant $\approx 1$ and $\delta=O(n^{-k})$ for some positive integer $k$. Smaller values of each of $\eps$ and $\delta$ enforce stricter privacy guarantees. 
Note that if $\delta = 0$, the mechanism is said to be $\eps$-differentially private. 
The following is a useful property for developing differentially private estimates. 
\begin{proposition}[\textbf{Basic Composition} \citep{dwork2006our}] \label{thm::comp}
    Suppose that, for $i = 1, 2, \dots, k$, the quantity $\theta_i$ is $(\eps_i, \delta_i)$-DP. Let $\theta = (\theta_1, \theta_2, \dots, \theta_k)$ be a sequence of these algorithms. Then, $\theta$ is $\left(\sum_{i = 1}^{k} \eps_i, \sum_{i = 1}^k \delta_i \right)$-differentially private.
\end{proposition} 
This allows us to quantify the privacy loss incurred when composing multiple separate algorithms. 
Next, we introduce the simplest way to generate a differentially private statistic: the Laplace mechanism. 
For a univariate statistic $T\colon\cD_n\to \reals$, define the global sensitivity of $T$ as 
$$\GS(T)=\sup_{(\bfy_n,\bfz_n)\in \cA_n}|T(\bfy_n)-T(\bfz_n)|.$$
The Laplace mechanism is given as follows:
\begin{proposition}[\textbf{Laplace Mechanism} \citep{dwork2006calibrating}]\label{thm::lapmech}
Let $T\colon\cD_n\to \reals$ be any statistic with $\GS(T)<\infty$. 
If $Z$ is a standard Laplace random variable, then $$T(\bfx_n)+\GS(T)Z/\eps$$ is $\eps$-differentially private. 
\end{proposition}
\noindent Proposition~\ref{thm::lapmech} illustrates how correctly calibrated Laplace noise can be added to a statistic to ensure differential privacy. 
\subsection{Unknown group sizes in private rank testing}
One issue in private, two-sample rank testing is that the group sizes must be known to execute the test. 
That is, the asymptotic distribution of the test statistic under the null hypothesis depends on the group sizes. 
However, the group sizes can potentially be sensitive information. 
Therefore, we cannot use them directly to obtain the asymptotic reference distribution or to compute the test statistic. 
In addition, it is not ideal to plug in a simple private estimate of the group sizes, as this can inflate the probability of type I error. 

\citet{couch2019differentially} considered this problem when creating their private Mann-Whitney test, where it is essential to know the minimum group size. 
They estimate the size of smallest group, $m$, such that with probability $1 - \delta$, the estimate of $m$, say $m^*$, is guaranteed to be less than $m$. 
This ensures that the type I error is bounded by $\alpha$, with probability $1-\delta$. 
However, this comes at the cost of pure differential privacy. 
Given $\eps_m>0$, the estimate is defined as follows $$m^*=\max \left\{\left\lceil\min\{n_1, n_2\}+ \frac{1}{\eps_m}Z + \frac{\ln(2\delta)}{\eps_m}\right\rceil, 0 \right\},$$
where $Z$ is a standard Laplace random variable. 
\citet{couch2019differentially} show that this algorithm is $(\eps_m, \delta)$-differentially private and that $m^* < m$ with probability $1 - \delta$. We generalize this approach to the private version of the Siegel--Tukey test. 



\section{Private Siegel--Tukey tests}\label{sec::methods}
\subsection{Definition} \label{sub::definition}
With the necessary background in hand, we can now introduce the RPST tests. 
Briefly, the tests proceed as follows: (i) We first rank the data similarly to that of the original Siegel--Tukey test. (ii) We then apply a percentile modification and a transformation $\psi$ to the ranks. (iii) We then compute a private rank-sum statistic based on these modified, transformed ranks. (iv) We conservatively and privately estimate the asymptotic variance of the private rank-sum statistic, which we use together with a central limit theorem to determine the $p$-value of the test. 
The last step relies on a private estimate of $|n/2-n_1|$, after which we adapt the techniques of \citet{couch2019differentially} previously mentioned.

In order to adapt the percentile modification to be used with the Siegel--Tukey test, we must modify the traditional Siegel--Tukey test statistic. 
In the original Siegel--Tukey test, under the null hypothesis, we have that $\widehat U_1-n_1(n_1+1)/2$ and $\widehat U_2-n_2(n_2+1)/2$ have the same distribution. 
If we were to naively replace the ranks in the calculation of the Siegel--Tukey test statistic with the rank-transformed, percentile-modified ranks, then this property would no longer hold in general. 
That is, the quantity associated with group 1 would no longer be equal in distribution to the quantity in group 2 under the null hypothesis. 

As a result, we propose the following, alternative ranking procedure.
As in the traditional percentile-modified tests, let $q$ be the proportion of central data points we wish to set to 0, and let $Q = \floor{nq}$ be then the number of points to set to 0. 
We propose a new way to assign the ranks, which is from the extremes inward, starting at $n-Q$. That is, the lowest data point is assigned rank $n - Q$, the highest and second-highest point are assigned ranks $n - Q - 1$ and $n - Q - 2$ respectively, and so on until the positive ranks have been exhausted. The remaining unranked points, which are the $Q$ most central points, are assigned rank 0. 
Denote these ranks by $r_1,\ldots,r_n$, and without loss of generality, take the first $r_1,\ldots,r_{n_1}$ to be the ranks pertaining to group 1. 
After computing these ranks, we apply an increasing, non-negative function $\psi: \reals^+ \rightarrow \left[0, \infty \right)$ to the ranks, where we require that $\psi(0)=0$. 
Here, $\psi$ is a \emph{rank transformation}. We have found that rank transformations alter the asymptotic distribution and global sensitivity of the test statistic. Therefore, in some cases, it can increase the power of the test. An example of our proposed ranking procedure is given in Figure~\ref{fig::example}. 

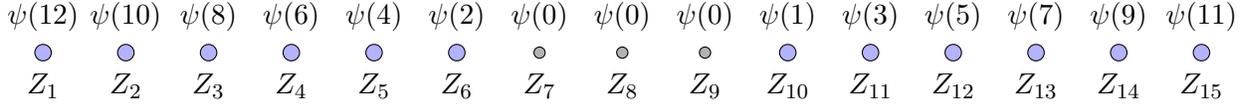
\begin{figure}[t]
    \centering
\begin{tikzpicture}[font=\small]
  \def\unit{1.1cm}

  \foreach \i/\psir in {
    1/12, 2/10, 3/8, 4/6, 5/4, 6/2,
    10/1,11/3,12/5,13/7,14/9,15/11
  }{
    \node[draw, fill=blue!30, circle, inner sep=2.2pt]
      at (\i*\unit,0)
      {};
    \node[below=4pt] at (\i*\unit,0) {\small $Z_{\i}$};
    \node[above=4pt] at (\i*\unit,0) {\small $\psi(\psir)$};
  }

  \foreach \i/\psir in {7/0, 8/0, 9/0}{
    \node[draw, fill=gray!60, circle, inner sep=1.5pt]
      at (\i*\unit,0)
      {};
    \node[below=4pt] at (\i*\unit,0) {\small $Z_{\i}$};
    \node[above=4pt] at (\i*\unit,0) {\small $\psi(\psir)$};
  }

\end{tikzpicture}
    \caption{An example of the ranking procedure for the RPST test, applied to a sample of size 15 with $q=0.2$. The ordered, combined sample $Z_1,\ldots, Z_n$ is pictured on the bottom. The assigned modified, transformed ranks for each of $Z_1,\ldots, Z_n$ are pictured on the top.   }
    \label{fig::example}
\end{figure}

Let $\delta_i=1$ if $i\in [n_1+n_2]$ pertains to an observation in group 1, and 0 otherwise. We can now define the non-private, percentile modified, rank transformed Siegel--Tukey test statistic.
\begin{definition}
Given $\psi: \reals^+ \rightarrow \left[0, \infty \right)$ such that $\psi(0)=0$ and $Q\in \{0,\ldots,n\}$, the non-private, percentile modified, rank transformed Siegel--Tukey test statistic is given by
$$U_1\coloneqq U_1(Q,\psi) = \sum_{i=1}^{n-Q}\delta_i\psi(i) - \frac{n_1}{n}  \sum_{i = 1}^{n - Q} \psi(i)\coloneqq \sum_{i=1}^{n_1}\psi(r_i)-\mu_1.$$ 
\end{definition}
The test statistic is simply the sum of the ranks in group 1, minus the expected value of the sum under the null hypothesis.
\begin{remark}
We lose no generality when we add the restriction $\psi(0)=0$. 
To see this, first for any real function $h$, define $U_1'(h)=\sum_{i=1}^{n_1}h(r_i)-E\left(\sum_{i=1}^{n_1}h(r_i)\right).$ Next, note that $U_1'(Q,h)=U_1'(Q,h+b)$ for some $b\in \reals$. 
Here, for a function $h$ and $b\in \reals$, $b+h$ denotes the function $x\mapsto h(x)+b$. 
\end{remark}
One should note that critically, if we were to define $ U_1$ in terms of group 2, say $U_2$, we have that $U_1=-U_2$, so the test is unaffected by which group is labelled group 1. 
To privatize the above test, we employ the Laplace mechanism after bounding the sensitivity of $U_1$. 
Let $\bar\psi_Q=n^{-1}\sum_{i=1}^{n-Q}\psi(i)$. 
Then, the private test statistic can be obtained as follows:
\begin{definition}
    Given $\psi: \reals^+ \rightarrow \left[0, \infty \right)$ such that $\psi(0)=0$, $\eps_U>0$ and $Q\in \{0,\ldots,n\}$, the private percentile modified, rank transformed Siegel--Tukey test statistic is given by \begin{align*}
    |\widetilde{U}_1|\coloneqq |\widetilde U_1(Q,\psi)| = |U_1 + Z\cdot    
        \GS^*(U_1)/\eps_U|,
\end{align*}
where $Z$ is a standard Laplace random variable and $$\GS(U_1)\leq \GS^*(U_1)=\max \left\{\psi(n - Q), \psi(n - Q) + \psi(n - Q - 1) - \bar\psi_Q \right\} .$$
\end{definition}
The proof of $\GS(U_1)\leq \GS^*(U_1)$ is given in the appendix, see Lemma~\ref{lem::rtpmstsens}. This implies that $\widetilde U_1$ is $\eps_U$-differentially private.

Let 
\begin{equation*}
   \sigma^2(n_1,n_2,\psi,Q) =\frac{n_1}{n}\left(1-\frac{n_1}{n}\right) \sum_{i = 1}^{n - Q} \psi^2(i)   + 2\frac{n_1}{n}\left(\frac{n_1 - 1}{n - 1}-\frac{n_1}{n}\right) \sum_{i = 1}^{n - Q - 1} \sum_{j = i + 1}^{n - Q} \psi(i) \psi(j)
    .
\end{equation*}
We show later, see Theorem~\ref{thm::main-result}, that under the null hypothesis, for many choices of $\psi$, we have that
\begin{align*}
    \widetilde U_1 \cond \cN\left(0,  \sigma^2(n_1,n_2,\psi,Q) \right)\text{.}
\end{align*}
Here, $\cond$ denotes convergence in distribution and $\cN(\mu,\sigma^2)$ denotes the normal distribution with mean $\mu$ and standard deviation $\sigma^2$. Therefore, if we knew the group sizes $n_1,n_2$, we could compute a $p$-value from $\cN\left(0,  \sigma^2(n_1,n_2,\psi,Q) \right)$. 




However, the group sizes are often subject to the constraint of differential privacy. Thus, we need a way to privately estimate the $p$-value, such that the probability of type I error does not increase. 
For this test, we reject the null hypothesis when the test statistic is sufficiently large. 
In that case, we must configure the algorithm to \textit{overestimate} the standard deviation of the asymptotic normal distribution of $\widetilde U_1$. This is equivalent to underestimating the disparity between the two group sizes, given by $ d_1 = \left|n_1 - n/2 \right|= \left|n_1/2 - n_2/2 \right|$, see Lemma~\ref{lem::vk}. 
Therefore, inspired by \citet{couch2019differentially}, given $\delta,\eps_d>0$, we can use the following estimate of $d_1$. 
Let $Z'$ be a standard Laplace random variable. 
\begin{enumerate}
    \item Compute $\widetilde{d}_1 \coloneqq d_1 + Z'/\eps_d$ and, subsequently, $d_1^* \coloneqq \max \{\ceil{\widetilde{d}_1 + \frac{\ln(2\delta)}{\eps_d}}, 0 \}$. 
    \item If $n$ is odd and $d_1^* \neq 0$, subtract $\frac{1}{2}$ from $d_1^*$. If $n$ is odd and $d_1^* = 0$, instead take $d_1^* = \frac{1}{2}$.
\end{enumerate}
Letting $\widetilde n_1 = n/2 - d_1^*$ and $\widetilde n_2 = n - \widetilde n_1$, we can then use the normal distribution with variance $\widetilde \sigma^2\coloneqq \sigma^2(\widetilde n_1,\widetilde n_2,\psi,Q)$ as the reference distribution to compute the $p$-value for our test. 
Using a similar argument to the one used to prove Lemma A.7 of \cite{couch2019differentially}, though fixing a very oversight in their proof, we find that $d_1^* \leq d_1$, implying that we are not introducing excess type I error into the procedure by estimating the group sizes. See Lemma~\ref{lemm::no-type-1} for a formal statement.


\subsection{Theory}\label{sec::theory}
We now present some theoretical results concerning the RPST tests. 
The first result shows that the RPST tests are differentially private. 
\begin{theorem} \label{thm::privacy}
For all $n_1,n_2\geq 1$, $Q\in \{0\}\cup [n-1]$, strictly-increasing, non-negative functions $\psi: \nats \rightarrow \reals$, the pair $(\widetilde U_1, \widetilde\sigma^2)$ are jointly $(\eps_U+\eps_d,\delta)$-differentially private. 
\end{theorem}
The proof Theorem~\ref{thm::privacy} is given in Appendix~\ref{sec::proofs}. Theorem~\ref{thm::privacy} guarantees that all versions of the RPST test are differentially private. 
Our next result concerns the asymptotic distribution of the test statistic under the null hypothesis.  
For this result, we need a few conditions. 
\begin{condition}\label{cond::group_sizes}
There exists $\lambda \in (0, 1)$ such that $\frac{n_1}{n} \rightarrow \lambda$ as $n\to \infty$.
\end{condition}
Condition \ref{cond::group_sizes} says that the groups sizes are bounded away from zero, or, one group size is not asymptotically dominating another. This is a relatively standard assumption in rank testing. We now introduce another condition. 
\begin{condition}\label{cond::strong}
Let $\psi\colon \re^+\to \re^+$ be a continuous, strictly increasing, positive, real-valued function on $\re^+$, such that $\psi(n)=O(n^k)$ for any $k>0$ and $\psi(0)=0$. 
\end{condition}
Condition~\ref{cond::strong} gives a sufficient condition on $\psi$ for asymptotic normality of $U_1$ under the null hypothesis. 
Condition~\ref{cond::strong} is satisfied for increasing functions which pass through the origin, and do not grow faster than some polynomial. 
For instance, Condition~\ref{cond::strong} is satisfied for $\psi(x)$ being any of $x,\log(x+1)$ or $x^k$. 
Condition~\ref{cond::strong} can also be weakened, see Condition~\ref{cond::psi} in the appendix. However, the weaker condition is very opaque, whereas Condition~\ref{cond::strong} is very straightforward. 
We now present our main theorem, which gives the distribution of the test statistic under the null hypothesis. 
\begin{theorem}\label{thm::main-result}
Let $Q\in \{0\}\cup [n-1]$ be the number of central data points to truncate. Suppose that Conditions \ref{cond::problem}--\ref{cond::strong} hold, as well as $H_0$ holds, as given in \eqref{eqn::hyp}. 
For all $\eps,\delta>0$, it holds that
\begin{equation*}
    \frac{\widetilde U_1}{\sigma(\widetilde n_1,\widetilde n_2,\psi,Q)} \cond \cN(0, 1) \text{.}
\end{equation*}    
\end{theorem}
The proof of this Theorem~\ref{thm::main-result} is given in Appendix~\ref{sec::proofs}.  
Theorem~\ref{thm::main-result} gives the limiting distribution of the test statistic, which yields that it suffices to use the distribution $\cN(0, \sigma^2(\widetilde n_1,\widetilde n_2,\psi,Q))$ as the reference distribution for computing the $p$-values. 

\begin{remark}\label{rem::ties}[Relaxing the absolute continuity assumption]
If $F$ and $G$ are not absolutely continuous, then the data may contain repeated observations. In that case, a small amount of noise can be added to the data in order to obtain data without repeated observations. In that case, it is obvious that the global sensitivity of the test statistic remains the same. In addition, the asymptotic distribution under the null hypothesis will also remain unchanged, and so the test can be applied without changes. The only hiccup is to ensure that the magnitude of noise is small enough so that the order of the non-repeated observations stays the same, which can be chosen based on domain knowledge. 
\end{remark}

\begin{figure}[t]
  \centering
  \begin{minipage}[b]{0.48\textwidth}
    \centering
    \includegraphics[width=\linewidth]{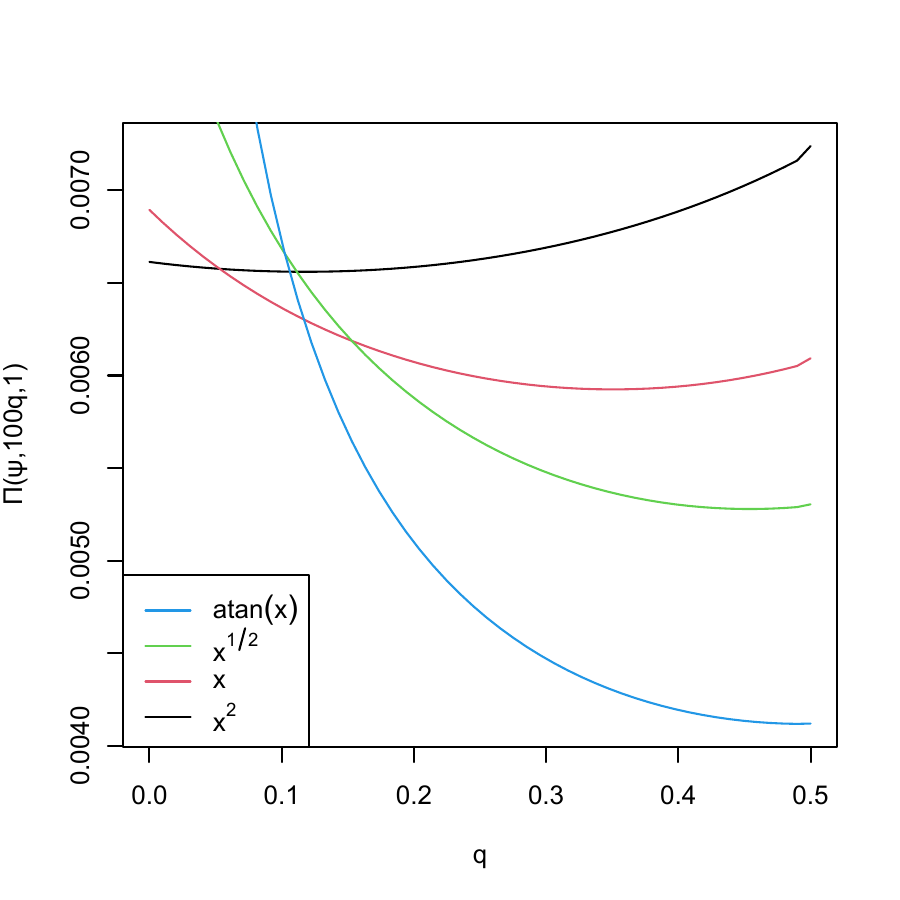}
  \end{minipage}
  \hfill
  \begin{minipage}[b]{0.48\textwidth}
    \centering
    \includegraphics[width=\linewidth]{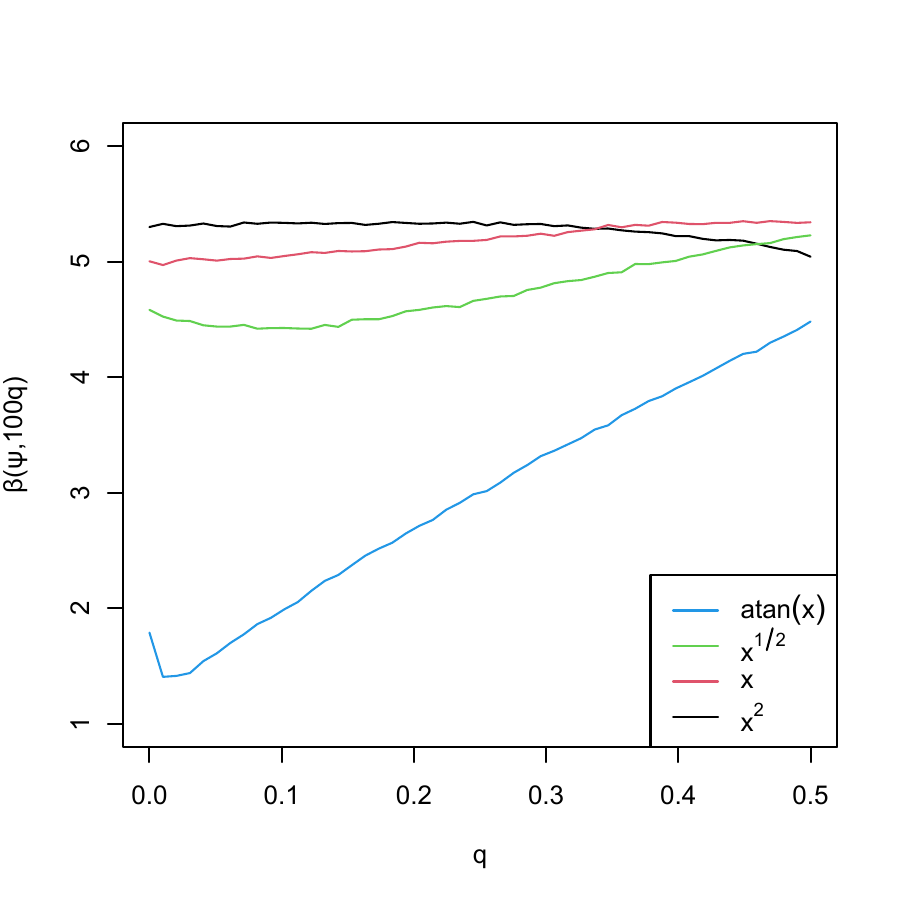}
  \end{minipage}
      \caption{Left: The function $\Pi_n(\psi,Q,1)$ as a function of $q=Q/n$ for different $\psi$ at $n=100$, with $n_1=n_2$. We see that for larger values of $q$ and slow growing $\psi$ is optimal, while the opposite is true for small values of $q$. Overall, we see that larger values of $q$ with slow growing $\psi$ seems to work well. Right: Estimates of $\beta_n(\psi,Q)$ when $F=\cN(0,1)$ and $G=\cN(0,9)$ via Monte Carlo simulation, as a function of $q=Q/n$ for different $\psi$ at $n=100$, with $n_1=n_2$. We see that larger values of $q$ and fast-growing $\psi$ is optimal, though at high values of $q$, the fastest growing $\psi$ is not always best. }
      \label{fig:sens_comparison}
\end{figure}

Theorem~\ref{thm::main-result} and its proof yields some insight into the choices of $\psi$ and $Q$. There is a trade-off between the amount of noise required for privacy is given by $\Pi_n(\psi,Q,\eps_U)={\GS^*(U_1)}/{\sigma(n_1, n_2, \psi, Q)\eps_U }$, and the power of the test, which is governed by the quantity $$\beta_n(\psi,Q)=\frac{\sum_{i=1}^{n_1}E(\psi(r_i))-\frac{n_1}{n}\sum_{i=1}^{n-Q}\psi(i)}{\sigma(n_1, n_2, \psi, Q)}.$$
In general, $\Pi_n(\psi,Q,\eps_U)$ is small for larger values of $Q$ and slow-growing $\psi$, as demonstrated in Figure~\ref{fig:sens_comparison}. 
On the other hand, $\beta_n(\psi,Q)$ is generally larger for fast-growing $\psi$ but can depend on the distribution and the value of $Q$. 
For example, Figure~\ref{fig:sens_comparison} estimates $\beta_{100}(\psi,Q)$ when $F=\cN(0,1)$ and $G=\cN(0,9)$ via Monte Carlo simulation. Here, we see that for large values of $Q$, moderately growing $\psi$ perform well, while very fast and very slow growth $\psi$ are subpar. 
A last remark is that $\Pi_n(\psi,Q,\eps_U)=O\left((n\lambda(1-\lambda)\eps_U  )^{-1/2}\right),$ see Lemma~\ref{lem::privacy-error}. This means that for large $n$, the sensitivity is small compared to $\beta_n(\psi,Q)$. In that case, we may opt for $\psi$ and $q$ which boost the power. We will see in the next section empirical simulations, which allow us to determine when these asymptotics take effect. As a result, we will put forward general recommendations for $\psi,q$.

\section{Simulations}\label{sec::sim}

\begin{table}[t]
\caption{Empirical sizes for the various versions of the test, compared to the test of tests approach when $\epsilon=1$ and the data was normally distributed. The empirical sizes for all tests are close to 5\%, with no notable patterns. Similar phenomena were observed under other distributions and privacy parameters. Values $\leq 5\%$ are bolded.}
\centering
\begin{tabular}{llllll|ll}
\toprule
\multicolumn{8}{c}{$n=100$}\\
\hline
$q$ & $\tan^{-1}(r)$ & $\ln(r + 1)$ & $r^{0.5}$ & $r$ & $r^2$ & ToT $r$ & ToT $\tan^{-1}(r)$\\
\hline
0     & 0.058 & \textbf{0.05}  & \textbf{0.048} & \textbf{0.042} & \textbf{0.05}  & \textbf{0.05}  & -      \\
0.25  & \textbf{0.046} & \textbf{0.048} & 0.072 & \textbf{0.044} & 0.056 & -      & -      \\
0.5   & 0.052 & 0.07  & \textbf{0.038} & 0.066 & \textbf{0.046} & -      & -      \\
0.75  & \textbf{0.046} & 0.052 & \textbf{0.04}  & 0.052 & 0.064 & -      & 0.052  \\
\hline
\multicolumn{8}{c}{$n=500$}\\
\hline
0     & 0.052 & \textbf{0.05}  & 0.052 & 0.062 & 0.056 & \textbf{0.046} & -      \\
0.25  & 0.058 & 0.072 & \textbf{0.038} & 0.07   & 0.07   & -             & -      \\
0.5   & 0.066 & 0.07  & \textbf{0.044} & 0.088 & 0.056 & -             & -      \\
0.75  & 0.06  & 0.056 & \textbf{0.048} & 0.062 & 0.06   & -             & \textbf{0.044} \\
\hline
\multicolumn{8}{c}{$n=1000$}\\
\hline
0     & 0.056 & 0.062 & 0.06  & 0.082 & \textbf{0.04}  & 0.058 & -      \\
0.25  & \textbf{0.05}  & \textbf{0.05}  & \textbf{0.048} & \textbf{0.04}  & 0.06  & -     & -      \\
0.5   & \textbf{0.042} & \textbf{0.042} & \textbf{0.046} & \textbf{0.048} & \textbf{0.04} & - & - \\
0.75  & \textbf{0.048} & 0.052 & 0.066 & \textbf{0.048} & \textbf{0.03}  & -   & 0.082  \\

\bottomrule
\end{tabular}
\label{tab:sizes}
\end{table}

\begin{table}[!t]
\centering
\caption{Empirical power for $\epsilon\in\{0.5,5\}$. Here, the data was normally distributed and the effect sizes were $\theta=2,1.5,$ and 1.25 for $n=100,500$, and 1000, respectfully. The tests with power within 1\% of the highest reported power are highlighted. We see that the RPST tests generally outperform the tests of tests tests, unless $n=100$ and $\epsilon=0.5$. We also see that taking $q$ larger and a slower growing choices of $\psi$, i.e., $\tan^{-1},\ln(\cdot+1)$ works well for higher privacy regimes and taking $\psi$ to be fast growing, i.e., $\psi(r)\geq r$ with $q=0.25$ works well for higher privacy budgets.}
\centering
\begin{tabular}[t]{l|lllll|ll}
\toprule
$q$ & $\tan^{-1}(r)$ & $\ln(r + 1)$ & $r^{0.5}$ & $r$ & $r^2$ & ToT $r$ & ToT $\tan^{-1}(r)$\\
\midrule
[0.3em]
\multicolumn{8}{l}{\textbf{$n=100$, $\epsilon=0.5$}}\\
\hline
\hspace{1em}0 & 0.048 & 0.076 & 0.066 & 0.08 & 0.06 & \textbf{0.122} & -\\
\hspace{1em}0.25 & 0.068 & 0.068 & 0.074 & 0.062 & 0.074 & - & -\\
\hspace{1em}0.5 & 0.102 & 0.092 & 0.082 & 0.064 & 0.07 & - & -\\
\hspace{1em}0.75 & 0.08 & 0.058 & 0.066 & 0.044 & 0.054 & - & \textbf{0.126}\\
\addlinespace[0.3em]
\multicolumn{8}{l}{\textbf{$n=100$, $\epsilon=5$}}\\
\hline
\hspace{1em}0 & 0.252 & 0.882 & 0.9 & 0.872 & 0.766 & 0.664 & -\\
\hspace{1em}0.25 & 0.498 & 0.744 & 0.862 & 0.912 & \textbf{0.942} & - & -\\
\hspace{1em}0.5 & 0.834 & 0.904 & 0.926 & \textbf{0.95} & 0.924 & - & -\\
\hspace{1em}0.75 & \textbf{0.944} & 0.928 & 0.92 & 0.882 & 0.752 & - & 0.734\\
\addlinespace[0.3em]
\multicolumn{8}{l}{\textbf{$n=500$, $\epsilon=0.5$}}\\
\hline
\hspace{1em}0 & 0.062 & 0.114 & 0.226 & 0.352 & 0.326 & 0.246 & -\\
\hspace{1em}0.25 & 0.28 & 0.368 & 0.422 & 0.432 & 0.328 & - & -\\
\hspace{1em}0.5 & \textbf{0.564} & \textbf{0.572} & 0.508 & 0.4 & 0.284 & - & -\\
\hspace{1em}0.75 & 0.532 & 0.414 & 0.364 & 0.224 & 0.15 & - & 0.322\\
\addlinespace[0.3em]
\multicolumn{8}{l}{\textbf{$n=500$, $\epsilon=5$}}\\
\hline
\hspace{1em}0 & 0.534 & \textbf{1} & \textbf{1} & \textbf{1} & 0.984 & 0.62 & -\\
\hspace{1em}0.25 & 0.722 & 0.942 & \textbf{0.998} & \textbf{1} & \textbf{1} & - & -\\
\hspace{1em}0.5 & 0.978 & \textbf{1} & \textbf{1} & \textbf{1} & \textbf{1} & - & -\\
\hspace{1em}0.75 & \textbf{1} & \textbf{1} & \textbf{1} & \textbf{1} & \textbf{1} & - & 0.68\\
\addlinespace[0.3em]
\multicolumn{8}{l}{\textbf{$n=1000$, $\epsilon=0.5$}}\\
\hline
\hspace{1em}0 & 0.05 & 0.14 & 0.32 & 0.378 & 0.374 & 0.13 & -\\
\hspace{1em}0.25 & 0.276 & 0.37 & 0.436 & 0.492 & 0.44 & - & -\\
\hspace{1em}0.5 & 0.528 & \textbf{0.59} & \textbf{0.584} & 0.47 & 0.328 & - & -\\
\hspace{1em}0.75 & \textbf{0.584} & 0.51 & 0.432 & 0.302 & 0.186 & - & 0.17\\
\addlinespace[0.3em]
\multicolumn{8}{l}{\textbf{$n=1000$, $\epsilon=5$}}\\
\hline
\hspace{1em}0 & 0.35 & \textbf{0.994} & \textbf{0.992} & 0.972 & 0.866 & 0.23 & -\\
\hspace{1em}0.25 & 0.53 & 0.788 & 0.94 & 0.978 & \textbf{0.99} & - & -\\
\hspace{1em}0.5 & 0.86 & 0.934 & 0.974 & 0.982 & \textbf{0.99} & - & -\\
\hspace{1em}0.75 & 0.954 & 0.972 & 0.978 & 0.982 & \textbf{0.986} & - & 0.368\\
\bottomrule
\end{tabular}
\label{tab::power}
\end{table}

To evaluate the performance of the proposed test relative to appropriate benchmarks \citep{kazan2023test}, as well as to investigate the effects of the parameters $q,\psi$, we conducted a simulation study. 
All code is made available in the \texttt{RPST} Github Repository. 
We tested sample sizes $n\in\{100, 500, 1000\}$, with $q \in\{ 0, 0.25, 0.5, 0.75\}$. Unless otherwise stated, we assumed that $n_1=n_2$. In Appendix~\ref{app::group}, we considered uneven groups, and found that the power decreases as the groups becoming increasingly uneven. Next, we considered $\psi(n) \in \{\tan^{-1}(n),  \ln(n + 1), n^{0.5} ,n,,n^2\}$. 
We primarily present results from normally distributed data. Results for other distributions, such as Exponential, and Lomax were very similar, and can be seen in Appendix~\ref{app::RPTS}. 
We measure the effect size by $\theta$, which one recalls is the ratio of the scale parameter of $F$ to that of $G$. 
For each combination, we set $\delta = 10^{-6}$ and ran the test on 500 simulated databases.
We considered privacy budgets given by $\eps=\eps_U+\eps_d= 0.5, 1, 5$, where $\eps_U=\eps_d$. 
We also conducted some simulations to test the optimal allocation of the privacy budget. We found that attributing 80\% of the budget to $\eps_U$ performed well, see Appendix~\ref{app::budget}. 

Given that, to our knowledge, there are no direct tests for differentially private scale, we compare our test to the test of tests framework introduced by \citet{kazan2023test} applied to the non-private Siegel--Tukey test. As a sort of post-hoc analysis, given the positive performance of the RPST tests, we also consider the tests of tests framework applied to the non-private version of the RPST test, with $\psi=\tan^{-1}$ and $q=0.75$. We do not compare to \citet{pena2022differentially} as the test of \citet{kazan2023test} has provably higher power \citep{kazan2023test}. 
We present the results with $m=\ceil{n/25}$ and $\alpha_0=0.1$, however, results for various $m$ and $\alpha_0$ can be seen in Appendix~\ref{app:tot}. The parameters $m=\ceil{n/25}$ and $\alpha_0=0.1$ were the best performing by visual inspection, when balancing both empirical size and power. 
If you examine the results for other values of $m$ and $\alpha_0$, it can be seen that our conclusions are not affected by the choice of $m$ and $\alpha_0$. 

We first discuss the empirical size of the tests, followed by the empirical power. 
Empirical sizes for RPST test and the test of tests can be seen in Table~\ref{tab:sizes}. Table~\ref{tab:sizes} shows the empirical sizes under $\epsilon=1$ with Normally distributed data. We defer the other results to Appendix~\ref{app::RPTS}, as similar phenomena were observed under other distributions and privacy parameters. First, for the RPST tests, the empirical sizes don't depend strongly on $\psi$ or $q$. Next, for all tests, the empirical sizes are within 2\% of the theoretical size.

Table~\ref{tab::power} displays the empirical power of the tests for normally distributed data, for $\epsilon=0.5,5$. 
We defer the results from other distributions and for $\epsilon=1$ to Appendix~\ref{app::RPTS}, as the conclusions are essentially identical. 
First, we see that there are some RPST tests that perform substantially better than the ``naive'' test, that is, taking $\psi(r)=r$ and $q=0$. 
For instance, taking $\psi=\tan^{-1}$ with $q=0.5$ or $q=0.75$ performs very well across all parameter settings, but particularly when $\epsilon<5$. For larger privacy budgets ($\epsilon=5$), we have that taking $\psi(r)$ to be $r$ or $r^2$ with $q=0.25$ performs the best. This is consistent with the tradeoff between sensitivity and power discussed in the previous section. 
We conclude that when chosen carefully, the rank transformation and percentile modification have a positive effect on power.

%

Next, comparing to the tests of tests version of the Siegel--Tukey tests, we see that for $n=100$ and $\epsilon=0.5$, the tests of tests tests have a higher power than the RPST tests, regardless of $\psi$ and $q$. However, for other parameter settings, especially for $\epsilon>1$ or $n>100$, we have that the RPST tests for the pairs of $q$ and $\psi$ mentioned previously, and even other versions of the test, outperforms the tests of tests tests. This effect is independent of the distribution, where the results can be seen in Appendix~\ref{app::RPTS}. 
We conclude that the specialized scale test outperforms a naive approach taken through existing, general private testing frameworks.

\section{Extensions to the Wilcoxon signed rank test}\label{sec::srtest}

The rank transformation and percentile modification can be applied to other rank-based testing procedures. As an example, we consider its application to the Wilcoxon Signed Rank Test. The Wilcoxon signed rank test applies to the setting where we observe $n$ pairs of data, say $(X_1,Y_1),\ldots,(X_n,Y_n)$ where $X_1\sim F$ and $Y_1\sim G$, and we want to test now whether
$$H_0\colon E(X_1)= E(Y_1)\ vs.\ H_1\colon E(X_1)\neq E(Y_1).$$
Note that $X_i$ is not necessarily independent of $Y_i$, but for $i\neq j$, $(X_i,Y_i)$ is independent of $(X_j,Y_j)$, and $(X_i,Y_i)\eqd (X_j,Y_j)$. 
In this case, the non-private test statistic is computed as follows. For each pair of observations, we compute the absolute difference $b_i=|Y_i-X_i|$ and the sign $s_i=sign(Y_i-X_i)$. Each observation gets a rank $r_i$, which is simply its linear rank amongst the $b_i$. We can then apply percentile modification and rank transformations to these ranks. The test statistic is then 
$$W_1\coloneqq W_1(\psi,Q) =\sum_{i=1}^n s_i\cdot\psi((r_i-Q) \vee 0 ).$$
Following similar logic as \citet{couch2018differentially}, we have that
\begin{theorem}\label{thm::WWWCRS}
    The global sensitivity of $W_1$ is bounded by $2 \psi(n - Q)$.
\end{theorem}
Therefore, we can define $\widetilde W_1=W_1+2Z\psi(n-Q)/\epsilon_U$ as the test statistic. 
We can call this the private rank transformed percentile modified signed rank test, or private RPSR test. 
We also have the following theorem 
\begin{theorem} \label{thm::rsasymp}
    Let $\psi: \nats \, \cup \, \{0\} \rightarrow \reals$ be a strictly-increasing, non-negative function that is $O(n^z)$ for some $z > 0$.  Under $H_0$, $\widetilde W_1 \cond \cN\left(0, \sum_{i = 1}^{n - Q} \psi^2(i)\right)$.
\end{theorem}
The proofs of both of these theorems can be seen in Appendix~\ref{app::sr_proofs}. 
We note that \citet{couch2018differentially} introduced a private version of $W_1(h,0)$, where $h(x)=x$. We now compare the rank transformed, percentile modified version to theirs via a simulation study. 
We consider the same $\psi,q,\epsilon,n$ combinations as in the previous section. We consider normally and student $t_3$ distributed data, and consider the effect size to be the difference in means between the paired data means. 
As in the previous section, each combination of parameters was run 500 times. 

In terms of the results, the empirical size was acceptable and comparable between all combinations of $\psi$ and $q$, so we defer those results to the appendix. In fact, in the interest of space, we defer many of the results to the appendix. To illustrate our results, we simply present Table~\ref{tab::signed}, which displays the empirical power of the procedure under various parameters with $n=100$. Similar conclusions hold for larger sample sizes, and student $t_3$ distributed data, see Appendix~\ref{app::sr}. Here, we see that for small privacy budgets, there is some gain in using certain transformations and values of $q.$ For example, using $\psi(r)=\tan^{-1}(r)$ with $q=0.25$. On the other hand, there is not much gain for larger privacy budgets. Thus, we conclude if the privacy budget is small, i.e., $\epsilon\leq 1$, then we should use say, $\psi(r)=\tan^{-1}(r)$ with $q=0.25$. Otherwise, the test of \citet{couch2018differentially} or the test with $\psi(r)=r^2$ with $q=0.25$ are acceptable. 

\begin{table}[!t]
\caption{Empirical power for the RPSR test, for Normally distributed data, with an effect size of 0.5. Values within 1\% of the highest reported power are bolded. The pair dependence is modelled with a Gaussian copula. Notice how the RPSR test performs well when $\epsilon$ is small, but the gains are minimal or non-existent for larger privacy budgets.}
\centering
\begin{tabular}{lll|ccccc}
\toprule
$q$ & $n$ & $\epsilon$ & $\psi=\ $ $\tan^{-1}(r)$ & $\ln(r+1)$ & $r^{0.5}$ & $r$ & $r^2$\\
\midrule
\addlinespace[0.3em]
\cellcolor{gray!15}{\hspace{1em}0} & \cellcolor{gray!15}{100} & \cellcolor{gray!15}{0.5} & \cellcolor{gray!15}{\textbf{0.724}} & \cellcolor{gray!15}{\textbf{0.72}} & \cellcolor{gray!15}{0.646} & \cellcolor{gray!15}{0.488} & \cellcolor{gray!15}{0.278}\\
\hspace{1em}0.25 & 100 &  0.5 & \textbf{0.726} & 0.642 & 0.578 & 0.398 & 0.204\\
\cellcolor{gray!15}{\hspace{1em}0.5} & \cellcolor{gray!15}{100} & \cellcolor{gray!15}{0.5} & \cellcolor{gray!15}{0.546} & \cellcolor{gray!15}{0.438} & \cellcolor{gray!15}{0.314} & \cellcolor{gray!15}{0.202} & \cellcolor{gray!15}{0.1}\\
\hspace{1em}0.75 & 100 &  0.5 & 0.202 & 0.152 & 0.122 & 0.094 & 0.06\\
\addlinespace[0.3em]
\hline
\cellcolor{gray!15}{\hspace{1em}0} & \cellcolor{gray!15}{100} & \cellcolor{gray!15}{1} & \cellcolor{gray!15}{0.928} & \cellcolor{gray!15}{\textbf{0.942}} & \cellcolor{gray!15}{0.936} & \cellcolor{gray!15}{0.928} & \cellcolor{gray!15}{0.842}\\
\hspace{1em}0.25 & 100 &  1 & \textbf{0.952} & \textbf{0.944} & 0.926 & 0.872 & 0.712\\
\cellcolor{gray!15}{\hspace{1em}0.5} & \cellcolor{gray!15}{100} & \cellcolor{gray!15}{1} & \cellcolor{gray!15}{0.93} & \cellcolor{gray!15}{0.894} & \cellcolor{gray!15}{0.824} & \cellcolor{gray!15}{0.696} & \cellcolor{gray!15}{0.412}\\
\hspace{1em}0.75 & 100 &  1 & 0.658 & 0.476 & 0.412 & 0.242 & 0.138\\
\addlinespace[0.3em]
\hline
\cellcolor{gray!15}{\hspace{1em}0} & \cellcolor{gray!15}{100} & \cellcolor{gray!15}{5} & \cellcolor{gray!15}{0.968} & \cellcolor{gray!15}{\textbf{0.986}} & \cellcolor{gray!15}{\textbf{0.992}} & \cellcolor{gray!15}{\textbf{0.996}} & \cellcolor{gray!15}{\textbf{0.994}}\\
\hspace{1em}0.25 & 100 &  5 & 0.984 & \textbf{0.992} & \textbf{0.996} & \textbf{0.998} & \textbf{0.992}\\
\cellcolor{gray!15}{\hspace{1em}0.5} & \cellcolor{gray!15}{100} & \cellcolor{gray!15}{5} & \cellcolor{gray!15}{\textbf{0.996}} & \cellcolor{gray!15}{\textbf{0.994}} & \cellcolor{gray!15}{\textbf{0.994}} & \cellcolor{gray!15}{\textbf{0.99}} & \cellcolor{gray!15}{0.972}\\
\hspace{1em}0.75 & 100 &  5 & 0.966 & 0.958 & 0.958 & 0.928 & 0.846\\
\bottomrule
\end{tabular}
\label{tab::signed}
\end{table}


\section{Conclusion}
We have investigated a new class of nonparametric, robust rank tests, called percentile modified rank transformed tests. We have introduced and studied differentially private versions of these tests. We have shown that these tests often perform better relative to existing tests. Some open questions remain, such as controlling the type I error under pure differential privacy, deriving asymptotic results under alternative hypotheses and private scale testing for multivariate data. Another interesting direction of future work is deriving the optimal transformation given a class of distributions. 

\subsubsection*{Acknowledgments}
The authors acknowledge the support of the Natural Sciences and Engineering Research Council of Canada (NSERC). Cette recherche a \'et\'e financ\'ee par le Conseil de recherches en sciences naturelles et en g\'enie du Canada (CRSNG),  [DGECR-2023-00311].

\newpage
\bibliographystyle{apalike}
\bibliography{main}
\newpage
\appendix

\section{Proofs}\label{sec::proofs}
\subsection{A useful structural theorem}
We begin by stating a useful structural theorem. 

\begin{theorem} \label{thm::sk}
    Let $x_1, x_2, \dots, x_n \in \reals$ (these values need not be distinct). Let $S_k$ be a random variable equal to the sum of $k$ values drawn from $(x_1, \dots, x_n)$ in a simple random sample without replacement. Then, we have that:
    \begin{enumerate}
        \item $\expect(S_k) = \frac{k}{n} \sum_{i = 1}^{n} x_i$. \label{thm::sk1}
        \item $\var(S_k) = \frac{k}{n} \sum_{i = 1}^{n} x_i^2 + 2 \frac{k(k - 1)}{n(n - 1)} \sum_{i = 1}^{n} \sum_{j = i + 1}^{n} x_i x_j - \left(\frac{k}{n} \sum_{i = 1}^{n} x_i\right)^2$. \label{thm::sk2}
    \end{enumerate}
\end{theorem}

\begin{proof}
    We begin with assertion \ref{thm::sk1}. Let $A = (A_1, A_2, \dots, A_k)$ be a random sample from $(x_1, \dots, x_n)$. For each $i \in \{1, 2, \dots, n\}$, let $\delta_i = \mathbf{1}\{x_i \in A\}$. Then,

    \begin{align*}
        \expect(S_k) &= \expect \left(\sum_{i = 1}^k A_i \right) \\
        &= \expect \left(\sum_{i = 1}^n x_i \delta_i \right) \\
        &\text{(As $A_i \in \{x_1, \dots, x_n\}$ for all $i \in \{1, 2, \dots, k\}$)} \\
        &= \sum_{i = 1}^n \expect(x_i \delta_i) \\
        &\stackrel{(!)}{=} \sum_{i = 1}^n x_i \, \frac{k}{n} \\
        &\text{(Due to the properties of simple random samples without replacement)} \\
        &= \frac{k}{n} \sum_{i = 1}^{n} x_i.
    \end{align*}

    This completes the proof of assertion \ref{thm::sk1}.

    To prove assertion \ref{thm::sk2}, it suffices to show that

    \begin{align*}
        \expect(S_k^2) = \frac{k}{n} \sum_{i = 1}^{n} x_i^2 + 2 \frac{k(k - 1)}{n(n - 1)} \sum_{i = 1}^{n} \sum_{j = i + 1}^{n} x_i x_j \text{.}
    \end{align*}

    We may write:

    \begin{align*}
        \expect(S_k^2) &= \expect \left[ \left(\sum_{i = 1}^k A_i \right)^2 \right]= \expect \left(\sum_{i = 1}^k A_i^2 + 2 \sum_{i = 1}^{k - 1} \sum_{j = i + 1}^{k} A_i A_j\right).
    \end{align*}

    To make further progress, we introduce the indicators $\delta_{ij}$, where

    \begin{align*}
        \delta_{ij} = \begin{cases}
            1 \quad &\text{if $\{x_i, x_j\} \subseteq A$} \\
            0 \quad &\text{otherwise}
        \end{cases}\cdot
    \end{align*}

    We then have that:

    \begin{align*}
        \expect(S_k^2) &= \expect \left(\sum_{i = 1}^k A_i^2\right) + 2 \,\expect \left(\sum_{i = 1}^{k - 1} \sum_{j = i + 1}^{k} A_i A_j\right) \\
        &= \expect \left(\sum_{i = 1}^n x_i^2 \delta_i \right) + 2 \,\expect \left(\sum_{i = 1}^{n - 1} \sum_{j = i + 1}^{n} x_i x_j \delta_{ij} \right) \\
        &= \sum_{i = 1}^n x_i^2 \, \expect(\delta_i) + 2 \, \sum_{i = 1}^{n - 1} \sum_{j = i + 1}^{n} x_i x_j \expect(\delta_{ij}) \\
        &= \sum_{i = 1}^n x_i^2 \, \frac{\binom{n - 1}{k - 1}}{\binom{n}{k}} + 2 \, \sum_{i = 1}^{n - 1} \sum_{j = i + 1}^{n} x_i x_j \frac{\binom{n - 2}{k - 2}}{\binom{n}{k}} \\
        &\begin{array}{r l}&\text{(As the expectation of an indicator of an event} \\ 
        &\text{is the probability of that event occurring)} \\
        \end{array} \\
        &= \frac{k}{n} \sum_{i = 1}^{n} x_i^2 + 2 \frac{k(k - 1)}{n(n - 1)} \sum_{i = 1}^{n} \sum_{j = i + 1}^{n} x_i x_j \text{.}
    \end{align*}

    Assertion \ref{thm::sk2} then follows from the fact that $\var({X}) = \expect(X^2) - \expect(X)^2$ for any random variable $X$.  
\end{proof}
\subsection{Results justifying the private group sizes}
The following two results ensure the validity (but not the privacy) of the algorithms we use to estimate the group sizes in the RPST tests.  
\begin{lemma}\label{thm::eqvar}
    Let $S_k$ be defined as in Theorem~\ref{thm::sk}. We have that $\var(S_k) = \var(S_{n - k})$.
\end{lemma}
\begin{proof}
    Let $A = (a_1, \dots, a_k)$ be a particular simple random sample without replacement from $(x_1, x_2, \dots, x_n)$. Consider the set of elements $B = (b_1, \dots, b_{n - k})$ which are not present in $A$. $A$ and $B$ are related in the following two ways:

    \begin{itemize}
        \item Since $\binom{n}{k} = \binom{n}{n -k}$, the number of possible samples of size $k$ is equal to the number of possible samples of size $n - k$. Therefore, the probability of obtaining sample $A$ from $k$ draws without replacement is equal to the probability of obtaining sample $B$ from $n - k$ draws without replacement.
        \item Since $A \cup B = \{x_1, x_2, \dots, x_n\}$, we must have that $\sum_{i = 1}^{k} a_i = \sum_{i = 1}^{n} x_i - \sum_{i = 1}^{n - k} b_i$.
    \end{itemize}

    Taken together, and noting that $A$ is an arbitrary sample of size $k$, we have that $S_k \stackrel{d}{=} \sum_{i = 1}^{n} x_i - S_{n - k}$. The desired result then follows from the fact that neither the addition of a constant nor negation affects the variance of a random variable.
\end{proof}
Next, we show that $\var(S_k)$ is a decreasing function of $|k - \frac{n}{2}|$, which justifies our method of computing the $p$-value. 
\begin{lemma} \label{lem::vk}
    Let $x_1, \dots, x_n$ and $S_k$ be defined as in Theorem~\ref{thm::sk}. Suppose that $x_h \geq 0$ for all $h \in \{1, 2, \dots, n\}$ and that there exists $i, j \in \{1, 2, \dots, n\}$ such that $x_i \neq x_j$. Then, $V(k) = \var(S_k)$ is decreasing in $|k - \frac{n}{2}|$.
\end{lemma}
\begin{proof}
   Applying Theorem \ref{thm::sk}, we have that:

    \begin{align*}
        V(k) &= \frac{k}{n} \sum_{i = 1}^{n - 1} x_i^2 + 2 \frac{k(k - 1)}{n(n - 1)} \sum_{i = 1}^{n} \sum_{j = i + 1}^{n} x_i x_j - \left(\frac{k}{n} \sum_{i = 1}^{n} x_i\right)^2 \\
        &\propto k \sum_{i = 1}^{n - 1} x_i^2 + 2 \frac{k^2 - k}{n - 1} \sum_{i = 1}^{n} \sum_{j = i + 1}^{n} x_i x_j - \frac{\left(k \sum_{i = 1}^{n} x_i\right)^2}{n}.
    \end{align*}

    We will extend the domain of $V$ to $\left[0, n\right]$ and take the first two derivatives with respect to $k$. Starting with the first derivative:

    \begin{align*}
        V'(k) &\propto \sum_{i = 1}^{n - 1} x_i^2 + 2 \frac{2k - 1}{n - 1} \sum_{i = 1}^{n} \sum_{j = i + 1}^{n} x_i x_j - \frac{2k}{n} \left(\sum_{i = 1}^{n} x_i\right)^2 \\
        &= \sum_{i = 1}^{n - 1} x_i^2 + 2 \frac{2k - 1}{n - 1} \sum_{i = 1}^{n} \sum_{j = i + 1}^{n} x_i x_j - \frac{2k}{n} \sum_{i = 1}^{n} x_i^2 - \frac{4k}{n} \sum_{i = 1}^{n} \sum_{j = i + 1}^{n} x_i x_j.
    \end{align*}

    It can easily be verified that $V'(\frac{n}{2}) = 0$, indicating the possible existence of a maximum or minimum there. We investigate further by taking the second derivative of $V$ with respect of $k$, with the aim of showing that it is negative:

    \begin{align*}
        V''(k) &\propto \frac{4}{n - 1} \sum_{i = 1}^{n} \sum_{j = i + 1}^{n} x_i x_j - \frac{2}{n} \left(\sum_{i = 1}^{n} x_i\right)^2 \\
        &= \frac{4}{n - 1} \sum_{i = 1}^{n} \sum_{j = i + 1}^{n} x_i x_j - \frac{2}{n} \sum_{i = 1}^{n} x_i^2 - \frac{4}{n} \sum_{i = 1}^{n} \sum_{j = i + 1}^{n} x_i x_j \\
        &= \frac{4}{n^2 - n} \sum_{i = 1}^{n} \sum_{j = i + 1}^{n} x_i x_j - \frac{2}{n} \sum_{i = 1}^{n} x_i^2 \\
        &\propto \frac{2}{n - 1} \sum_{i = 1}^{n} \sum_{j = i + 1}^{n} x_i x_j - \sum_{i = 1}^{n} x_i^2 \\
        &= \frac{2}{n - 1} \sum_{i = 1}^{n} \sum_{j = i + 1}^{n} x_i x_j - \sum_{i = 1}^{n} x_i^2  + \frac{1}{n - 1} \sum_{i = 1}^n x_i^2 - \frac{1}{n - 1} \sum_{i = 1}^n x_i^2\\
        &= \frac{1}{n - 1}\left(\sum_{i = 1}^n x_i\right)^2 - \frac{n}{n - 1} \sum_{i = 1}^n x_i^2 \\
        &\propto \left(\sum_{i = 1}^n x_i\right)^2 - n \sum_{i = 1}^n x_i^2 .
    \end{align*}
    To show that the above expression is negative under our hypotheses, we employ the Cauchy--Schwarz inequality. 
    We obtain:
    \begin{align*}
        \left|\sum_{i = 1}^n x_i \right| &\leq \sqrt{\sum_{i = 1}^n x_i^2} \, \sqrt{\sum_{i = 1}^n 1^2} \iff \left|\sum_{i = 1}^n x_i \right| &\leq \sqrt{\sum_{i = 1}^n x_i^2} \, \sqrt{n}   \iff \left(\sum_{i = 1}^n x_i \right)^2 &\leq n \, \sum_{i = 1}^n x_i^2.
    \end{align*}

    Therefore, the aforementioned expression is non-positive, and is equal to 0 if and only if $x_1 = x_2 = \dots = x_n$. This cannot be the case under our assumptions, and so $V''$ is negative. Since $V''$ is also a constant, we may draw the following conclusions:

    \begin{enumerate}
        \item $V$ attains a maximum at $k = \frac{n}{2}$.
        \item $V$ is decreasing in $\left|k - \frac{n}{2}\right|$.
    \end{enumerate}

    This is what we wanted to show.
\end{proof}

\subsection{The privacy of the test}
We now, prove the privacy of the test, which follows mainly from the following lemma: 
\begin{lemma} \label{lem::rtpmstsens}
It holds that for all $Q\in \{0,\ldots,n\}$, and all increasing, non-negative real functions $\psi$, 
\begin{align*}
      \GS(U_1)\leq   \max \left\{\psi(n - Q), \psi(n - Q) + \psi(n - Q - 1) - \bar\psi_Q  \right\}.
    \end{align*}
\end{lemma}
\noindent We will prove this lemma following the proof of Theorem~\ref{thm::privacy}.

\begin{proof}[Proof of Theorem~\ref{thm::privacy}]
From Lemma~\ref{lem::rtpmstsens} and Theorem~\ref{thm::lapmech}, we have that $\widetilde U_1$ is $\eps_U$-differentially private. Furthermore, the properties of the Laplace distribution directly yield that $d_1^*$ is $(\eps_d,\delta)$-differentially private, see also Theorem A.6 of \cite{couch2019differentially}, with the oversight corrected as given in Appendix~\ref{app::Correction}. 
Lastly, basic composition (Theorem~\ref{thm::comp}) gives the final result. 
\end{proof}
In order to prove Lemma~\ref{lem::rtpmstsens}, we first prove a simpler theorem, a bound on the global sensitivity of $U_1$ when $Q=0$. 
\begin{lemma}\label{thm::rtstsens}
It holds that
    \begin{align*}
       \GS(U_1(\psi,0))\leq  \max \left\{\psi(n) - \psi(1), \psi(n) + \frac{1}{n} \sum_{i = 1}^n  \psi(i) - \psi(1) - \psi(2)\right\} \text{.}
    \end{align*}
\end{lemma}
\begin{proof} 
    The following proof is similar in spirit to that of Theorem 4.1 of \cite{couch2019differentially} in form. However, since ranks are assigned differently, we will require new arguments. 

    Consider two neighboring databases $\mathbf{x}$ and $\mathbf{x'}$.
    Arrange the data points in database $\mathbf{x}$ as follows: $X_1, X_2, \dots, X_{n_1}, Y_1, Y_2, \dots, Y_{n_2}$. Label these arranged data points $Z_1, Z_2, \dots, Z_n$, in order from least to greatest. Denote by $w_i$ the ``raw rank'' of the $i$th data point in this sequence, $i \in \{1, 2, \dots, n_1 + n_2 = n\}$. From there, let $r(w, n)$ be the map from the ``raw ranks'' to the ``working ranks'' (i.e., the ranks used in the rank-transformed Siegel--Tukey testing procedure), and let $r_i$ denote the ``working rank'' of data point $i$. Do the same for database $\mathbf{x'}$ and denote the corresponding quantities with the same notation, except with an apostrophe (ex., the ``working ranks'' of database $x'$ will be denoted $w'_i$). 

    Suppose the data point with $w_i = p$ from database $\mathbf{x}$ is changed to become the data point with $w'_i = q$ in database $\mathbf{x'}$. Without loss of generality, we may assume that $p \leq q$. Divide the rest of the data points in database $\mathbf{x}$ into four regions as follows, depending on the case. Let $h$ be the raw rank of the data point with the highest working rank (this will depend on $n$). Note that $r(h, n) = \psi(n)$. When $p \leq h$ and $q > h$, the regions are:

    \begin{itemize}
        \item Region $A_1$, consisting of the data points with $w_i < p$;
        \item Region $A_2$, consisting of the data points with $p < w_i \leq h$;
        \item Region $A_3$, consisting of the data points with $h < w_i < q$; and
        \item Region $A_4$, consisting of the data points with $w_i > q$.
    \end{itemize}

    When $p, q \leq h$ or $p, q > h$, the regions are:

    \begin{itemize}
        \item Region $A_1$, consisting of the data points with $w_i < p$;
        \item Region $A_2$, consisting of the data points with $p < w_i < q$; and
        \item Region $A_3$, consisting of the data points with $w_i > q$.
    \end{itemize}

    Now define:

    \begin{align*}
        b_j = \sum_{i: Z_i \in \bfX_{n_1}, w_i \in A_j} (r_i - r'_i),
    \end{align*}
    where $\bfX_{n_1}$ consists of all the data points in group 1 (i.e., the $X$ points). Armed with this notation, we are now ready to conduct the case analysis.

    \hfill\newline
    \noindent \textbf{Case 1: The Changed Row Retains its Group}
    \paragraph{Case 1a: $p, q \leq h$}
    Suppose without loss of generality that the changed row is from group 1. Let us first bound the sensitivity of $U_1$. As in \cite{couch2019differentially}, we decompose the sensitivity into the sum of three effects:
    \begin{enumerate}
        \item The removal of the row with raw rank $p$ from database $\mathbf{x}$, which (by itself) decreases $U_1$ by $r(p, n)$. 
        \item The addition of the row (whose raw rank becomes $q$) into database $\mathbf{x'}$, which increases $U_1$ by $r(q, n)$.
        \item The reduction of the raw ranks of the data points in region $A_2$ of database $\mathbf{x}$ by one apiece, which decreases $U_1$ by $|b_2|$ (as working ranks decrease away from the centre).
    \end{enumerate}
    This means that:
    \begin{align*}
        \GS(U_1) \leq |-r(p, n) + r(q, n) - |b_2||.
    \end{align*}
    Since
    \begin{align*}
        r(1, n) = \psi(1) \leq ||b_2| + r(p, n)| \leq r(h, n) - r(1, n) = \psi(n) - \psi(1),
    \end{align*}
    and
    \begin{align*}
        \psi(1) \leq r(q, n) \leq \psi(n) \text{,}
    \end{align*}
    we have that:
    \begin{align*}
        \GS(U_1) \leq \psi(n) - \psi(1) \text{.}
    \end{align*}
    By a similar argument, if it were the case that $p, q > h$, one could obtain the same sensitivity bound. 
    \paragraph{Case 1b: $p \leq h$ and $q > h$}
    Suppose without loss of generality that the changed row is from group 1. Let us first bound the sensitivity of $U_1$. As in \cite{couch2019differentially}, we decompose the sensitivity into the sum of four effects:
    \begin{enumerate}
        \item The removal of the row with raw rank $p$ from database $\mathbf{x}$, which (by itself) decreases $U_1$ by $r(p, n)$. 
        \item The addition of the row (whose raw rank becomes $q$) into database $\mathbf{x'}$, which increases $U_1$ by $r(q, n)$. 
        \item The reduction of the raw ranks of the data points in region $A_3$ of database $\mathbf{x}$ by one apiece, which increases $U_1$ by $|b_3|$ (as working ranks increase toward the centre).
        \item The reduction of the raw ranks of the data points in region $A_2$ of database $\mathbf{x}$ by one apiece, which decreases $U_1$ by $|b_2|$ (as working ranks decrease away from the centre).
    \end{enumerate}
    This means that:
    \begin{align*}
        \GS(U_1) \leq |-r(p, n) + r(q, n) + |b_3| - |b_2||.
    \end{align*}
    Since
    \begin{align*}
        0 \leq ||b_2 + r(p, n)| \leq r(h, n) - r(1, n) = \psi(n) - \psi(1),
    \end{align*}
    and
    \begin{align*}
        0 \leq ||b_3 + r(q, n)| \leq r(h, n) - r(n, n) = \psi(n) - \psi(2) \text{,}
    \end{align*}
    we have that:
    \begin{align*}
        \GS(U_1) \leq \psi(n) - \psi(1) \text{.}
    \end{align*}
    This is the same bound as was achieved in Case 1a. 
    
    \hfill\newline
    \textbf{Case 2: The Changed Row Switches its Group}
    \paragraph{Case 2a: $p, q \leq h$}
    Suppose, without loss of generality, that the changed data point transfers from group 1 to group 2, and that $p \leq q$. As before, we decompose the sensitivity into the sum of three effects:
    \begin{enumerate}
        \item The removal of the row with raw rank $p$ from database $\mathbf{x}$, which (by itself) decreases $U_1$ by $r(p, n)$.
        \item The reduction of the raw ranks of the data points in region $A_2$ of database $\mathbf{x}$ by one apiece, which decreases $U_1$ by $|b_2|$ (as working ranks decrease away from the centre).
        \item The removal of a row from group 1, which increases $U_1$ by $\frac{1}{n} \sum_{i = 1}^n \psi(i)$.
    \end{enumerate}
    This means that:
    \begin{align*}
        \GS(U_1) \leq \left|-r(p, n) - |b_2| + \frac{1}{n} \sum_{i = 1}^n \psi(i)\right|.
    \end{align*}
    Since
    \begin{align*}
        \psi(1) \leq ||b_2| + r(p, n)| \leq r(h, n) - r(1, n) = \psi(n) - \psi(1),
    \end{align*}
    and $\sum_{i = 1}^n \psi(i) \leq \psi(n)$, we have that $\GS(U_1) \leq \psi(n) - \psi(1)$.
    By a similar argument, the same bound is achieved when $p, q > h$.

    \paragraph{Case 2c: $p \leq h$ and $q > h$}
    The assumptions remain the same as in case 2a. The sensitivity of $U_1$ can be written as the sum of the following four effects:
    \begin{enumerate}
        \item The removal of the row with raw rank $p$ from database $\mathbf{x}$, which (by itself) decreases $U_1$ by $r(p, n)$. 
        \item The removal of a row from group 1, which increases $U_1$ by $\frac{1}{n} \sum_{i = 1}^n \psi(i)$.
        \item The reduction of the raw ranks of the data points in region $A_3$ of database $\mathbf{x}$ by one apiece, which increases $U_1$ by $|b_3|$ (as working ranks increase toward the centre).
        \item The reduction of the raw ranks of the data points in region $A_2$ of database $\mathbf{x}$ by one apiece, which decreases $U_1$ by $|b_2|$ (as working ranks decrease away from the centre).
    \end{enumerate}
    This means that:
    \begin{align*}
        \GS(U_1) \leq \left|-r(p, n) + \frac{1}{n} \sum_{i = 1}^n \psi(i) + |b_3| - |b_2|\right|.
    \end{align*}
    To make further progress, we will utilize the technique of bounding the expression within the absolute value in two separate cases: when it is positive and when it is negative. This technique is used by \cite{couch2019differentially} to bound the sensitivity of the Mann-Whitney test statistic. When the expression is positive, we have that:
    \begin{align*}
        \GS(U_1) &\leq \frac{1}{n} \sum_{i = 1}^n \psi(i) + |b_3| - r(p, n) - |b_2| \leq \frac{1}{n} \sum_{i = 1}^n \psi(i) + \psi(n) - \psi(2) - \psi(1),
    \end{align*}
    The last line follows from the fact that, $|b_3| \leq \psi(n) - \psi(2)$ and $r(p, n) \leq \psi(1)$.
    If the aforementioned expression is negative, then:
    \begin{align*}
        \GS(U_1) &\leq -\frac{1}{n} \sum_{i = 1}^n \psi(i) - |b_3| + r(p, n) + |b_2| \leq \psi(n) - \psi(1).
    \end{align*}
    The last line follows from the facts that $r(p, n) + |b_2| \leq \psi(n)$ and $\frac{1}{n} \sum_{i = 1}^n \psi(i) \geq \psi(1)$.
    All told,
    \begin{align*}
        \GS(U_1) = \max \left\{\psi(n) - \psi(1), \psi(n) + \frac{1}{n} \sum_{i = 1}^n  \psi(i) - \psi(1) - \psi(2)\right\} \text{.}
    \end{align*}
    This is what we wanted to show.
\end{proof}

\begin{proof}[Proof of Lemma~\ref{lem::rtpmstsens}]
    The proof is in the spirit of that of Lemma~\ref{thm::rtstsens}, which proves the result for $Q=0$, but with several key differences. In addition to the notation we used in Lemma~\ref{thm::rtstsens}, let $B$ be the lowest ``raw rank'' for which the ``working rank'' of $\psi(0)$ is assigned, and let $T$ be the highest such rank. The differences are outlined below:

    \begin{itemize}
        \item The dataset $\mathbf{x}$ is separated into regions as follows:
        \begin{itemize}
            \item Region $A_1$, consisting of the data points with $w_i < p$;
            \item Region $A_2$, consisting of the data points with $p < w_i \leq B$;
            \item Region $A_3$, consisting of the data points with $B < w_i \leq T$; 
            \item Region $A_4$, consisting of the data points with $T < w_i < q$; and
            \item Region $A_5$, consisting of the data points with $w_i > q$.
        \end{itemize}
        \item For each ``main case'' (i.e., the changed data point switching and not switching groups), there is a new ``sub-case'' to consider: $p < B$ and $B \leq q \leq T$ (there is also the possibility that $p > T$ and $B \leq q \leq T$, but the symmetry would result in the same bound).
        \item The following sensitivity bounds are attained:
        \begin{itemize}
            \item Case 1: The changed row does not switch groups
            \begin{itemize}
                \item $p, q < B$: $\GS(U_1) \leq \psi(n - Q) - \psi(1)$
                \item $p < B \leq q \leq T$: $\GS(U_1) \leq \psi(n - Q)$
                \item $p < B < T < q$: $\GS(U_1) \leq \psi(n - Q) - \psi(1)$.
            \end{itemize}
            \item Case 2: The changed row switches groups
            \begin{itemize}
                \item $p, q < B$: $\GS(U_1) \leq \max \left\{\bar\psi_Q, \psi(n - Q) - \bar\psi_Q\right\}$
                \item $p < B \leq q \leq T$: $\GS(U_1) \leq \psi(n - Q)$
                \item $p < B < T < q$: $\GS(U_1) \leq \max \{\psi(n - Q) + \psi(n - Q - 1) - \bar\psi_Q, \bar\psi_Q\}$.
            \end{itemize}
        \end{itemize}
    \end{itemize}
    The desired result then follows.
\end{proof}


\subsection{The limiting distribution under $H_0$}
We first consider what happens to the ``Laplace'' portion of the test statistic. That is, the term $Z\cdot \GS^*(U_1)/\sigma(n_1, n_2, \psi, Q),$ where $$\GS^*(U_1)=\max \left\{\psi(n - Q), \psi(n - Q) + \psi(n - Q - 1) - \bar\psi_Q  \right\}.$$
\begin{lemma}\label{lem::privacy-error}
Suppose that Conditions \ref{cond::problem}--\ref{cond::group_sizes} hold as well as $H_0$ (as given in \eqref{eqn::hyp}) holds. Then, if $\psi$ is continuous on $\re^+$, then 
    $$ \frac{\GS^*(U_1)}{\sigma(n_1, n_2, \psi, Q)}=O\left(\frac{1}{\sqrt{n\lambda(1-\lambda) }}\right).$$
\end{lemma}
\begin{proof}
We have that 
\begin{align*}
    \sigma^2(n_1,n_2,\psi,Q) &=\frac{n_1}{n}\left(1-\frac{n_1}{n}\right) \sum_{i = 1}^{n - Q} \psi^2(i)   + 2\frac{n_1}{n}\left(\frac{n_1 - 1}{n - 1}-\frac{n_1}{n}\right) \sum_{i = 1}^{n - Q - 1} \sum_{j = i + 1}^{n - Q} \psi(i) \psi(j)
    .
\end{align*}
Now, 
\begin{align*}
        \frac{n_1}{n} - \frac{n_1 - 1}{n - 1}=\frac{n_1(n-1)-n(n_1 - 1)}{n(n - 1)} &= \frac{n - n_1}{n(n - 1)} \stackrel{(!)}{=} \frac{n - \lambda n + O(n^{1/2})}{n(n - 1)} = \frac{1 - \lambda }{n - 1}+ O(n^{-3/2}),
    \end{align*}
Using this, 
\begin{align*}
    \sigma^2(n_1,n_2,\psi,Q) &\approx\lambda(1-\lambda)\left(\left(1+\frac{1}{n-1} \right) \sum_{i = 1}^{n - Q} \psi^2(i) - \frac{1}{n-1} \sum_{i = 1}^{n - Q } \sum_{j = 1}^{n - Q} \psi(i)\psi(j)
    \right)\\
    &=\lambda(1-\lambda)\left(\left(1+\frac{1}{n-1} \right) n\bar\psi_Q^2 - \frac{n^2}{n-1} (\bar\psi_Q)^2
    \right)\\
    &\approx n\lambda(1-\lambda)\left( \bar\psi_Q^2 - (\bar\psi_Q)^2
    \right). 
\end{align*}
In general, we have that 
\begin{align*}
    \sigma^2(n_1,n_2,\psi,Q) &= n\lambda(1-\lambda)\left( \bar\psi_Q^2 - (\bar\psi_Q)^2
    \right)+R, 
\end{align*}
where $R$ is a remainder of lower order. 
    Note that
\begin{equation*}
   \GS^*(U_1)
    \to   \psi(n - Q) + \psi(n - Q - 1) - \bar\psi_Q.
\end{equation*}
Thus, we consider the quantity
$$\frac{\psi(n - Q)}{\sqrt{n\lambda(1-\lambda)\left( \bar\psi_Q^2 - (\bar\psi_Q)^2
    \right)}}.$$
Now, let
\begin{align*}
    C(n,\psi)&=( \bar\psi_Q^2 - (\bar\psi_Q)^2)/\psi(n - Q)^2\\
    &=n^{-1}\sum_{i=1}^{n-Q}\psi(i)^2/\psi(n - Q)^2-\left(n^{-1}\sum_{i=1}^{n-Q}\psi(i)/\psi(n - Q)\right)^2\\
    &\coloneqq V_{n,1}-V_{n,2}^2.
\end{align*}
Then, since $\psi$ is continuous and monotonic, it is differentiable. We have by the mean value theorem there exists $n-Q\leq r\leq n-Q+1$, $n_0\in\bbN$ and $0<p<1$ so that for all $n\geq n_0$,
\begin{align*}
    \psi(n-Q+1)-\psi(n-Q)=\psi'(r)&\implies 1-\frac{\psi(n-Q)}{\psi(n-Q+1)}=\frac{\psi'(r)}{\psi(n-Q+1)}\\
    &\implies \frac{\psi(n-Q)}{\psi(n-Q+1)}\geq 1-\frac{\psi'(r)}{\psi(n-Q+1)}\geq \sqrt{p}. 
\end{align*}
In addition, $a_n=n/(n+1)$ is also greater than $\sqrt{p}$ eventually.
Thus, 
\begin{align*}
    C(n+1,\psi)&=V_{n+1,1}-V_{n+1,2}^2\\
    &=a_n b_n V_{n,1}+(n+1)^{-1}-a_n^2 b_n^2 V_{n,2}^2-2a_n b_n V_{n,2}/n-1/(n+1)^2\\
    &\approx a_n b_n V_{n,1}-a_n^2 b_n^2 V_{n,2}^2\\
    &=a_n b_n C(n,\psi)+(a_n b_n-a_n^2 b_n^2) V_{n,2}^2\\
    &\geq p C(n,\psi) \text{ eventually. }
\end{align*}
We have that 
\begin{align*}
    C(n,\psi)&\geq p\sum_{i=n_0}^{n-1}C(i,\psi)+a_{n_0} b_{n_0} C(n_0,\psi)+(a_{n_0} b_n-a_{n_0}^2 b_{n_0}^2) V_{n_0,2}^2\geq (a_{n_0} b_n-a_{n_0}^2 b_{n_0}^2) V_{n_0,2}^2. 
\end{align*}
Thus, there exists $0<c<1$ and $n_0\in\bbN$ such that $c<C(n,\psi)<1$ for all $n\geq n_0$. 
\end{proof}
Before proving Theorem~\ref{thm::main-result}, we now introduce a condition which is implied by Condition~\ref{cond::strong}. 
For a sequence \seq{y} and given $m, r \in \nats$, define:
\begin{align*}
            \mu_r(y_n, n) = \frac{1}{n} \sum_{i=1}^n \left(y_i - \sum_{j = 1}^n y_j/n \right)^r.
\end{align*}
\begin{condition}\label{cond::psi}
The function $\psi: \re^+ \rightarrow \reals^+$ is a continuous, strictly-increasing, non-negative function that satisfies $\psi(n)=O(n^z)$ for some $z \in \reals$. Furthermore, for all $r > 2$, the sequence $\psi_n\coloneqq\psi_{n,Q}=\psi(1),\psi(2),\ldots$ has the property that
$   \mu_r(\psi_n, n)/\mu_2(\psi_n, n)^{\frac{r}{2}} = O(1).$
\end{condition}
This condition allows us to apply the Theorem~\ref{thm::waldwolf}, which results in asymptotic normality of the non-private test statistic, under the null hypothesis.
\begin{lemma}\label{lem::most-general}
Let $Q\in \{0\}\cup [n-1]$ be the number of central data points to truncate. Suppose that Conditions \ref{cond::problem}, \ref{cond::group_sizes} and \ref{cond::psi} hold, as well as $H_0$ holds, as given in \eqref{eqn::hyp}. 
It holds that 
\begin{equation*}
    \frac{\widetilde U_1}{\sigma(n_1,n_2,\psi,Q)} \cond \cN(0, 1)\text{.}
\end{equation*}       
\end{lemma}
\begin{proof}
    The result follows from Theorem~\ref{thm::waldwolf}, Lemma~\ref{lem::privacy-error} and Slutsky's theorem. 
\end{proof}
\begin{proof}[Proof of Theorem~\ref{thm::main-result}]
In view of Lemma~\ref{lem::most-general}, it suffices to show that Condition~\ref{cond::strong} implies Condition~\ref{cond::psi} and that $\sigma(\widetilde n_1,\widetilde n_2,Q,\psi)/\sigma( n_1, n_2,Q,\psi)\to 1 $ in probability. The latter is trivial from the definition of $\widetilde n_1,\widetilde n_2$ and continuous mapping theorem. To show that Condition~\ref{cond::strong} implies Condition~\ref{cond::psi}, first, consider the case where $Q=0$. 
We bound above the growth rate of $\mu_r(\psi_n, n)$ as follows:
    \begin{align*}
        |\mu_r(\psi_n, n)| &\leq \frac{1}{n}  \sum_{i = 1}^n \left|\psi(i) - \bar\psi_{0,n}\right|^r \leq (2 \psi(n))^r = O(\psi^r(n)).
    \end{align*}
It now suffices to show that $\mu_2(\psi_n, n)$ is $\Omega(\psi^2(n))$, but this was shown in the proof of Lemma~\ref{lem::privacy-error}. 
Now, the case of $Q>0$ is implied by the previous case, since we have already shown that $\mu_2(\psi_n, n)$ is $\Omega(\psi^2(n))$ for all $Q\in\{0\}\cup [n]$ in the proof of Lemma~\ref{lem::privacy-error}. Furthermore, we have that $|\mu_r(\psi_{n,Q}, n)|\leq |\mu_r(\psi_{n,0}, n)|$ for all $Q\in\{0\}\cup [n]$. 
\end{proof}
\section{Proofs concerning the signed rank test}\label{app::sr_proofs}

\begin{proof}[Proof of Theorem~\ref{thm::rsasymp}]
Following similar logic as the proof of Lemma~\ref{lem::privacy-error}, we have that $2\psi(n-Q)\bigg/\sqrt{\sum_{i=1}^{n-Q}\psi(i)^2}\to 0$ as $n\to \infty$. Therefore, it suffices to consider $W_1$.  
    If we can show that $W_1$ satisfies the hypotheses of the \nameref{thm::lya}, the desired result would automatically follow from said theorem. To this end, we shall define, for each $i \in \{1, 2, \dots, Q\}$, the mutually-independent random variables $W_{0, i}$ which take on the values $\psi(0)$ and $-\psi(0)$ with probability $\frac{1}{2}$ apiece. As well, for each $j \in \{1, 2, \dots, n-Q\}$, define:

    \begin{align*}
        W_j = \begin{cases}
            \psi(j) \quad &\text{w.p. $\frac{1}{2}$} \\
            -\psi(j) \quad &\text{w.p. $\frac{1}{2}$}
        \end{cases}
        \text{,}
    \end{align*}
where the $W_j$'s and the $W_{0, i}$'s are mutually independent. Notice that under $H_0$, 
    \begin{align*}
    W_1 \stackrel{d}{=} \sum_{i = 0}^Q W_{0, i} + \sum_{j = 1}^{n - Q} W_j. 
    \end{align*}
    We have just re-expressed the test statistic in a manner that facilitates the application of the \nameref{thm::lya}. Before using this result, we need to compute the relevant moments of the $W_{0, i}$'s and the $W_j$'s. First, we observe that for all $i \in \{1, 2, \dots, Q\}$, for all $j \in \{1, 2, \dots, n - Q\}$ and for all $m \in \nats$:

    \begin{enumerate}[1)]
        \item \begin{align} \label{eqn::EWj}
        \expect(W_j^m) &= \begin{cases}
            0 \quad &\text{if $m$ is odd} \\
            \psi^m(j) \quad &\text{if $m$ is even}
        \end{cases}
    \end{align}
        \item \begin{align}
        \expect(|W_j|^m) = \psi^m(j) \label{eqn::EabsWj}
    \end{align}
        \item \begin{align} \label{eqn::EW0}
        \expect(W_{0, i}^m) &= \begin{cases}
            0 \quad &\text{if $m$ is odd} \\
            \psi^m(0) \quad &\text{if $m$ is even}
        \end{cases}
    \end{align}
        \item \begin{align}
        \expect(|W_{0, i}|^m) = \psi^m(0) \label{eqn::EabsW0}
    \end{align}
    \item \begin{align}
        \var(W_j) = \sigma^2_j = \psi^2(j) \text{.} \label{eqn::VarWj}
    \end{align} 
    \item \begin{align}
        \var(W_{0, i}) = \sigma^2_{0, i} = \psi^2(0) \text{.} \label{eqn::VarW0}
    \end{align}
    \end{enumerate}

    To make further progress, we note that there must exist $\ell \in \nats$ such that $\ell \in \nats \cup \{0\}$ and $\psi$ is $\Omega(n^{\frac{\ell}{4}})$ and $o(n^{\frac{\ell + 1}{4}})$. Our strategy is then to bound above the growth rate of the numerator of the fraction present in Condition \ref{eqn::lyacond} and bound below the growth rate of the denominator. For reference, Condition \ref{eqn::lyacond} reads as follows in our context (due to the above equations):

    \begin{equation} \label{lyacondcontext}
            \frac{Q |\psi(0)|^{2 + \delta} + \sum_{j = 1}^{n - Q} \left(|\psi(j)|^{2 + \delta} \right)}{\left(Q \psi^2(0) + \sum_{j = 1}^{n - Q} \psi^2(j) \right)^{1 + \frac{\delta}{2}}} \rightarrow 0 \quad \text{as $n \rightarrow \infty$ for some $\delta > 0$}.
    \end{equation}

    Since $\psi$ is non-negative, increasing and $o(n^{\frac{\ell + 1}{4}})$, term $j$ of the sum in the numerator is $o(j^{\frac{\ell + 1}{2} + \delta \, \frac{\ell + 1}{4}})$, so the sum as a whole is $o((n - Q)^{\frac{\ell + 3}{2} + \delta \, \frac{\ell + 1}{4}})$. The $Q |\psi(0)|^{2 + \delta}$ term does not affect this bound on the asymptotic growth rate of the sum as $Q \leq n$ and $\psi$ is strictly increasing. Meanwhile, the sum in the denominator is $\Omega((n - Q)^{\frac{\ell}{2} + 1})$, so the denominator as a whole is $\Omega((n - Q)^{\frac{\ell + 2}{2} + \delta \, \frac{\ell + 2}{4}})$. Thus, if $\delta \geq 2$, the limit of the fraction at large as $n \rightarrow \infty$ is indeed 0. This indicates that the hypotheses of the \nameref{thm::lya} are satisfied, which is what we needed to show. The asymptotic distribution of $W_1$ then follows from the \nameref{thm::lya} and Theorem \ref{thm::sk}.
\end{proof}
\begin{proof}[Proof of Theorem~\ref{thm::WWWCRS}]
    Fix neighboring databases $\mathbf{v}$ and $\mathbf{w}$ and assume row $j$ of the former can be changed to produce the latter. We may assume, without loss of generality, that $a_j^{'} > a_j$. We will also divide the rows of $\mathbf{v}$ into three regions:
\begin{enumerate}
    \item $A_1$, consisting of the rows with ranks not exceeding $r_j$;
    \item $A_2$, consisting of the rows with ranks exceeding $r_j$ but with $a_i < a_j^{'}$, where $a_j^{'}$ is the absolute difference of row $j$ of database $\mathbf{w}$; and,
    \item $A_3$, consisting of the rows with $a_i > a_j^{'}$.
\end{enumerate}

Before proceeding further, we will define some more notation. Let $\Delta j$ be the change in the signed rank of row $j$ when transitioning from $\mathbf{v}$ to $\mathbf{w}$. In addition, for $i = 1, 2, 3$, let $\Delta A_i$ be the change in the sum of the signed ranks of all rows in region $A_i$ upon making this transition. The change in the overall test statistic, $\Delta W_1$, is the sum of these two effects. This is because the ranks of the data points in $A_1$ and $A_3$ are unaffected by the transition between the neighbouring databases.

We notice that the only rows whose ranks differ between the databases are row $j$ and the rows in region $A_2$. Let us first bound $\Delta j$. As $\psi$ is increasing and $s_j$ changes in the worst-case scenario (this is the same justification as in the proof of Theorem 4.1 of \cite{couch2018differentially}), we have that:

\begin{align*}
    \Delta j \leq \psi(n - Q) + \psi(\max\{c_1 + 1 - Q, 0\}).
\end{align*}

Meanwhile, the ranks of each row in region $A_2$ decrease by 1 apiece, while their signs remain unchanged. Thus,

\begin{align*}
    \Delta A_2 &= \left|\sum_{i \in A_2} \psi(\max\{r_i^{'} - Q, 0\} ) - \psi(\max \{r_i - Q, 0\}) \right| \\
    &\leq \sum_{i \in A_2} \left| \psi(\max \{r_i^{'} - Q, 0 \}) - \psi(\max \{r_i - Q, 0\}) \right| \quad \text{(by the triangle inequality)} \\
    &=(\psi(\max \{c_1 + 2 - Q, 0\}) - \psi(\max\{c_1 + 1 - Q, 0\})) \\
    &\phantom{=}+ (\psi(\max \{c_1 + 3 - Q, 0\}) - \psi(\max\{c_1 + 2 - Q, 0\})) \\
    &\phantom{= \,} + \dots + (\psi(\max\{c_1 + c_2 + 1 - Q, 0\}) - \psi(\max\{c_1 + c_2 - Q, 0\})) \\
    &= \psi(\max\{c_1 + c_2 + 1 - Q, 0\}) - \psi(\max \{c_1 + 1 - Q, 0\}) \\
    &\leq \psi(n - Q) - \psi(\max\{c_1 + 1 - Q, 0\}).
\end{align*}

Then,

\begin{align*}
    \Delta W_1 &\leq \Delta j + \Delta A_2 \\
    &\leq \psi(n - Q) + \psi(\max\{c_1 + 1 - Q, 0\}) + \psi(n - Q) - \psi(\max\{c_1 + 1 - Q, 0\}) \\
    &= 2 \psi(n - Q).
\end{align*}

This is what we wanted to show. 
\end{proof}
\section{Useful auxiliary results}
We now present a set of results which will be useful in our proofs pertaining to the asymptotic distributions of our test statistics. 
Our theoretical results rely upon two central limit theorems. 
The first of these results is the \nameref{thm::lya}, which generalizes the `ordinary' Central Limit Theorem to sets of random variables which are not necessarily identically distributed. We adopt the formulation of this result from \citet{pratt2012concepts}.
\begin{theorem}[\textbf{Lyapunov Central Limit Theorem}] \label{thm::lya}
If $X_1, X_2, \dots$ are independent, real-valued random variables with possibly different distributions, each having finite absolute moments of the order $2 + \delta$ for some $\delta > 0$, and if
        \begin{equation} \label{eqn::lyacond}
            \frac{\sum_{j = 1}^{n} \expect \left(|X_j - \mu_j|^{2 + \delta} \right)}{\left(\sum_{j = 1}^{n} \sigma_j^2 \right)^{1 + \frac{\delta}{2}}} \rightarrow 0 \quad \text{as $n \rightarrow \infty$},
        \end{equation}
where $\mu_j = \expect(X_j)$ and $\sigma_j^2 = \var(X_j)$, then
        \begin{align*}
            \frac{\sum_{j = 1}^{n} (X_j - \mu_j)}{\left(\sum_{j = 1}^{n} \sigma_j^2 \right)^{\frac{1}{2}}},
        \end{align*}
        converges in distribution to the standard normal distribution.
\end{theorem}
\noindent The second result is due to \citet{wald1944statistical}:
\begin{theorem}[\textbf{Wald-Wolfowitz Theorem}] \label{thm::waldwolf}
        For a sequence \seq{y} and for all $m, r \in \nats$, define:

        \begin{align*}
            \mu_r(y_n, m) = \frac{1}{m} \sum_{i=1}^m \left(y_i - \sum_{j = 1}^m y_j \right)^r.
        \end{align*}
Let the sequence \seq{a} have the property that
\begin{align} \label{eqn::wwcond}
            \frac{\mu_r(a_n, m)}{\mu_2(a_n, m)^{\frac{r}{2}}} = O(1),
\end{align}
for all $r > 2$, and let the sequence \seq{d} also satisfy this property. Let $X = \left(x_1, x_2, \dots, x_m \right)^T$ be a random vector whose elements form a permutation of the first $m$ terms of \seq{a} with probability 1. Then, the random variable 
\begin{align*}
            L_m = \sum_{i = 1}^m d_i \, x_i,
\end{align*}
has the property that
\begin{align*}
            \frac{L_m - \expect(L_m)}{\text{SD}(L_m)} \cond N(0, 1) \text{.}
        \end{align*}
\end{theorem}

\section{Correction of the Proof of Theorem A.8 of \cite{couch2019differentially}} \label{app::Correction}
In Section~\ref{sec::methods}, we asserted that there is a minor oversight in the proof of Theorem A.8 of \citet{couch2019differentially}. However, using their techniques, and correcting the mistake, Lemma~\ref{lemm::no-type-1} below implies Theorem A.8 of \citet{couch2019differentially} holds whenever 
\begin{align*}
        1 - \erf{\sqrt{\frac{3}{2}(n - 1)}} < \alpha < 0.5 \text{,}
    \end{align*} 
where $$\erf{x} = \frac{2}{\sqrt{\pi}} \int_{0}^x \exp(-t^2)  dt\,$$ is the error function. 
This is restriction is essentially immaterial, for example, it holds for $n\geq 3$ when $\alpha=0.05$, but does depend on $\alpha$. 

The issue with the proof of Theorem A.8 of \citet{couch2019differentially} is the following. 
At one point in the proof, \cite{couch2019differentially} seem to imply that if $X \sim N(\mu_X, \sigma_X^2)$ and if $Y \sim N(\mu_Y, \sigma_Y^2)$ with $\mu_X < \mu_Y$ and $\sigma^2_X < \sigma^2_Y$, then the lower-tail critical values of $Y$ are less than the corresponding critical values of $X$. In general, this is not true. To give a counterexample, the 1st percentile of $X \sim N(0, 1)$ (which is approximately -2.33) is less than the 1st percentile of $Y \sim N(-0.5, 0.5^2)$ (which is approximately -1.66). 

Another issue is that, in the Mann--Whitney test (as in the Siegel--Tukey Test), while $U_1$ and $U_2$ are both normally-distributed, the ``final'' test statistic $U$ follows a different distribution (i.e., a half-normal distribution). Below, we prove a modified version of their Theorem A.8 which accounts for this. 

In this work, we adopt a similar algorithm as that of \citet{couch2019differentially} to underestimate the difference between the groups sizes with high probability. Lemma~\ref{lemm::no-type-1} below gives that underestimating the difference in the groups sizes does not increase the type I error of the Siegel--Tukey test.

\begin{lemma}\label{lemm::no-type-1}
Suppose that $\alpha$ satisfies
    \begin{align*}
        1 - \erf{\sqrt{\frac{3}{2}(n - 1)}} < \alpha < 0.5 \text{.}
    \end{align*}
Let $\pi_\alpha(n_1; n)$ be the $100\alpha$th percentile of the limiting distribution of the private (or non-private) Siegel--Tukey test statistic under the null hypothesis when the group sizes are $n_1$ and $n - n_1$, respectively. Then, $\pi_\alpha$ is decreasing in $|\frac{n}{2} - n_1|$.
\end{lemma}
\begin{proof}
Let $n_2 = n - n_1$. We know that $U_1 \cond N(\frac{n_1n_2}{2}, \frac{n_1n_2(n + 1)}{12})$. Let us consider $U$. By definition,

\begin{align*}
    U &= \begin{cases}
        U_1 \quad &\text{if $U_1 \leq \frac{n_1n_2}{2}$} \\
        n_1n_2 - U_1 \quad &\text{if $U_1 > \frac{n_1n_2}{2}$.}
    \end{cases} 
\end{align*}

Thus,

\begin{align*}
    U - \frac{n_1n_2}{2} &= \begin{cases}
        U_1 - \frac{n_1n_2}{2} \quad &\text{if $U_1 \leq \frac{n_1n_2}{2}$} \\
        \frac{n_1n_2}{2} - U_1 \quad &\text{if $U_1 > \frac{n_1n_2}{2}$} \\
    \end{cases} \\
    &= \left|U_1 - \frac{n_1n_2}{2} \right| \text{.}
\end{align*}

Since 

\begin{align*}
    \left(U_1 - \frac{n_1n_2}{2}\right) \cond N\left(0, \frac{n_1n_2(n + 1)}{12}\right)
\end{align*}

it is the case that

\begin{align*}
    -\left|U_1 - \frac{n_1n_2}{2}\right| \stackrel{d}{\longrightarrow} \text{Half-Normal}\left(\sigma = \sqrt{\frac{n_1n_2(n + 1)}{12}} \right) \text{.}
\end{align*}

As such, $\pi_\alpha$, the $100\alpha$th percentile of $U$, is given by:

\begin{align*}
    \pi_\alpha = \frac{m(n - m)}{2} - \sqrt{\frac{m(n - m)(n + 1)}{12}} \sqrt{2} \, \erfinv{1 - \alpha} 
\end{align*}

where $m = \min \{n_1, n_2\}$. Extending $\pi_\alpha$ to be continuous in $m$ and taking the first derivative, we obtain:

\begin{align*}
    \frac{\partial \pi_\alpha}{\partial m} = \left(\frac{n}{2} - m\right)\left(1 - \frac{\erfinv{1 - \alpha}}{\sqrt{6m(n - m)}}\right) \text{.}
\end{align*}

We would like $\pi_\alpha$ to be increasing among $m \in \nats$ for $0 \leq m \leq \frac{n}{2}$. This would be the case if the following two conditions were met:

\begin{enumerate}
    \item $\frac{\partial \pi_\alpha}{\partial m}$ is positive for all $m \in \left[1, \frac{n}{2} \right)$.
    \item $\pi_\alpha(0; n) < \pi_\alpha(1; n)$.
\end{enumerate}

Let us start with item 1. We note that $\frac{n}{2} - m$ is  positive for all $m \in \left[1, \frac{n}{2} \right)$. Consequently, $\frac{\partial \pi_\alpha}{\partial m}$ is positive whenever

\begin{align*}
    \frac{\erfinv{1 - \alpha}}{\sqrt{6m(n - m)}} < 1 \text{.}
\end{align*}

We may also note that $m(n - m)$ is increasing in $m$ for $0 < m < n/2$. Hence, the left-hand side of the above expression is maximized when $m$ takes the lowest value in the interval under consideration, i.e., when $m = 1$. By substituting $m = 1$ and rearranging, we find that item 1 holds whenever

\begin{itemize}
    \item $n > \frac{\left[\erfinv{1 - \alpha} \right]^2}{6} + 1$ \quad \text{or}
    \item $\alpha > 1 - \erf{\sqrt{6(n - 1)}}$ \text{.}
\end{itemize}

Let us move on to item 2, which, after substituting $m = 0$ and $m = 1$ into the expression for $\pi_\alpha$ and setting the latter to be greater than or equal to the former, turns out to be equivalent to:

\begin{align*}
    0 \leq \frac{n - 1}{2} - \sqrt{\frac{n - 1}{6}} \erfinv{1 - \alpha}.
\end{align*}

This is in turn equivalent to either or both of the two following inequalities:

\begin{itemize}
    \item $n > \frac{2}{3} \left[\erfinv{1 - \alpha} \right]^2 + 1$,
    \item $\alpha > 1 - \erf{\sqrt{\frac{3}{2}(n - 1)}}$.
\end{itemize}

As long as these inequalities (which are stricter than the ones from before) hold, we have our desired result.  
\end{proof}

The restrictions on $\alpha$ and $n$ are very mild and would be inconsequential for the purposes of virtually any study. To illustrate, in order to be able to use any $\alpha \geq 0.05$, we only require $n \geq 3$. Even the use of $\alpha = 0.001$, which is low enough for a broad range of disciplines, only requires $n \geq 5$. 

\section{Additional Simulation Results}\label{app:sim}
\subsection{Additional Signed Rank Simulation Results}\label{app::sr}

\begin{table}[!h]
\caption{Empirical power for the RPSR test, for Normally distributed data, with an effect size of 0.25. The pair dependence is modelled with a Gaussian copula. Notice how the RPSR test performs well when $\epsilon$ is small, but the gains are minimal or non-existent for larger privacy budgets.}
\centering
\begin{tabular}{lll|ccccc}
\toprule
$q$ & $n$ & $\epsilon$ & $\tan^{-1}(r)$ & $\ln(r+1)$ & $r^{0.5}$ & $r$ & $r^2$\\
\midrule
\addlinespace[0.3em]
\multicolumn{8}{l}{$\epsilon=$\textbf{ 0.5}}\\
\cellcolor{gray!15}{\hspace{1em}0} & \cellcolor{gray!15}{500} & \cellcolor{gray!15}{0.5} & \cellcolor{gray!15}{0.966} & \cellcolor{gray!15}{0.986} & \cellcolor{gray!15}{0.984} & \cellcolor{gray!15}{0.978} & \cellcolor{gray!15}{0.936}\\
\hspace{1em}0.25 & 500 &  0.5 & 0.984 & 0.98 & 0.98 & 0.964 & 0.906\\
\cellcolor{gray!15}{\hspace{1em}0.5} & \cellcolor{gray!15}{500} & \cellcolor{gray!15}{0.5} & \cellcolor{gray!15}{0.97} & \cellcolor{gray!15}{0.97} & \cellcolor{gray!15}{0.942} & \cellcolor{gray!15}{0.912} & \cellcolor{gray!15}{0.704}\\
\hspace{1em}0.75 & 500 &  0.5 & 0.876 & 0.81 & 0.716 & 0.554 & 0.266\\
\addlinespace[0.3em]
\multicolumn{8}{l}{$\epsilon=$\textbf{1}}\\
\cellcolor{gray!15}{\hspace{1em}0} & \cellcolor{gray!15}{500} & \cellcolor{gray!15}{1} & \cellcolor{gray!15}{0.99} & \cellcolor{gray!15}{0.992} & \cellcolor{gray!15}{0.998} & \cellcolor{gray!15}{1} & \cellcolor{gray!15}{1}\\
\bottomrule
\end{tabular}
\end{table}

\begin{table}[!h]
\caption{Empirical power for the RPSR test, for Normally distributed data, with an effect size of 0.1. The pair dependence is modelled with a Gaussian copula. Notice how the RPSR test performs well when $\epsilon$ is small, but the gains are minimal or non-existent for larger privacy budgets.}
\centering
\begin{tabular}{lll|ccccc}
\toprule
$q$ & $n$ & $\epsilon$ & $\tan^{-1}(r)$ & $\ln(r+1)$ & $r^{0.5}$ & $r$ & $r^2$\\
\midrule
\addlinespace[0.3em]
\multicolumn{8}{l}{$\epsilon=$\textbf{ 0.5}}\\
\cellcolor{gray!15}{\hspace{1em}0} & \cellcolor{gray!15}{1000} & \cellcolor{gray!15}{0.5} & \cellcolor{gray!15}{0.616} & \cellcolor{gray!15}{0.664} & \cellcolor{gray!15}{0.704} & \cellcolor{gray!15}{0.7} & \cellcolor{gray!15}{0.668}\\
\hspace{1em}0.25 & 1000 &  0.5 & 0.696 & 0.694 & 0.736 & 0.688 & 0.58\\
\cellcolor{gray!15}{\hspace{1em}0.5} & \cellcolor{gray!15}{1000} & \cellcolor{gray!15}{0.5} & \cellcolor{gray!15}{0.668} & \cellcolor{gray!15}{0.684} & \cellcolor{gray!15}{0.646} & \cellcolor{gray!15}{0.582} & \cellcolor{gray!15}{0.414}\\
\hspace{1em}0.75 & 1000 &  0.5 & 0.516 & 0.466 & 0.414 & 0.306 & 0.19\\
\addlinespace[0.3em]
\multicolumn{8}{l}{$\epsilon=$\textbf{1}}\\
\cellcolor{gray!15}{\hspace{1em}0} & \cellcolor{gray!15}{1000} & \cellcolor{gray!15}{1} & \cellcolor{gray!15}{0.684} & \cellcolor{gray!15}{0.746} & \cellcolor{gray!15}{0.798} & \cellcolor{gray!15}{0.814} & \cellcolor{gray!15}{0.828}\\
\bottomrule
\end{tabular}
\end{table}

\begin{table}[!h]
\caption{Empirical power for the RPSR test, for $t_3$ distributed data, with an effect size of 0.5. The pair dependence is modelled with a copula. Notice how the RPSR test performs well when $\epsilon$ is small, but the gains are minimal or non-existent for larger privacy budgets.}
\centering
\begin{tabular}{lll|ccccc}
\toprule
$q$ & $n$ & $\epsilon$ & $\tan^{-1}(r)$ & $\ln(r+1)$ & $r^{0.5}$ & $r$ & $r^2$\\
\midrule
\addlinespace[0.3em]
\multicolumn{8}{l}{$\epsilon=$\textbf{ 0.5}}\\
\cellcolor{gray!15}{\hspace{1em}0} & \cellcolor{gray!15}{100} & \cellcolor{gray!15}{0.5} & \cellcolor{gray!15}{0.924} & \cellcolor{gray!15}{0.878} & \cellcolor{gray!15}{0.82} & \cellcolor{gray!15}{0.684} & \cellcolor{gray!15}{0.336}\\
\hspace{1em}0.25 & 100 &  0.5 & 0.922 & 0.87 & 0.728 & 0.47 & 0.23\\
\cellcolor{gray!15}{\hspace{1em}0.5} & \cellcolor{gray!15}{100} & \cellcolor{gray!15}{0.5} & \cellcolor{gray!15}{0.704} & \cellcolor{gray!15}{0.514} & \cellcolor{gray!15}{0.39} & \cellcolor{gray!15}{0.198} & \cellcolor{gray!15}{0.108}\\
\hspace{1em}0.75 & 100 &  0.5 & 0.206 & 0.136 & 0.12 & 0.076 & 0.066\\
\addlinespace[0.3em]
\multicolumn{8}{l}{$\epsilon=$\textbf{1}}\\
\cellcolor{gray!15}{\hspace{1em}0} & \cellcolor{gray!15}{100} & \cellcolor{gray!15}{1} & \cellcolor{gray!15}{0.998} & \cellcolor{gray!15}{0.99} & \cellcolor{gray!15}{0.994} & \cellcolor{gray!15}{0.974} & \cellcolor{gray!15}{0.88}\\
\bottomrule
\end{tabular}
\end{table}

\begin{table}[!h]
\caption{Empirical power for the RPSR test, for $t_3$ distributed data, with an effect size of 0.25. The pair dependence is modelled with a copula. Notice how the RPSR test performs well when $\epsilon$ is small, but the gains are minimal or non-existent for larger privacy budgets.}
\centering
\begin{tabular}{lll|ccccc}
\toprule
$q$ & $n$ & $\epsilon$ & $\tan^{-1}(r)$ & $\ln(r+1)$ & $r^{0.5}$ & $r$ & $r^2$\\
\midrule
\addlinespace[0.3em]
\multicolumn{8}{l}{$\epsilon=$\textbf{ 0.5}}\\
\cellcolor{gray!15}{\hspace{1em}0} & \cellcolor{gray!15}{500} & \cellcolor{gray!15}{0.5} & \cellcolor{gray!15}{1} & \cellcolor{gray!15}{1} & \cellcolor{gray!15}{1} & \cellcolor{gray!15}{0.998} & \cellcolor{gray!15}{0.982}\\
\hspace{1em}0.25 & 500 &  0.5 & 1 & 1 & 1 & 0.984 & 0.942\\
\cellcolor{gray!15}{\hspace{1em}0.5} & \cellcolor{gray!15}{500} & \cellcolor{gray!15}{0.5} & \cellcolor{gray!15}{0.998} & \cellcolor{gray!15}{0.998} & \cellcolor{gray!15}{0.98} & \cellcolor{gray!15}{0.946} & \cellcolor{gray!15}{0.75}\\
\hspace{1em}0.75 & 500 &  0.5 & 0.916 & 0.852 & 0.714 & 0.472 & 0.18\\
\addlinespace[0.3em]
\multicolumn{8}{l}{$\epsilon=$\textbf{1}}\\
\cellcolor{gray!15}{\hspace{1em}0} & \cellcolor{gray!15}{500} & \cellcolor{gray!15}{1} & \cellcolor{gray!15}{1} & \cellcolor{gray!15}{1} & \cellcolor{gray!15}{1} & \cellcolor{gray!15}{1} & \cellcolor{gray!15}{1}\\
\bottomrule
\end{tabular}
\end{table}

\begin{table}[!h]
\caption{Empirical power for the RPSR test, for $t_3$ distributed data, with an effect size of 0.1. The pair dependence is modelled with a copula. Notice how the RPSR test performs well when $\epsilon$ is small, but the gains are minimal or non-existent for larger privacy budgets.}
\centering
\begin{tabular}{lll|ccccc}
\toprule
$q$ & $n$ & $\epsilon$ & $\tan^{-1}(r)$ & $\ln(r+1)$ & $r^{0.5}$ & $r$ & $r^2$\\
\midrule
\addlinespace[0.3em]
\multicolumn{8}{l}{$\epsilon=$\textbf{ 0.5}}\\
\cellcolor{gray!15}{\hspace{1em}0} & \cellcolor{gray!15}{1000} & \cellcolor{gray!15}{0.5} & \cellcolor{gray!15}{0.886} & \cellcolor{gray!15}{0.91} & \cellcolor{gray!15}{0.93} & \cellcolor{gray!15}{0.904} & \cellcolor{gray!15}{0.79}\\
\hspace{1em}0.25 & 1000 &  0.5 & 0.916 & 0.918 & 0.87 & 0.844 & 0.69\\
\cellcolor{gray!15}{\hspace{1em}0.5} & \cellcolor{gray!15}{1000} & \cellcolor{gray!15}{0.5} & \cellcolor{gray!15}{0.868} & \cellcolor{gray!15}{0.858} & \cellcolor{gray!15}{0.81} & \cellcolor{gray!15}{0.668} & \cellcolor{gray!15}{0.422}\\
\hspace{1em}0.75 & 1000 &  0.5 & 0.616 & 0.506 & 0.428 & 0.248 & 0.132\\
\addlinespace[0.3em]
\multicolumn{8}{l}{$\epsilon=$\textbf{1}}\\
\cellcolor{gray!15}{\hspace{1em}0} & \cellcolor{gray!15}{1000} & \cellcolor{gray!15}{1} & \cellcolor{gray!15}{0.926} & \cellcolor{gray!15}{0.948} & \cellcolor{gray!15}{0.962} & \cellcolor{gray!15}{0.96} & \cellcolor{gray!15}{0.93}\\
\bottomrule
\end{tabular}
\end{table}
\begin{table}[!h]
\caption{Empirical size for the RPSR test, for Normally distributed data, with an effect size of 0.5. The pair dependence is modelled with a copula.}
\centering
\begin{tabular}{lll|ccccc}
\toprule
$q$ & $n$ & $\epsilon$ & $\tan^{-1}(r)$ & $\ln(r+1)$ & $r^{0.5}$ & $r$ & $r^2$\\
\midrule
\addlinespace[0.3em]
\multicolumn{8}{l}{$\epsilon=$\textbf{ 0.5}}\\
\cellcolor{gray!15}{\hspace{1em}0} & \cellcolor{gray!15}{100} & \cellcolor{gray!15}{0.5} & \cellcolor{gray!15}{0.058} & \cellcolor{gray!15}{0.05} & \cellcolor{gray!15}{0.06} & \cellcolor{gray!15}{0.056} & \cellcolor{gray!15}{0.05}\\
\hspace{1em}0.25 & 100 &  0.5 & 0.058 & 0.062 & 0.04 & 0.07 & 0.06\\
\cellcolor{gray!15}{\hspace{1em}0.5} & \cellcolor{gray!15}{100} & \cellcolor{gray!15}{0.5} & \cellcolor{gray!15}{0.068} & \cellcolor{gray!15}{0.036} & \cellcolor{gray!15}{0.042} & \cellcolor{gray!15}{0.056} & \cellcolor{gray!15}{0.04}\\
\hspace{1em}0.75 & 100 &  0.5 & 0.042 & 0.048 & 0.064 & 0.062 & 0.036\\
\addlinespace[0.3em]
\multicolumn{8}{l}{$\epsilon=$\textbf{1}}\\
\cellcolor{gray!15}{\hspace{1em}0} & \cellcolor{gray!15}{100} & \cellcolor{gray!15}{1} & \cellcolor{gray!15}{0.048} & \cellcolor{gray!15}{0.042} & \cellcolor{gray!15}{0.058} & \cellcolor{gray!15}{0.038} & \cellcolor{gray!15}{0.064}\\
\bottomrule
\end{tabular}
\end{table}

\begin{table}[!h]
\caption{Empirical size for the RPSR test, for Normally distributed data, with an effect size of 0.25. The pair dependence is modelled with a copula.}
\centering
\begin{tabular}{lll|ccccc}
\toprule
$q$ & $n$ & $\epsilon$ & $\tan^{-1}(r)$ & $\ln(r+1)$ & $r^{0.5}$ & $r$ & $r^2$\\
\midrule
\addlinespace[0.3em]
\multicolumn{8}{l}{$\epsilon=$\textbf{ 0.5}}\\
\cellcolor{gray!15}{\hspace{1em}0} & \cellcolor{gray!15}{500} & \cellcolor{gray!15}{0.5} & \cellcolor{gray!15}{0.038} & \cellcolor{gray!15}{0.044} & \cellcolor{gray!15}{0.038} & \cellcolor{gray!15}{0.038} & \cellcolor{gray!15}{0.046}\\
\hspace{1em}0.25 & 500 &  0.5 & 0.048 & 0.054 & 0.046 & 0.034 & 0.036\\
\cellcolor{gray!15}{\hspace{1em}0.5} & \cellcolor{gray!15}{500} & \cellcolor{gray!15}{0.5} & \cellcolor{gray!15}{0.054} & \cellcolor{gray!15}{0.038} & \cellcolor{gray!15}{0.04} & \cellcolor{gray!15}{0.032} & \cellcolor{gray!15}{0.042}\\
\hspace{1em}0.75 & 500 &  0.5 & 0.034 & 0.042 & 0.036 & 0.052 & 0.04\\
\addlinespace[0.3em]
\multicolumn{8}{l}{$\epsilon=$\textbf{1}}\\
\cellcolor{gray!15}{\hspace{1em}0} & \cellcolor{gray!15}{500} & \cellcolor{gray!15}{1} & \cellcolor{gray!15}{0.066} & \cellcolor{gray!15}{0.054} & \cellcolor{gray!15}{0.07} & \cellcolor{gray!15}{0.058} & \cellcolor{gray!15}{0.054}\\
\bottomrule
\end{tabular}
\end{table}

\begin{table}[!h]
\caption{Empirical size for the RPSR test, for Normally distributed data, with an effect size of 0.1. The pair dependence is modelled with a copula.}
\centering
\begin{tabular}{lll|ccccc}
\toprule
$q$ & $n$ & $\epsilon$ & $\tan^{-1}(r)$ & $\ln(r+1)$ & $r^{0.5}$ & $r$ & $r^2$\\
\midrule
\addlinespace[0.3em]
\multicolumn{8}{l}{$\epsilon=$\textbf{ 0.5}}\\
\cellcolor{gray!15}{\hspace{1em}0} & \cellcolor{gray!15}{1000} & \cellcolor{gray!15}{0.5} & \cellcolor{gray!15}{0.034} & \cellcolor{gray!15}{0.034} & \cellcolor{gray!15}{0.03} & \cellcolor{gray!15}{0.046} & \cellcolor{gray!15}{0.038}\\
\hspace{1em}0.25 & 1000 &  0.5 & 0.03 & 0.032 & 0.036 & 0.038 & 0.044\\
\cellcolor{gray!15}{\hspace{1em}0.5} & \cellcolor{gray!15}{1000} & \cellcolor{gray!15}{0.5} & \cellcolor{gray!15}{0.024} & \cellcolor{gray!15}{0.034} & \cellcolor{gray!15}{0.036} & \cellcolor{gray!15}{0.04} & \cellcolor{gray!15}{0.038}\\
\hspace{1em}0.75 & 1000 &  0.5 & 0.056 & 0.036 & 0.05 & 0.042 & 0.062\\
\addlinespace[0.3em]
\multicolumn{8}{l}{$\epsilon=$\textbf{1}}\\
\cellcolor{gray!15}{\hspace{1em}0} & \cellcolor{gray!15}{1000} & \cellcolor{gray!15}{1} & \cellcolor{gray!15}{0.056} & \cellcolor{gray!15}{0.052} & \cellcolor{gray!15}{0.05} & \cellcolor{gray!15}{0.046} & \cellcolor{gray!15}{0.05}\\
\bottomrule
\end{tabular}
\end{table}

\begin{table}[!h]
\caption{Empirical size for the RPSR test, for $t_3$ distributed data, with an effect size of 0.5. The pair dependence is modelled with a copula.}
\centering
\begin{tabular}{lll|ccccc}
\toprule
$q$ & $n$ & $\epsilon$ & $\tan^{-1}(r)$ & $\ln(r+1)$ & $r^{0.5}$ & $r$ & $r^2$\\
\midrule
\addlinespace[0.3em]
\multicolumn{8}{l}{$\epsilon=$\textbf{ 0.5}}\\
\cellcolor{gray!15}{\hspace{1em}0} & \cellcolor{gray!15}{100} & \cellcolor{gray!15}{0.5} & \cellcolor{gray!15}{0.066} & \cellcolor{gray!15}{0.044} & \cellcolor{gray!15}{0.04} & \cellcolor{gray!15}{0.048} & \cellcolor{gray!15}{0.04}\\
\hspace{1em}0.25 & 100 &  0.5 & 0.042 & 0.046 & 0.048 & 0.04 & 0.028\\
\cellcolor{gray!15}{\hspace{1em}0.5} & \cellcolor{gray!15}{100} & \cellcolor{gray!15}{0.5} & \cellcolor{gray!15}{0.058} & \cellcolor{gray!15}{0.058} & \cellcolor{gray!15}{0.046} & \cellcolor{gray!15}{0.04} & \cellcolor{gray!15}{0.042}\\
\hspace{1em}0.75 & 100 &  0.5 & 0.05 & 0.034 & 0.064 & 0.044 & 0.032\\
\addlinespace[0.3em]
\multicolumn{8}{l}{$\epsilon=$\textbf{1}}\\
\cellcolor{gray!15}{\hspace{1em}0} & \cellcolor{gray!15}{100} & \cellcolor{gray!15}{1} & \cellcolor{gray!15}{0.06} & \cellcolor{gray!15}{0.05} & \cellcolor{gray!15}{0.066} & \cellcolor{gray!15}{0.058} & \cellcolor{gray!15}{0.05}\\
\bottomrule
\end{tabular}
\end{table}

\begin{table}[!h]
\caption{Empirical size for the RPSR test, for $t_3$ distributed data, with an effect size of 0.25. The pair dependence is modelled with a copula.}
\centering
\begin{tabular}{lll|ccccc}
\toprule
$q$ & $n$ & $\epsilon$ & $\tan^{-1}(r)$ & $\ln(r+1)$ & $r^{0.5}$ & $r$ & $r^2$\\
\midrule
\addlinespace[0.3em]
\multicolumn{8}{l}{$\epsilon=$\textbf{ 0.5}}\\
\cellcolor{gray!15}{\hspace{1em}0} & \cellcolor{gray!15}{500} & \cellcolor{gray!15}{0.5} & \cellcolor{gray!15}{0.034} & \cellcolor{gray!15}{0.054} & \cellcolor{gray!15}{0.046} & \cellcolor{gray!15}{0.048} & \cellcolor{gray!15}{0.05}\\
\hspace{1em}0.25 & 500 &  0.5 & 0.042 & 0.048 & 0.06 & 0.044 & 0.048\\
\cellcolor{gray!15}{\hspace{1em}0.5} & \cellcolor{gray!15}{500} & \cellcolor{gray!15}{0.5} & \cellcolor{gray!15}{0.06} & \cellcolor{gray!15}{0.042} & \cellcolor{gray!15}{0.06} & \cellcolor{gray!15}{0.046} & \cellcolor{gray!15}{0.072}\\
\hspace{1em}0.75 & 500 &  0.5 & 0.056 & 0.054 & 0.044 & 0.05 & 0.038\\
\addlinespace[0.3em]
\multicolumn{8}{l}{$\epsilon=$\textbf{1}}\\
\cellcolor{gray!15}{\hspace{1em}0} & \cellcolor{gray!15}{500} & \cellcolor{gray!15}{1} & \cellcolor{gray!15}{0.044} & \cellcolor{gray!15}{0.042} & \cellcolor{gray!15}{0.06} & \cellcolor{gray!15}{0.05} & \cellcolor{gray!15}{0.054}\\
\bottomrule
\end{tabular}
\end{table}

\begin{table}[!h]
\caption{Empirical size for the RPSR test, for $t_3$ distributed data, with an effect size of 0.1. The pair dependence is modelled with a copula.}
\centering
\begin{tabular}{lll|ccccc}
\toprule
$q$ & $n$ & $\epsilon$ & $\tan^{-1}(r)$ & $\ln(r+1)$ & $r^{0.5}$ & $r$ & $r^2$\\
\midrule
\addlinespace[0.3em]
\multicolumn{8}{l}{$\epsilon=$\textbf{ 0.5}}\\
\cellcolor{gray!15}{\hspace{1em}0} & \cellcolor{gray!15}{1000} & \cellcolor{gray!15}{0.5} & \cellcolor{gray!15}{0.05} & \cellcolor{gray!15}{0.062} & \cellcolor{gray!15}{0.052} & \cellcolor{gray!15}{0.062} & \cellcolor{gray!15}{0.048}\\
\hspace{1em}0.25 & 1000 &  0.5 & 0.064 & 0.07 & 0.04 & 0.062 & 0.052\\
\cellcolor{gray!15}{\hspace{1em}0.5} & \cellcolor{gray!15}{1000} & \cellcolor{gray!15}{0.5} & \cellcolor{gray!15}{0.022} & \cellcolor{gray!15}{0.06} & \cellcolor{gray!15}{0.034} & \cellcolor{gray!15}{0.06} & \cellcolor{gray!15}{0.066}\\
\hspace{1em}0.75 & 1000 &  0.5 & 0.036 & 0.064 & 0.04 & 0.054 & 0.048\\
\addlinespace[0.3em]
\multicolumn{8}{l}{$\epsilon=$\textbf{1}}\\
\cellcolor{gray!15}{\hspace{1em}0} & \cellcolor{gray!15}{1000} & \cellcolor{gray!15}{1} & \cellcolor{gray!15}{0.044} & \cellcolor{gray!15}{0.044} & \cellcolor{gray!15}{0.038} & \cellcolor{gray!15}{0.052} & \cellcolor{gray!15}{0.05}\\
\bottomrule
\end{tabular}
\end{table}
\subsection{Additional test of tests simulation results}\label{app:tot}
\begin{table}[!h]
\caption{This table displays the empirical size and power for different $\alpha_0$ and $m=\ceil{n/m_\star}$ in the test of tests procedures with the base Siegel--Tukey test, written in the formal `size'/`power' for $n=100$ and $\theta=2$. }
\centering
\begin{tabular}{cccccc}
\toprule
$\epsilon$ & $m_\star/\alpha_0=$ & 0.05 & 0.1 & 0.25 & 0.5\\
\midrule
0.5 & 5 & 0.024/0.032 & 0.096/0.134 & 0.212/0.268 & 0.06/0.11\\
0.5 & 15 & 0.058/0.074 & 0.048/0.1 & 0.04/0.156 & 0.062/0.086\\
0.5 & 25 & 0.042/0.11 & 0.058/0.118 & 0.062/0.122 & 0.05/0.114\\
0.5 & 50 & 0.044/0.086 & 0.054/0.114 & 0.05/0.094 & 0.042/0.084\\
\hline
1.0 & 5 & 0.02/0.028 & 0.172/0.222 & 0.37/0.474 & 0.088/0.11\\
1.0 & 15 & 0.048/0.152 & 0.04/0.254 & 0.054/0.248 & 0.074/0.184\\
1.0 & 25 & 0.044/0.248 & 0.048/0.234 & 0.06/0.29 & 0.066/0.136\\
1.0 & 50 & 0.052/0.208 & 0.06/0.186 & 0.042/0.148 & 0.064/0.1\\
\hline
5.0 & 5 & 0.01/0.006 & 0.218/0.286 & 0.438/0.522 & 0.058/0.128\\
5.0 & 15 & 0.048/0.35 & 0.074/0.45 & 0.074/0.424 & 0.05/0.314\\
5.0 & 25 & 0.05/0.52 & 0.054/0.642 & 0.052/0.62 & 0.046/0.422\\
5.0 & 50 & 0.06/0.688 & 0.052/0.688 & 0.05/0.656 & 0.044/0.17\\
\bottomrule
\end{tabular}
\end{table}
\begin{table}[!h]
\caption{This table displays the empirical size and power for different $\alpha_0$ and $m=\ceil{n/m_\star}$ in the test of tests procedures with the base Siegel--Tukey test, written in the formal `size'/`power' for $n=500$ and $\theta=1.5$. }
\centering
\begin{tabular}{cccccc}
\toprule
$\epsilon$ & $m_\star$ & 0.05 & 0.1 & 0.25 & 0.5\\
\midrule
0.5 & 5 & 0.012/0.004 & 0.354/0.482 & 0.924/0.93 & 0.134/0.16\\
0.5 & 15 & 0.032/0.122 & 0.076/0.194 & 0.062/0.238 & 0.066/0.19\\
0.5 & 25 & 0.05/0.156 & 0.058/0.214 & 0.068/0.312 & 0.054/0.248\\
0.5 & 50 & 0.052/0.206 & 0.056/0.258 & 0.07/0.296 & 0.046/0.216\\
\hline
1.0 & 5 & 0.004/0.002 & 0.542/0.572 & 0.952/0.954 & 0.16/0.172\\
1.0 & 15 & 0.05/0.208 & 0.078/0.312 & 0.058/0.4 & 0.072/0.282\\
1.0 & 25 & 0.05/0.41 & 0.052/0.5 & 0.048/0.562 & 0.064/0.412\\
1.0 & 50 & 0.046/0.504 & 0.05/0.594 & 0.056/0.646 & 0.048/0.458\\
\hline
5.0 & 5 & 0/0 & 0.548/0.62 & 0.962/0.97 & 0.132/0.178\\
5.0 & 15 & 0.04/0.292 & 0.06/0.452 & 0.062/0.434 & 0.078/0.334\\
5.0 & 25 & 0.044/0.534 & 0.034/0.612 & 0.038/0.592 & 0.044/0.456\\
5.0 & 50 & 0.044/0.742 & 0.048/0.844 & 0.058/0.828 & 0.042/0.696\\
\bottomrule
\end{tabular}
\end{table}

\begin{table}[!h]
\caption{This table displays the empirical size and power for different $\alpha_0$ and $m=\ceil{n/m_\star}$ in the test of tests procedures with the base Siegel--Tukey test, written in the formal `size'/`power' for $n=1000$ and $\theta=1.25$. }
\centering
\begin{tabular}{cccccc}
\toprule
$\epsilon$ & $m_\star$ & 0.05 & 0.1 & 0.25 & 0.5\\
\midrule
0.5 & 5 & 0.006/0.004 & 0.67/0.704 & 0.994/0.998 & 0.22/0.238\\
0.5 & 15 & 0.04/0.06 & 0.052/0.126 & 0.088/0.18 & 0.074/0.15\\
0.5 & 25 & 0.068/0.102 & 0.04/0.142 & 0.056/0.154 & 0.052/0.132\\
0.5 & 50 & 0.058/0.118 & 0.042/0.18 & 0.06/0.218 & 0.04/0.168\\
\hline
1.0 & 5 & 0/0 & 0.75/0.792 & 1/1 & 0.232/0.228\\
1.0 & 15 & 0.046/0.082 & 0.064/0.146 & 0.086/0.214 & 0.106/0.148\\
1.0 & 25 & 0.03/0.17 & 0.058/0.208 & 0.056/0.252 & 0.052/0.184\\
1.0 & 50 & 0.042/0.258 & 0.048/0.324 & 0.052/0.372 & 0.048/0.24\\
\hline
5.0 & 5 & 0/0 & 0.794/0.826 & 0.998/1 & 0.236/0.236\\
5.0 & 15 & 0.05/0.106 & 0.076/0.206 & 0.072/0.218 & 0.066/0.17\\
5.0 & 25 & 0.038/0.198 & 0.048/0.23 & 0.046/0.246 & 0.056/0.22\\
5.0 & 50 & 0.04/0.396 & 0.044/0.412 & 0.054/0.41 & 0.058/0.322\\
\bottomrule
\end{tabular}
\end{table}

\subsection{Additional simulation results for the RPTS test}\label{app::RPTS}
\begin{table}[!h]
\centering
\caption{Empirical power of each test for the simulation study discussed in Section~\ref{sec::sim} for $n=100$, $\epsilon=0.5$, $\theta=2$.}
\centering
\begin{tabular}[t]{llllll|ll}
\toprule
$q$ & $\tan^{-1}(r)$ & $\ln(r + 1)$ & $r^{0.5}$ & $r$ & $r^2$ & ToT $r$ & ToT $\tan^{-1}(r)$\\
\midrule
[0.3em]
\multicolumn{8}{l}{\textbf{Normal}}\\
\hspace{1em}0 & 0.048 & 0.076 & 0.066 & 0.08 & 0.06 & 0.122 & -\\
\hspace{1em}0.25 & 0.068 & 0.068 & 0.074 & 0.062 & 0.074 & - & -\\
\hspace{1em}0.5 & 0.102 & 0.092 & 0.082 & 0.064 & 0.07 & - & -\\
\hspace{1em}0.75 & 0.08 & 0.058 & 0.066 & 0.044 & 0.054 & - & 0.126\\
\addlinespace[0.3em]
\multicolumn{8}{l}{\textbf{Exponential}}\\
\hspace{1em}0 & 0.04 & 0.044 & 0.06 & 0.062 & 0.058 & 0.156 & -\\
\hspace{1em}0.25 & 0.072 & 0.096 & 0.088 & 0.072 & 0.078 & - & -\\
\hspace{1em}0.5 & 0.122 & 0.084 & 0.092 & 0.078 & 0.064 & - & -\\
\hspace{1em}0.75 & 0.09 & 0.08 & 0.064 & 0.08 & 0.096 & - & 0.13\\
\addlinespace[0.3em]
\multicolumn{8}{l}{\textbf{Lomax}}\\
\hspace{1em}0 & 0.06 & 0.07 & 0.052 & 0.084 & 0.07 & 0.138 & -\\
\hspace{1em}0.25 & 0.07 & 0.07 & 0.072 & 0.08 & 0.072 & - & -\\
\hspace{1em}0.5 & 0.086 & 0.098 & 0.07 & 0.066 & 0.064 & - & -\\
\hspace{1em}0.75 & 0.096 & 0.074 & 0.05 & 0.048 & 0.056 & - & 0.138\\
\bottomrule
\end{tabular}
\end{table}
\begin{table}[!h]
\centering
\caption{Empirical power of each test for the simulation study discussed in Section~\ref{sec::sim} for $n=100$, $\epsilon=1$, $\theta=2$.}
\centering
\begin{tabular}[t]{llllll|ll}
\toprule
$q$ & $\tan^{-1}(r)$ & $\ln(r + 1)$ & $r^{0.5}$ & $r$ & $r^2$ & ToT $r$ & ToT $\tan^{-1}(r)$\\
\midrule
[0.3em]
\multicolumn{8}{l}{\textbf{Normal}}\\
\hspace{1em}0 & 0.074 & 0.078 & 0.112 & 0.128 & 0.146 & 0.26 & -\\
\hspace{1em}0.25 & 0.12 & 0.168 & 0.138 & 0.142 & 0.13 & - & -\\
\hspace{1em}0.5 & 0.24 & 0.216 & 0.19 & 0.154 & 0.11 & - & -\\
\hspace{1em}0.75 & 0.192 & 0.15 & 0.112 & 0.096 & 0.086 & - & 0.28\\
\addlinespace[0.3em]
\multicolumn{8}{l}{\textbf{Exponential}}\\
\hspace{1em}0 & 0.07 & 0.068 & 0.122 & 0.16 & 0.17 & 0.348 & -\\
\hspace{1em}0.25 & 0.158 & 0.192 & 0.202 & 0.192 & 0.17 & - & -\\
\hspace{1em}0.5 & 0.356 & 0.31 & 0.268 & 0.176 & 0.096 & - & -\\
\hspace{1em}0.75 & 0.246 & 0.132 & 0.12 & 0.086 & 0.084 & - & 0.376\\
\addlinespace[0.3em]
\multicolumn{8}{l}{\textbf{Lomax}}\\
\hspace{1em}0 & 0.052 & 0.114 & 0.12 & 0.162 & 0.146 & 0.336 & -\\
\hspace{1em}0.25 & 0.164 & 0.18 & 0.194 & 0.216 & 0.15 & - & -\\
\hspace{1em}0.5 & 0.304 & 0.282 & 0.23 & 0.186 & 0.11 & - & -\\
\hspace{1em}0.75 & 0.23 & 0.126 & 0.164 & 0.11 & 0.066 & - & 0.36\\
\bottomrule
\end{tabular}
\end{table}
\begin{table}[!h]
\centering
\caption{Empirical power of each test for the simulation study discussed in Section~\ref{sec::sim} for $n=100$, $\epsilon=5$, $\theta=2$.}
\centering
\begin{tabular}[t]{llllll|ll}
\toprule
$q$ & $\tan^{-1}(r)$ & $\ln(r + 1)$ & $r^{0.5}$ & $r$ & $r^2$ & ToT $r$ & ToT $\tan^{-1}(r)$\\
\midrule
[0.3em]
\multicolumn{8}{l}{\textbf{Normal}}\\
\hspace{1em}0 & 0.252 & 0.882 & 0.9 & 0.872 & 0.766 & 0.664 & -\\
\hspace{1em}0.25 & 0.498 & 0.744 & 0.862 & 0.912 & 0.942 & - & -\\
\hspace{1em}0.5 & 0.834 & 0.904 & 0.926 & 0.95 & 0.924 & - & -\\
\hspace{1em}0.75 & 0.944 & 0.928 & 0.92 & 0.882 & 0.752 & - & 0.734\\
\addlinespace[0.3em]
\multicolumn{8}{l}{\textbf{Exponential}}\\
\hspace{1em}0 & 0.28 & 0.92 & 0.94 & 0.916 & 0.786 & 0.742 & -\\
\hspace{1em}0.25 & 0.502 & 0.796 & 0.884 & 0.954 & 0.95 & - & -\\
\hspace{1em}0.5 & 0.876 & 0.926 & 0.966 & 0.978 & 0.966 & - & -\\
\hspace{1em}0.75 & 0.984 & 0.98 & 0.972 & 0.914 & 0.752 & - & 0.776\\
\addlinespace[0.3em]
\multicolumn{8}{l}{\textbf{Lomax}}\\
\hspace{1em}0 & 0.216 & 0.886 & 0.948 & 0.906 & 0.808 & 0.762 & -\\
\hspace{1em}0.25 & 0.486 & 0.802 & 0.9 & 0.958 & 0.962 & - & -\\
\hspace{1em}0.5 & 0.882 & 0.938 & 0.964 & 0.97 & 0.942 & - & -\\
\hspace{1em}0.75 & 0.978 & 0.964 & 0.932 & 0.844 & 0.628 & - & 0.714\\
\bottomrule
\end{tabular}
\end{table}
\begin{table}[!h]
\centering
\caption{Empirical power of each test for the simulation study discussed in Section~\ref{sec::sim} for $n=500$, $\epsilon=0.5$, $\theta=1.5$.}
\centering
\begin{tabular}[t]{llllll|ll}
\toprule
$q$ & $\tan^{-1}(r)$ & $\ln(r + 1)$ & $r^{0.5}$ & $r$ & $r^2$ & ToT $r$ & ToT $\tan^{-1}(r)$\\
\midrule
[0.3em]
\multicolumn{8}{l}{\textbf{Normal}}\\
\hspace{1em}0 & 0.062 & 0.114 & 0.226 & 0.352 & 0.326 & 0.246 & -\\
\hspace{1em}0.25 & 0.28 & 0.368 & 0.422 & 0.432 & 0.328 & - & -\\
\hspace{1em}0.5 & 0.564 & 0.572 & 0.508 & 0.4 & 0.284 & - & -\\
\hspace{1em}0.75 & 0.532 & 0.414 & 0.364 & 0.224 & 0.15 & - & 0.322\\
\addlinespace[0.3em]
\multicolumn{8}{l}{\textbf{Exponential}}\\
\hspace{1em}0 & 0.068 & 0.144 & 0.362 & 0.468 & 0.368 & 0.394 & -\\
\hspace{1em}0.25 & 0.234 & 0.418 & 0.482 & 0.61 & 0.574 & - & -\\
\hspace{1em}0.5 & 0.646 & 0.706 & 0.694 & 0.62 & 0.472 & - & -\\
\hspace{1em}0.75 & 0.85 & 0.742 & 0.62 & 0.424 & 0.166 & - & 0.526\\
\addlinespace[0.3em]
\multicolumn{8}{l}{\textbf{Lomax}}\\
\hspace{1em}0 & 0.056 & 0.156 & 0.346 & 0.462 & 0.334 & 0.392 & -\\
\hspace{1em}0.25 & 0.246 & 0.428 & 0.512 & 0.562 & 0.542 & - & -\\
\hspace{1em}0.5 & 0.63 & 0.714 & 0.69 & 0.638 & 0.424 & - & -\\
\hspace{1em}0.75 & 0.866 & 0.732 & 0.602 & 0.322 & 0.152 & - & 0.516\\
\bottomrule
\end{tabular}
\end{table}
\begin{table}[!h]
\centering
\caption{Empirical power of each test for the simulation study discussed in Section~\ref{sec::sim} for $n=500$, $\epsilon=1$, $\theta=1.5$.}
\centering
\begin{tabular}[t]{llllll|ll}
\toprule
$q$ & $\tan^{-1}(r)$ & $\ln(r + 1)$ & $r^{0.5}$ & $r$ & $r^2$ & ToT $r$ & ToT $\tan^{-1}(r)$\\
\midrule
[0.3em]
\multicolumn{8}{l}{\textbf{Normal}}\\
\hspace{1em}0 & 0.05 & 0.484 & 0.824 & 0.84 & 0.76 & 0.46 & -\\
\hspace{1em}0.25 & 0.522 & 0.76 & 0.852 & 0.886 & 0.9 & - & -\\
\hspace{1em}0.5 & 0.85 & 0.918 & 0.944 & 0.902 & 0.842 & - & -\\
\hspace{1em}0.75 & 0.914 & 0.926 & 0.85 & 0.77 & 0.54 & - & 0.506\\
\addlinespace[0.3em]
\multicolumn{8}{l}{\textbf{Exponential}}\\
\hspace{1em}0 & 0.086 & 0.71 & 0.918 & 0.936 & 0.856 & 0.644 & -\\
\hspace{1em}0.25 & 0.564 & 0.838 & 0.936 & 0.972 & 0.974 & - & -\\
\hspace{1em}0.5 & 0.922 & 0.972 & 0.976 & 0.972 & 0.964 & - & -\\
\hspace{1em}0.75 & 0.994 & 0.988 & 0.966 & 0.946 & 0.752 & - & 0.818\\
\addlinespace[0.3em]
\multicolumn{8}{l}{\textbf{Lomax}}\\
\hspace{1em}0 & 0.06 & 0.672 & 0.872 & 0.942 & 0.85 & 0.656 & -\\
\hspace{1em}0.25 & 0.526 & 0.814 & 0.92 & 0.95 & 0.97 & - & -\\
\hspace{1em}0.5 & 0.928 & 0.954 & 0.978 & 0.982 & 0.94 & - & -\\
\hspace{1em}0.75 & 0.994 & 0.992 & 0.972 & 0.88 & 0.624 & - & 0.75\\
\bottomrule
\end{tabular}
\end{table}
\begin{table}[!h]
\centering
\caption{Empirical power of each test for the simulation study discussed in Section~\ref{sec::sim} for $n=500$, $\epsilon=5$, $\theta=1.5$.}
\centering
\begin{tabular}[t]{llllll|ll}
\toprule
$q$ & $\tan^{-1}(r)$ & $\ln(r + 1)$ & $r^{0.5}$ & $r$ & $r^2$ & ToT $r$ & ToT $\tan^{-1}(r)$\\
\midrule
[0.3em]
\multicolumn{8}{l}{\textbf{Normal}}\\
\hspace{1em}0 & 0.534 & 1 & 1 & 1 & 0.984 & 0.62 & -\\
\hspace{1em}0.25 & 0.722 & 0.942 & 0.998 & 1 & 1 & - & -\\
\hspace{1em}0.5 & 0.978 & 1 & 1 & 1 & 1 & - & -\\
\hspace{1em}0.75 & 1 & 1 & 1 & 1 & 1 & - & 0.68\\
\addlinespace[0.3em]
\multicolumn{8}{l}{\textbf{Exponential}}\\
\hspace{1em}0 & 0.532 & 1 & 1 & 1 & 0.986 & 0.848 & -\\
\hspace{1em}0.25 & 0.748 & 0.968 & 0.998 & 1 & 1 & - & -\\
\hspace{1em}0.5 & 0.986 & 1 & 1 & 1 & 1 & - & -\\
\hspace{1em}0.75 & 1 & 1 & 1 & 1 & 1 & - & 0.928\\
\addlinespace[0.3em]
\multicolumn{8}{l}{\textbf{Lomax}}\\
\hspace{1em}0 & 0.422 & 1 & 1 & 0.998 & 0.99 & 0.826 & -\\
\hspace{1em}0.25 & 0.798 & 0.984 & 0.996 & 1 & 1 & - & -\\
\hspace{1em}0.5 & 0.988 & 1 & 1 & 1 & 1 & - & -\\
\hspace{1em}0.75 & 1 & 1 & 1 & 1 & 1 & - & 0.91\\
\bottomrule
\end{tabular}
\end{table}
\begin{table}[!h]
\centering
\caption{Empirical power of each test for the simulation study discussed in Section~\ref{sec::sim} for $n=1000$, $\epsilon=0.5$, $\theta=1.25$.}
\centering
\begin{tabular}[t]{llllll|ll}
\toprule
$q$ & $\tan^{-1}(r)$ & $\ln(r + 1)$ & $r^{0.5}$ & $r$ & $r^2$ & ToT $r$ & ToT $\tan^{-1}(r)$\\
\midrule
[0.3em]
\multicolumn{8}{l}{\textbf{Normal}}\\
\hspace{1em}0 & 0.05 & 0.14 & 0.32 & 0.378 & 0.374 & 0.13 & -\\
\hspace{1em}0.25 & 0.276 & 0.37 & 0.436 & 0.492 & 0.44 & - & -\\
\hspace{1em}0.5 & 0.528 & 0.59 & 0.584 & 0.47 & 0.328 & - & -\\
\hspace{1em}0.75 & 0.584 & 0.51 & 0.432 & 0.302 & 0.186 & - & 0.17\\
\addlinespace[0.3em]
\multicolumn{8}{l}{\textbf{Exponential}}\\
\hspace{1em}0 & 0.044 & 0.216 & 0.568 & 0.63 & 0.458 & 0.24 & -\\
\hspace{1em}0.25 & 0.268 & 0.474 & 0.608 & 0.714 & 0.736 & - & -\\
\hspace{1em}0.5 & 0.612 & 0.706 & 0.76 & 0.792 & 0.714 & - & -\\
\hspace{1em}0.75 & 0.866 & 0.856 & 0.812 & 0.728 & 0.484 & - & 0.422\\
\addlinespace[0.3em]
\multicolumn{8}{l}{\textbf{Lomax}}\\
\hspace{1em}0 & 0.06 & 0.218 & 0.566 & 0.612 & 0.466 & 0.274 & -\\
\hspace{1em}0.25 & 0.274 & 0.432 & 0.614 & 0.742 & 0.748 & - & -\\
\hspace{1em}0.5 & 0.636 & 0.782 & 0.808 & 0.776 & 0.708 & - & -\\
\hspace{1em}0.75 & 0.86 & 0.858 & 0.824 & 0.708 & 0.458 & - & 0.446\\
\bottomrule
\end{tabular}
\end{table}
\begin{table}[!h]
\centering
\caption{Empirical power of each test for the simulation study discussed in Section~\ref{sec::sim} for $n=1000$, $\epsilon=1$, $\theta=1.25$.}
\centering
\begin{tabular}[t]{llllll|ll}
\toprule
$q$ & $\tan^{-1}(r)$ & $\ln(r + 1)$ & $r^{0.5}$ & $r$ & $r^2$ & ToT $r$ & ToT $\tan^{-1}(r)$\\
\midrule
[0.3em]
\multicolumn{8}{l}{\textbf{Normal}}\\
\hspace{1em}0 & 0.046 & 0.494 & 0.782 & 0.83 & 0.678 & 0.188 & -\\
\hspace{1em}0.25 & 0.42 & 0.624 & 0.784 & 0.864 & 0.854 & - & -\\
\hspace{1em}0.5 & 0.754 & 0.848 & 0.858 & 0.88 & 0.828 & - & -\\
\hspace{1em}0.75 & 0.886 & 0.872 & 0.862 & 0.812 & 0.61 & - & 0.26\\
\addlinespace[0.3em]
\multicolumn{8}{l}{\textbf{Exponential}}\\
\hspace{1em}0 & 0.08 & 0.836 & 0.952 & 0.944 & 0.816 & 0.354 & -\\
\hspace{1em}0.25 & 0.466 & 0.758 & 0.914 & 0.988 & 0.996 & - & -\\
\hspace{1em}0.5 & 0.844 & 0.928 & 0.976 & 0.97 & 0.992 & - & -\\
\hspace{1em}0.75 & 0.992 & 1 & 0.992 & 0.986 & 0.952 & - & 0.596\\
\addlinespace[0.3em]
\multicolumn{8}{l}{\textbf{Lomax}}\\
\hspace{1em}0 & 0.054 & 0.804 & 0.948 & 0.946 & 0.818 & 0.382 & -\\
\hspace{1em}0.25 & 0.42 & 0.732 & 0.92 & 0.968 & 0.974 & - & -\\
\hspace{1em}0.5 & 0.852 & 0.936 & 0.972 & 0.988 & 0.986 & - & -\\
\hspace{1em}0.75 & 0.984 & 0.992 & 0.992 & 0.99 & 0.93 & - & 0.578\\
\bottomrule
\end{tabular}
\end{table}
\begin{table}[!h]
\centering
\caption{Empirical power of each test for the simulation study discussed in Section~\ref{sec::sim} for $n=1000$, $\epsilon=5$, $\theta=1.25$.}
\centering
\begin{tabular}[h]{llllll|ll}
\toprule
$q$ & $\tan^{-1}(r)$ & $\ln(r + 1)$ & $r^{0.5}$ & $r$ & $r^2$ & ToT $r$ & ToT $\tan^{-1}(r)$\\
\midrule
[0.3em]
\multicolumn{8}{l}{\textbf{Normal}}\\
\hspace{1em}0 & 0.35 & 0.994 & 0.992 & 0.972 & 0.866 & 0.23 & -\\
\hspace{1em}0.25 & 0.53 & 0.788 & 0.94 & 0.978 & 0.99 & - & -\\
\hspace{1em}0.5 & 0.86 & 0.934 & 0.974 & 0.982 & 0.99 & - & -\\
\hspace{1em}0.75 & 0.954 & 0.972 & 0.978 & 0.982 & 0.986 & - & 0.368\\
\addlinespace[0.3em]
\multicolumn{8}{l}{\textbf{Exponential}}\\
\hspace{1em}0 & 0.464 & 1 & 1 & 1 & 0.946 & 0.482 & -\\
\hspace{1em}0.25 & 0.518 & 0.844 & 0.996 & 1 & 1 & - & -\\
\hspace{1em}0.5 & 0.928 & 0.986 & 1 & 1 & 1 & - & -\\
\hspace{1em}0.75 & 1 & 1 & 1 & 1 & 1 & - & 0.726\\
\addlinespace[0.3em]
\multicolumn{8}{l}{\textbf{Lomax}}\\
\hspace{1em}0 & 0.378 & 1 & 1 & 1 & 0.952 & 0.462 & -\\
\hspace{1em}0.25 & 0.53 & 0.824 & 0.988 & 1 & 1 & - & -\\
\hspace{1em}0.5 & 0.918 & 0.988 & 1 & 1 & 1 & - & -\\
\hspace{1em}0.75 & 1 & 1 & 1 & 1 & 1 & - & 0.724\\
\bottomrule
\end{tabular}
\end{table}
\begin{table}[t]
\caption{Empirical size of each test for the simulation study discussed in Section~\ref{sec::sim} for $n=100$, $\epsilon=0.5$.}
\centering
\begin{tabular}{llllll|ll}
$q$ & $\tan^{-1}(r)$ & $\ln(r + 1)$ & $r^{0.5}$ & $r$ & $r^2$ & ToT $r$ & ToT $\tan^{-1}(r)$\\
\hline
\multicolumn{8}{l}{\textbf{Normal}}\\
\hspace{1em}0 & 0.062 & 0.04 & 0.048 & 0.066 & 0.05 & 0.028 & -\\
\hspace{1em}0.25 & 0.052 & 0.058 & 0.05 & 0.058 & 0.046 & - & -\\
\hspace{1em}0.5 & 0.042 & 0.046 & 0.068 & 0.06 & 0.068 & - & -\\
\hspace{1em}0.75 & 0.042 & 0.048 & 0.062 & 0.06 & 0.042 & - & 0.04\\
\addlinespace[0.3em]
\multicolumn{8}{l}{\textbf{Exponential}}\\
\hspace{1em}0 & 0.056 & 0.044 & 0.06 & 0.04 & 0.066 & 0.054 & -\\
\hspace{1em}0.25 & 0.054 & 0.036 & 0.056 & 0.056 & 0.054 & - & -\\
\hspace{1em}0.5 & 0.07 & 0.038 & 0.04 & 0.052 & 0.048 & - & -\\
\hspace{1em}0.75 & 0.058 & 0.048 & 0.06 & 0.05 & 0.054 & - & 0.054\\
\addlinespace[0.3em]
\multicolumn{8}{l}{\textbf{Lomax}}\\
\hspace{1em}0 & 0.04 & 0.046 & 0.052 & 0.052 & 0.036 & 0.052 & -\\
\hspace{1em}0.25 & 0.036 & 0.068 & 0.046 & 0.042 & 0.064 & - & -\\
\hspace{1em}0.5 & 0.052 & 0.054 & 0.044 & 0.054 & 0.03 & - & -\\
\hspace{1em}0.75 & 0.058 & 0.046 & 0.056 & 0.032 & 0.044 & - & 0.04\\
\bottomrule
\end{tabular}
\end{table}
\begin{table}[h]
\caption{Empirical size of each test for the simulation study discussed in Section~\ref{sec::sim} for $n=100$, $\epsilon=1$.}
\centering
\begin{tabular}{llllll|ll}
$q$ & $\tan^{-1}(r)$ & $\ln(r + 1)$ & $r^{0.5}$ & $r$ & $r^2$ & ToT $r$ & ToT $\tan^{-1}(r)$ \\
\hline
\multicolumn{8}{l}{\textbf{Normal}}\\
\hspace{1em}0 & 0.058 & 0.05 & 0.048 & 0.042 & 0.05 & 0.05 & -\\
\hspace{1em}0.25 & 0.046 & 0.048 & 0.072 & 0.044 & 0.056 & - & -\\
\hspace{1em}0.5 & 0.052 & 0.07 & 0.038 & 0.066 & 0.046 & - & -\\
\hspace{1em}0.75 & 0.046 & 0.052 & 0.04 & 0.052 & 0.064 & - & 0.052\\
\addlinespace[0.3em]
\multicolumn{8}{l}{\textbf{Exponential}}\\
\hspace{1em}0 & 0.04 & 0.054 & 0.048 & 0.046 & 0.048 & 0.062 & -\\
\hspace{1em}0.25 & 0.082 & 0.048 & 0.042 & 0.054 & 0.044 & - & -\\
\hspace{1em}0.5 & 0.05 & 0.044 & 0.074 & 0.042 & 0.056 & - & -\\
\hspace{1em}0.75 & 0.058 & 0.038 & 0.058 & 0.06 & 0.036 & - & 0.054\\
\addlinespace[0.3em]
\multicolumn{8}{l}{\textbf{Lomax}}\\
\hspace{1em}0 & 0.042 & 0.052 & 0.038 & 0.048 & 0.05 & 0.054 & -\\
\hspace{1em}0.25 & 0.036 & 0.044 & 0.044 & 0.034 & 0.052 & - & -\\
\hspace{1em}0.5 & 0.038 & 0.048 & 0.07 & 0.054 & 0.054 & - & -\\
\hspace{1em}0.75 & 0.044 & 0.064 & 0.044 & 0.046 & 0.052 & - & 0.056\\
\bottomrule
\end{tabular}
\end{table}
\begin{table}[h]
\caption{Empirical size of each test for the simulation study discussed in Section~\ref{sec::sim} for $n=100$, $\epsilon=5$.}
\centering
\begin{tabular}{llllll|ll}
$q$ & $\tan^{-1}(r)$ & $\ln(r + 1)$ & $r^{0.5}$ & $r$ & $r^2$ & ToT $r$ & ToT $\tan^{-1}(r)$ \\
\hline
\multicolumn{8}{l}{\textbf{Normal}}\\
\hspace{1em}0 & 0.046 & 0.062 & 0.062 & 0.048 & 0.038 & 0.056 & -\\
\hspace{1em}0.25 & 0.046 & 0.052 & 0.05 & 0.044 & 0.052 & - & -\\
\hspace{1em}0.5 & 0.046 & 0.05 & 0.084 & 0.064 & 0.054 & - & -\\
\hspace{1em}0.75 & 0.054 & 0.048 & 0.054 & 0.05 & 0.052 & - & 0.062\\
\addlinespace[0.3em]
\multicolumn{8}{l}{\textbf{Exponential}}\\
\hspace{1em}0 & 0.042 & 0.04 & 0.058 & 0.046 & 0.042 & 0.058 & -\\
\hspace{1em}0.25 & 0.036 & 0.042 & 0.03 & 0.048 & 0.042 & - & -\\
\hspace{1em}0.5 & 0.048 & 0.05 & 0.048 & 0.066 & 0.038 & - & -\\
\hspace{1em}0.75 & 0.054 & 0.048 & 0.06 & 0.054 & 0.046 & - & 0.06\\
\addlinespace[0.3em]
\multicolumn{8}{l}{\textbf{Lomax}}\\
\hspace{1em}0 & 0.048 & 0.06 & 0.05 & 0.048 & 0.044 & 0.034 & -\\
\hspace{1em}0.25 & 0.05 & 0.042 & 0.042 & 0.046 & 0.046 & - & -\\
\hspace{1em}0.5 & 0.05 & 0.054 & 0.04 & 0.034 & 0.034 & - & -\\
\hspace{1em}0.75 & 0.048 & 0.038 & 0.036 & 0.048 & 0.034 & - & 0.052\\
\bottomrule
\end{tabular}
\end{table}

\begin{table}[h]
\caption{Empirical size of each test for the simulation study discussed in Section~\ref{sec::sim} for $n=500$, $\epsilon=0.5$.}
\centering
\begin{tabular}{llllll|ll}
\hline 
$q$ & $\tan^{-1}(r)$ & $\ln(r + 1)$ & $r^{0.5}$ & $r$ & $r^2$ & ToT $r$ & ToT $\tan^{-1}(r)$\\
\addlinespace[0.3em]
\multicolumn{8}{l}{\textbf{Normal}}\\
\hspace{1em}0 & 0.04 & 0.07 & 0.046 & 0.044 & 0.044 & 0.052 & -\\
\hspace{1em}0.25 & 0.034 & 0.034 & 0.056 & 0.076 & 0.044 & - & -\\
\hspace{1em}0.5 & 0.074 & 0.06 & 0.054 & 0.058 & 0.058 & - & -\\
\hspace{1em}0.75 & 0.066 & 0.054 & 0.058 & 0.056 & 0.062 & - & 0.058\\
\addlinespace[0.3em]
\multicolumn{8}{l}{\textbf{Exponential}}\\
\hspace{1em}0 & 0.038 & 0.046 & 0.042 & 0.056 & 0.048 & 0.046 & -\\
\hspace{1em}0.25 & 0.066 & 0.062 & 0.05 & 0.048 & 0.034 & - & -\\
\hspace{1em}0.5 & 0.062 & 0.056 & 0.04 & 0.056 & 0.066 & - & -\\
\hspace{1em}0.75 & 0.062 & 0.052 & 0.068 & 0.058 & 0.048 & - & 0.07\\
\addlinespace[0.3em]
\multicolumn{8}{l}{\textbf{Lomax}}\\
\hspace{1em}0 & 0.042 & 0.03 & 0.056 & 0.048 & 0.044 & 0.036 & -\\
\hspace{1em}0.25 & 0.056 & 0.036 & 0.04 & 0.044 & 0.066 & - & -\\
\hspace{1em}0.5 & 0.056 & 0.062 & 0.058 & 0.058 & 0.056 & - & -\\
\hspace{1em}0.75 & 0.052 & 0.038 & 0.048 & 0.056 & 0.046 & - & 0.034\\
\bottomrule
\end{tabular}
\end{table}
\begin{table}[t]
\caption{Empirical size of each test for the simulation study discussed in Section~\ref{sec::sim} for $n=500$, $\epsilon=1$.}
\centering
\begin{tabular}{llllll|ll}
$q$ & $\tan^{-1}(r)$ & $\ln(r + 1)$ & $r^{0.5}$ & $r$ & $r^2$ & ToT $r$ & ToT $\tan^{-1}(r)$ \\
\hline
\multicolumn{8}{l}{\textbf{Normal}}\\
\hspace{1em}0 & 0.052 & 0.05 & 0.052 & 0.062 & 0.056 & 0.046 & -\\
\hspace{1em}0.25 & 0.058 & 0.072 & 0.038 & 0.07 & 0.07 & - & -\\
\hspace{1em}0.5 & 0.066 & 0.07 & 0.044 & 0.088 & 0.056 & - & -\\
\hspace{1em}0.75 & 0.06 & 0.056 & 0.048 & 0.062 & 0.06 & - & 0.044\\
\addlinespace[0.3em]
\multicolumn{8}{l}{\textbf{Exponential}}\\
\hspace{1em}0 & 0.04 & 0.054 & 0.04 & 0.04 & 0.068 & 0.046 & -\\
\hspace{1em}0.25 & 0.058 & 0.062 & 0.04 & 0.028 & 0.034 & - & -\\
\hspace{1em}0.5 & 0.062 & 0.05 & 0.054 & 0.05 & 0.044 & - & -\\
\hspace{1em}0.75 & 0.04 & 0.046 & 0.062 & 0.052 & 0.064 & - & 0.048\\
\addlinespace[0.3em]
\multicolumn{8}{l}{\textbf{Lomax}}\\
\hspace{1em}0 & 0.04 & 0.048 & 0.042 & 0.062 & 0.058 & 0.062 & -\\
\hspace{1em}0.25 & 0.06 & 0.052 & 0.036 & 0.044 & 0.05 & - & -\\
\hspace{1em}0.5 & 0.064 & 0.044 & 0.054 & 0.046 & 0.048 & - & -\\
\hspace{1em}0.75 & 0.062 & 0.05 & 0.056 & 0.036 & 0.032 & - & 0.068\\
\bottomrule
\end{tabular}
\end{table}
\begin{table}[h]
\caption{Empirical size of each test for the simulation study discussed in Section~\ref{sec::sim} for $n=500$, $\epsilon=5$.}
\centering
\begin{tabular}{llllll|ll}
\toprule
$q$ & $\tan^{-1}(r)$ & $\ln(r + 1)$ & $r^{0.5}$ & $r$ & $r^2$ & ToT $r$ & ToT $\tan^{-1}(r)$ \\
\hline
\multicolumn{8}{l}{\textbf{Normal}}\\
\hspace{1em}0 & 0.062 & 0.042 & 0.048 & 0.054 & 0.04 & 0.064 & -\\
\hspace{1em}0.25 & 0.06 & 0.05 & 0.052 & 0.056 & 0.046 & - & -\\
\hspace{1em}0.5 & 0.052 & 0.048 & 0.046 & 0.048 & 0.052 & - & -\\
\hspace{1em}0.75 & 0.054 & 0.042 & 0.036 & 0.04 & 0.05 & - & 0.044\\
\addlinespace[0.3em]
\multicolumn{8}{l}{\textbf{Exponential}}\\
\hspace{1em}0 & 0.062 & 0.062 & 0.062 & 0.064 & 0.052 & 0.04 & -\\
\hspace{1em}0.25 & 0.06 & 0.06 & 0.056 & 0.066 & 0.062 & - & -\\
\hspace{1em}0.5 & 0.066 & 0.076 & 0.066 & 0.062 & 0.044 & - & -\\
\hspace{1em}0.75 & 0.05 & 0.048 & 0.052 & 0.054 & 0.052 & - & 0.06\\
\addlinespace[0.3em]
\multicolumn{8}{l}{\textbf{Lomax}}\\
\hspace{1em}0 & 0.04 & 0.044 & 0.046 & 0.054 & 0.05 & 0.056 & -\\
\hspace{1em}0.25 & 0.046 & 0.054 & 0.058 & 0.05 & 0.042 & - & -\\
\hspace{1em}0.5 & 0.052 & 0.05 & 0.036 & 0.046 & 0.05 & - & -\\
\hspace{1em}0.75 & 0.038 & 0.044 & 0.036 & 0.05 & 0.048 & - & 0.054\\
\bottomrule
\end{tabular}
\end{table}
\begin{table}[h]
\caption{Empirical size of each test for the simulation study discussed in Section~\ref{sec::sim} for $n=1000$, $\epsilon=0.5$.}

\centering
\begin{tabular}{llllll|ll}

$q$ & $\tan^{-1}(r)$ & $\ln(r + 1)$ & $r^{0.5}$ & $r$ & $r^2$ & ToT $r$ & ToT $\tan^{-1}(r)$ \\
\midrule
\multicolumn{8}{c}{$\epsilon = 0.5$, distribution: Normal} \\
\hline
0 & 0.048 & 0.048 & 0.066 & 0.046 & 0.038 & 0.058 & - \\
0.25 & 0.048 & 0.030 & 0.038 & 0.034 & 0.044 & - & - \\
0.5 & 0.048 & 0.038 & 0.044 & 0.062 & 0.066 & - & - \\
0.75 & 0.046 & 0.056 & 0.052 & 0.062 & 0.056 & - & 0.120 \\
\hline
\multicolumn{8}{c}{$\epsilon = 0.5$, distribution: Exponential} \\
\hline
0 & 0.036 & 0.056 & 0.040 & 0.046 & 0.046 & 0.064 & - \\
0.25 & 0.052 & 0.070 & 0.066 & 0.038 & 0.044 & - & - \\
0.5 & 0.060 & 0.044 & 0.056 & 0.054 & 0.056 & - & - \\
0.75 & 0.048 & 0.020 & 0.050 & 0.052 & 0.054 & - & 0.102 \\
\hline
\multicolumn{8}{c}{$\epsilon = 0.5$, distribution: Lomax} \\
\hline
0 & 0.042 & 0.046 & 0.034 & 0.044 & 0.052 & 0.056 & - \\
0.25 & 0.054 & 0.046 & 0.038 & 0.046 & 0.048 & - & - \\
0.5 & 0.058 & 0.054 & 0.046 & 0.056 & 0.052 & - & - \\
0.75 & 0.034 & 0.042 & 0.050 & 0.042 & 0.068 & - & 0.128 \\
\bottomrule
\end{tabular}
\end{table}
\begin{table}[h]
\caption{Empirical size of each test for the simulation study discussed in Section~\ref{sec::sim} for $n=1000$, $\epsilon=1$.}
\centering
\begin{tabular}{llllll|ll}
$q$ & $\tan^{-1}(r)$ & $\ln(r + 1)$ & $r^{0.5}$ & $r$ & $r^2$ & ToT $r$ & ToT $\tan^{-1}(r)$ \\
\hline
\addlinespace[0.3em]
\multicolumn{8}{l}{\textbf{Normal}}\\
\hspace{1em}0 & 0.056 & 0.062 & 0.06 & 0.082 & 0.04 & 0.058 & -\\
\hspace{1em}0.25 & 0.05 & 0.05 & 0.048 & 0.04 & 0.06 & - & -\\
\hspace{1em}0.5 & 0.042 & 0.042 & 0.046 & 0.048 & 0.04 & - & -\\
\hspace{1em}0.75 & 0.048 & 0.052 & 0.066 & 0.048 & 0.03 & - & 0.082\\
\addlinespace[0.3em]
\multicolumn{8}{l}{\textbf{Exponential}}\\
\hspace{1em}0 & 0.046 & 0.044 & 0.048 & 0.06 & 0.044 & 0.052 & -\\
\hspace{1em}0.25 & 0.04 & 0.036 & 0.044 & 0.056 & 0.054 & - & -\\
\hspace{1em}0.5 & 0.04 & 0.054 & 0.05 & 0.052 & 0.05 & - & -\\
\hspace{1em}0.75 & 0.05 & 0.042 & 0.06 & 0.048 & 0.052 & - & 0.054\\
\addlinespace[0.3em]
\multicolumn{8}{l}{\textbf{Lomax}}\\
\hspace{1em}0 & 0.044 & 0.038 & 0.058 & 0.058 & 0.034 & 0.05 & -\\
\hspace{1em}0.25 & 0.05 & 0.044 & 0.056 & 0.046 & 0.042 & - & -\\
\hspace{1em}0.5 & 0.046 & 0.034 & 0.034 & 0.042 & 0.04 & - & -\\
\hspace{1em}0.75 & 0.048 & 0.052 & 0.04 & 0.056 & 0.044 & - & 0.064\\
\bottomrule
\end{tabular}
\end{table}
\begin{table}[h]
\caption{Empirical size of each test for the simulation study discussed in Section~\ref{sec::sim} for $n=1000$, $\epsilon=5$.}
\centering
\begin{tabular}{llllll|ll}
$q$ & $\tan^{-1}(r)$ & $\ln(r + 1)$ & $r^{0.5}$ & $r$ & $r^2$ & ToT $r$ & ToT $\tan^{-1}(r)$ \\
\hline
\addlinespace[0.3em]
\multicolumn{8}{l}{\textbf{Normal}}\\
\hspace{1em}0 & 0.07 & 0.048 & 0.05 & 0.038 & 0.038 & 0.034 & -\\
\hspace{1em}0.25 & 0.05 & 0.048 & 0.04 & 0.046 & 0.036 & - & -\\
\hspace{1em}0.5 & 0.042 & 0.04 & 0.04 & 0.038 & 0.052 & - & -\\
\hspace{1em}0.75 & 0.036 & 0.05 & 0.044 & 0.048 & 0.056 & - & 0.06\\
\addlinespace[0.3em]
\multicolumn{8}{l}{\textbf{Exponential}}\\
\hspace{1em}0 & 0.044 & 0.058 & 0.056 & 0.058 & 0.072 & 0.048 & -\\
\hspace{1em}0.25 & 0.06 & 0.066 & 0.06 & 0.05 & 0.052 & - & -\\
\hspace{1em}0.5 & 0.064 & 0.064 & 0.062 & 0.054 & 0.06 & - & -\\
\hspace{1em}0.75 & 0.05 & 0.048 & 0.046 & 0.05 & 0.066 & - & 0.056\\
\addlinespace[0.3em]
\multicolumn{8}{l}{\textbf{Lomax}}\\
\hspace{1em}0 & 0.042 & 0.042 & 0.042 & 0.044 & 0.062 & 0.046 & -\\
\hspace{1em}0.25 & 0.058 & 0.06 & 0.048 & 0.052 & 0.046 & - & -\\
\hspace{1em}0.5 & 0.052 & 0.062 & 0.04 & 0.054 & 0.048 & - & -\\
\hspace{1em}0.75 & 0.06 & 0.058 & 0.06 & 0.064 & 0.042 & - & 0.062\\
\bottomrule
\end{tabular}
\end{table}
\subsection{Uneven group sizes}\label{app::group}
\begin{figure}[!htb]
    \begin{minipage}{.5\textwidth}
    \centering
    \includegraphics[width=\textwidth]{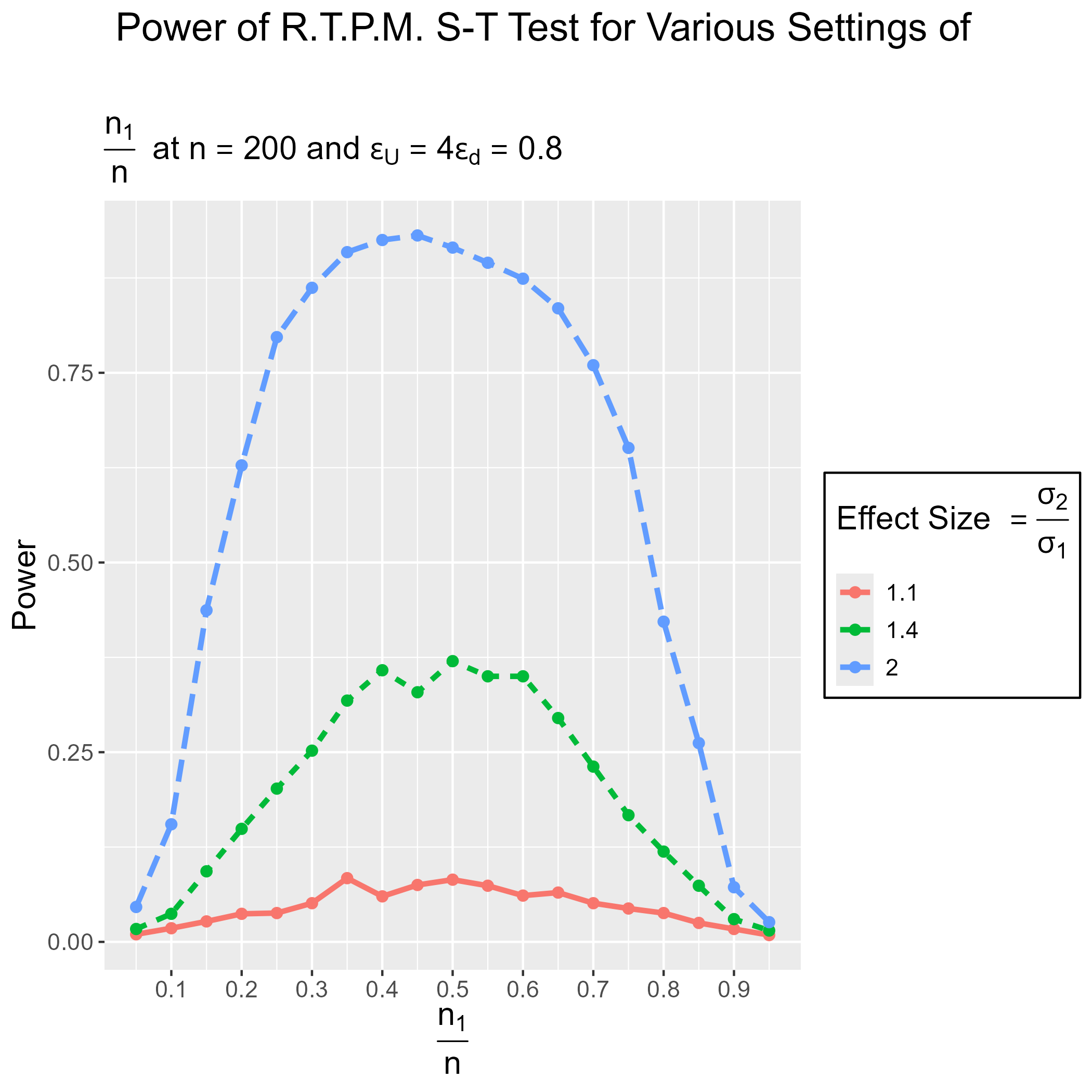}
    \end{minipage}
        \begin{minipage}{.5\textwidth}
    \centering
    \includegraphics[width=\textwidth]{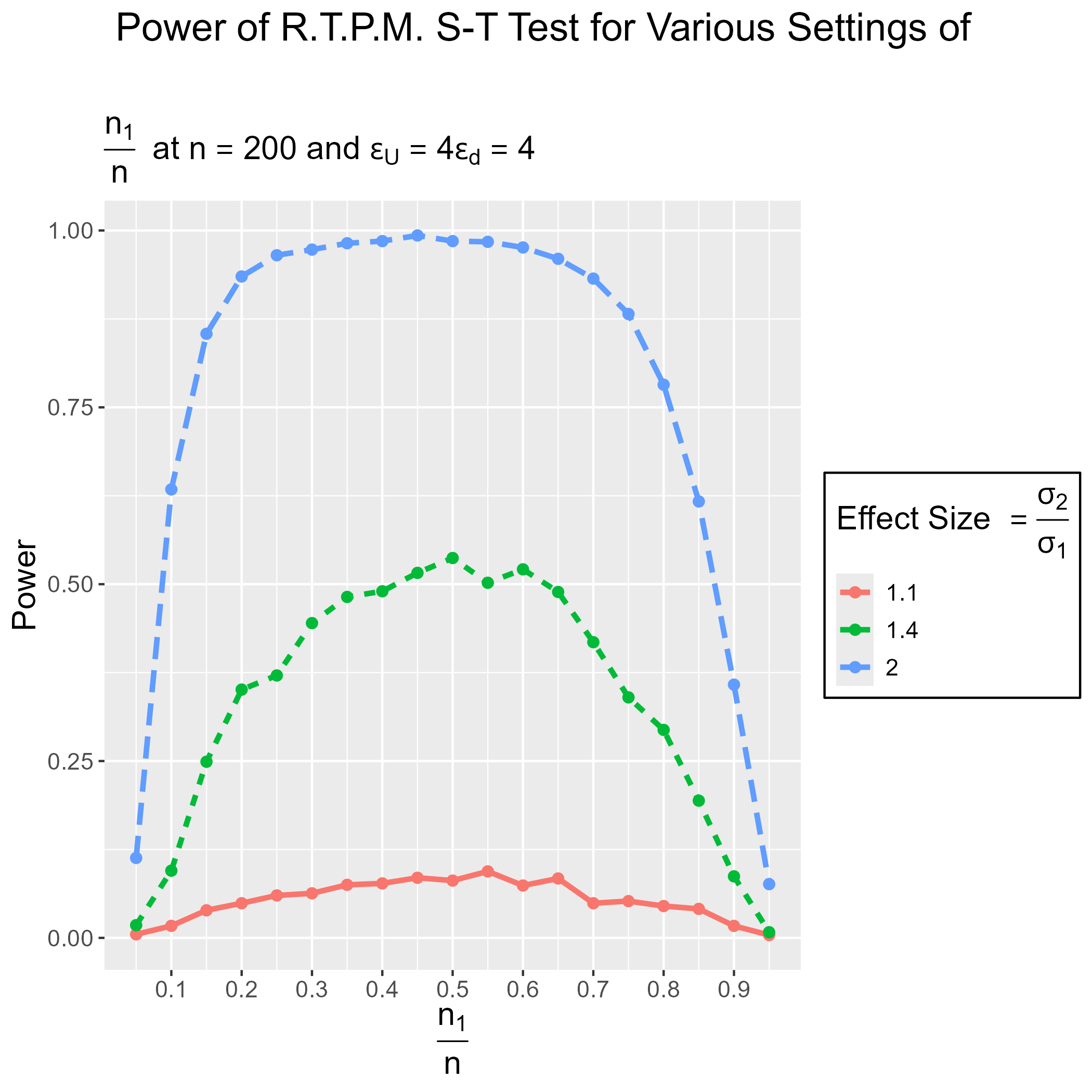}
    \end{minipage}\\
        \begin{minipage}{.5\textwidth}
    \centering
    \includegraphics[width=\textwidth]{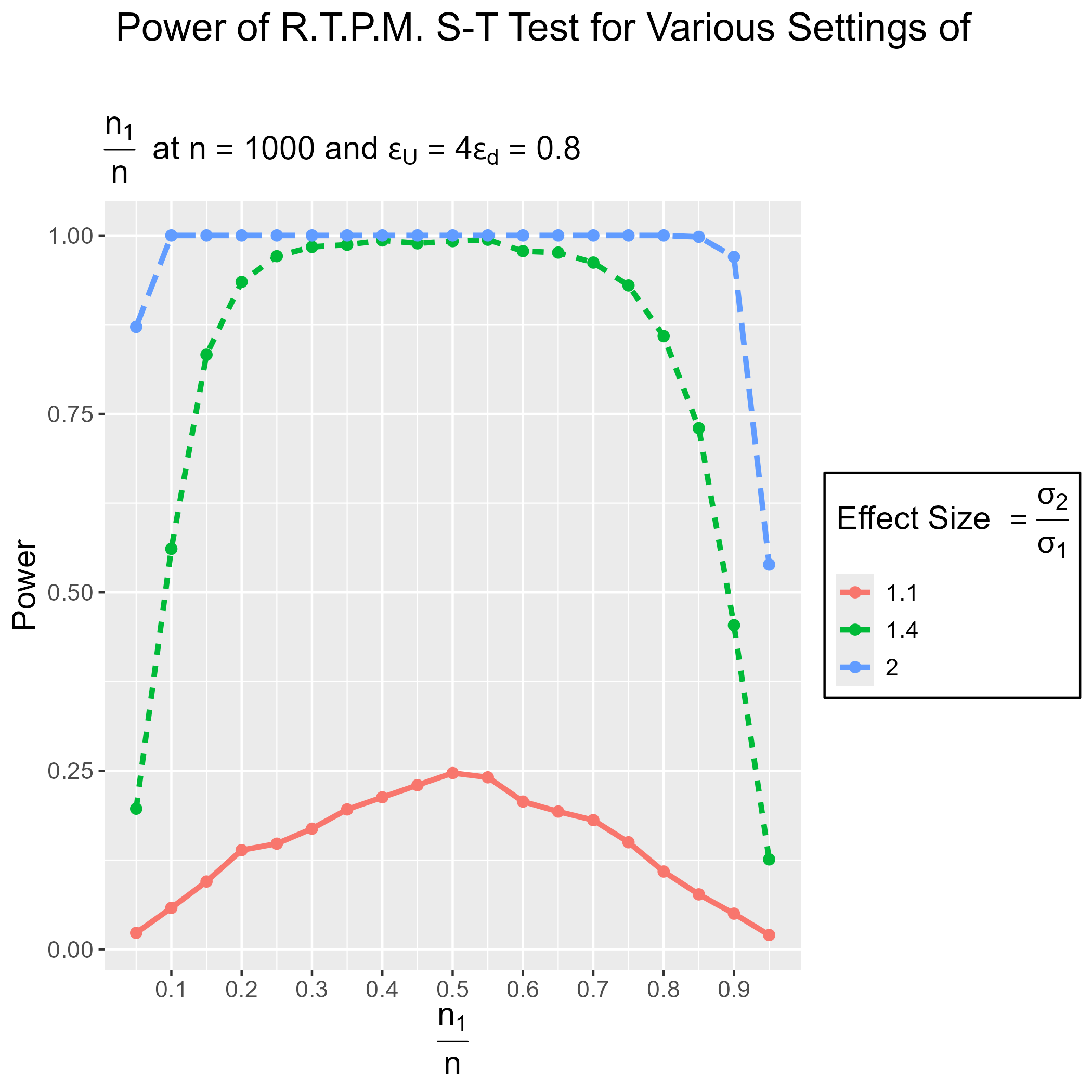}
    \end{minipage}
        \begin{minipage}{.5\textwidth}
    \centering
    \includegraphics[width=\textwidth]{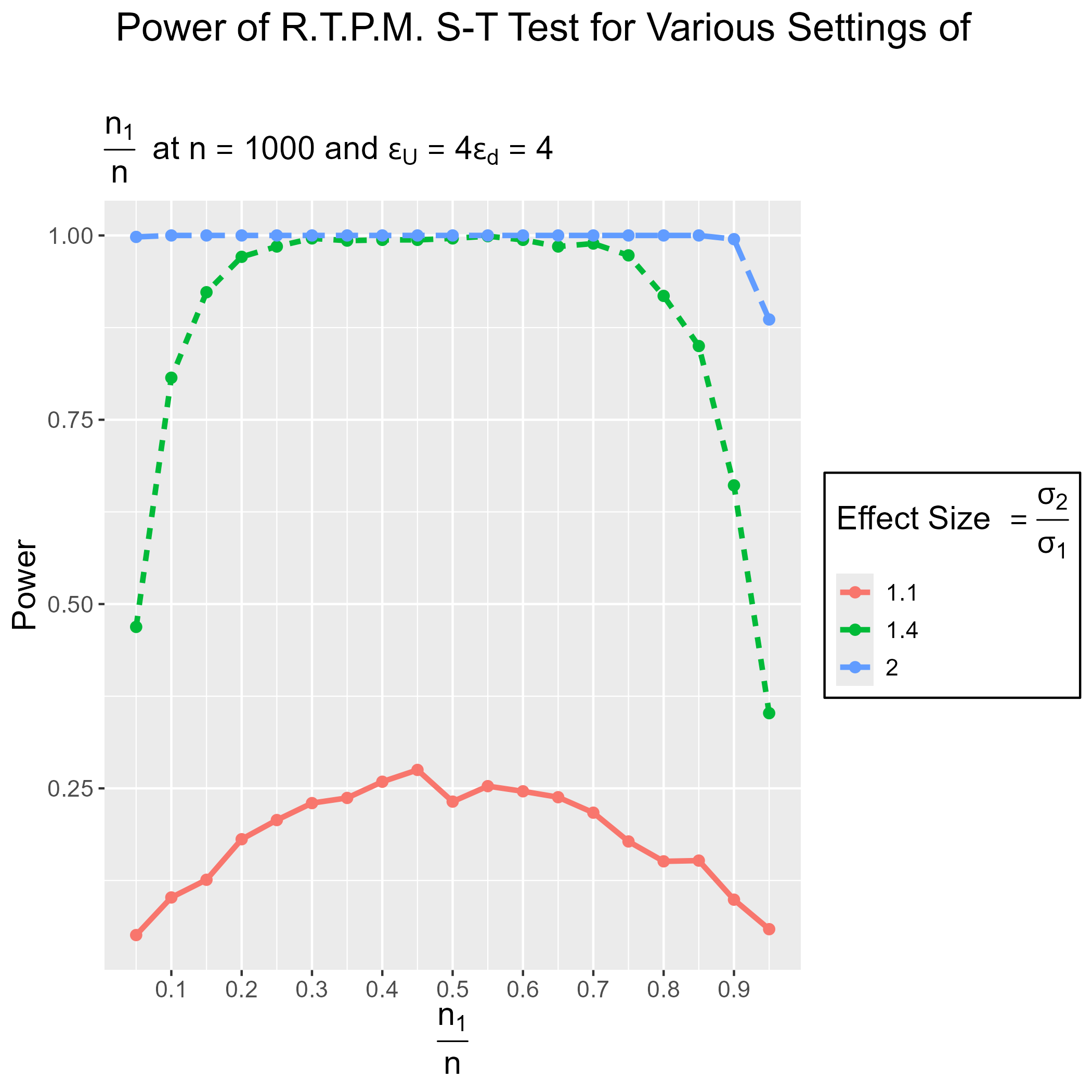}
    \end{minipage}
     \caption{Empirical power of the RPST test with different group sizes under normally distributed data. Here, $\psi=\tan^{-1}$.}
\end{figure}
\subsection{Privacy budget allocation}\label{app::budget}
\begin{figure}[!htb]
    \begin{minipage}{.5\textwidth}
    \centering
    \includegraphics[width=\textwidth]{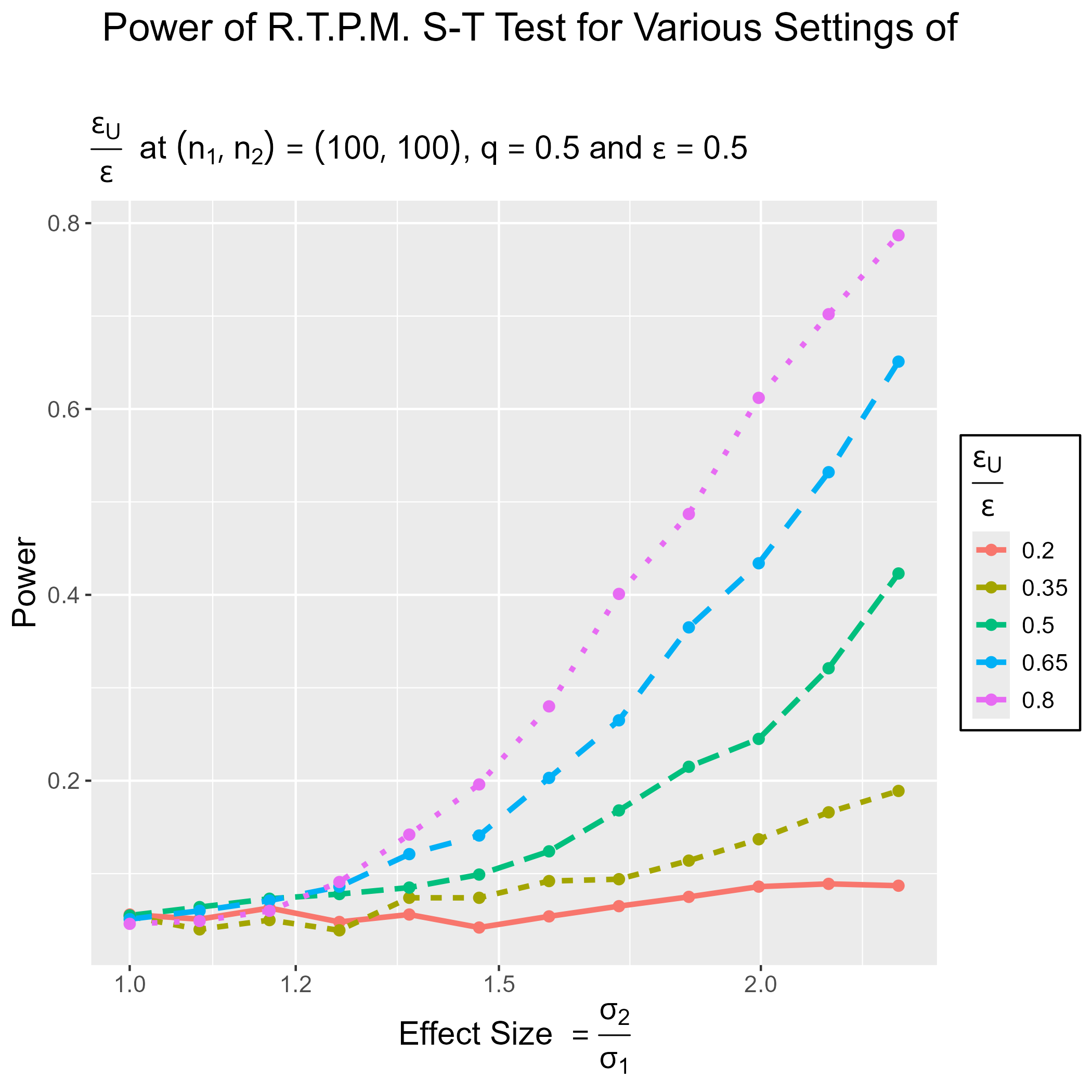}
    \label{fig:BA}
    \end{minipage}
        \begin{minipage}{.5\textwidth}
    \centering
    \includegraphics[width=\textwidth]{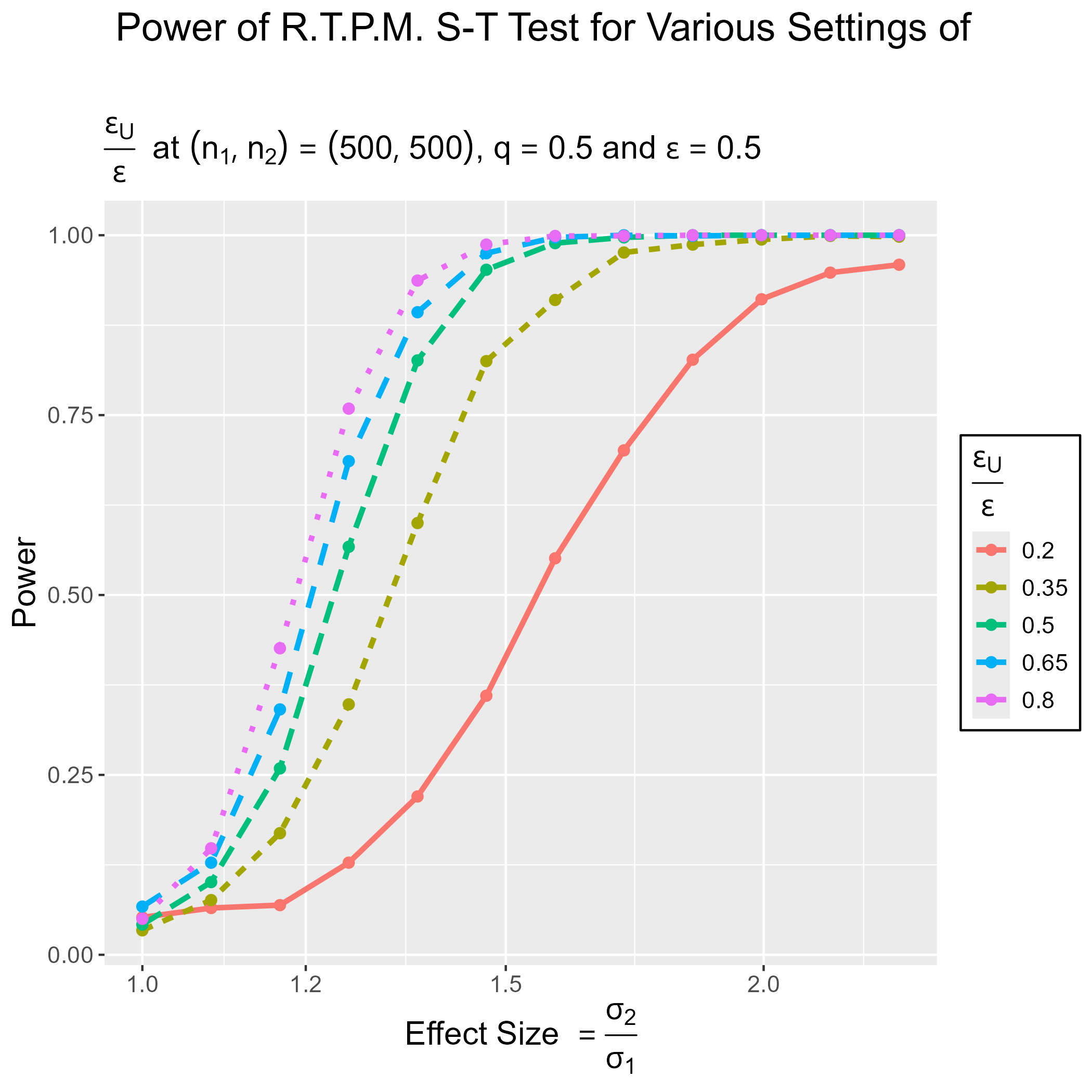}
    \label{fig:BA2}
    \end{minipage}\\
        \begin{minipage}{.5\textwidth}
    \centering
    \includegraphics[width=\textwidth]{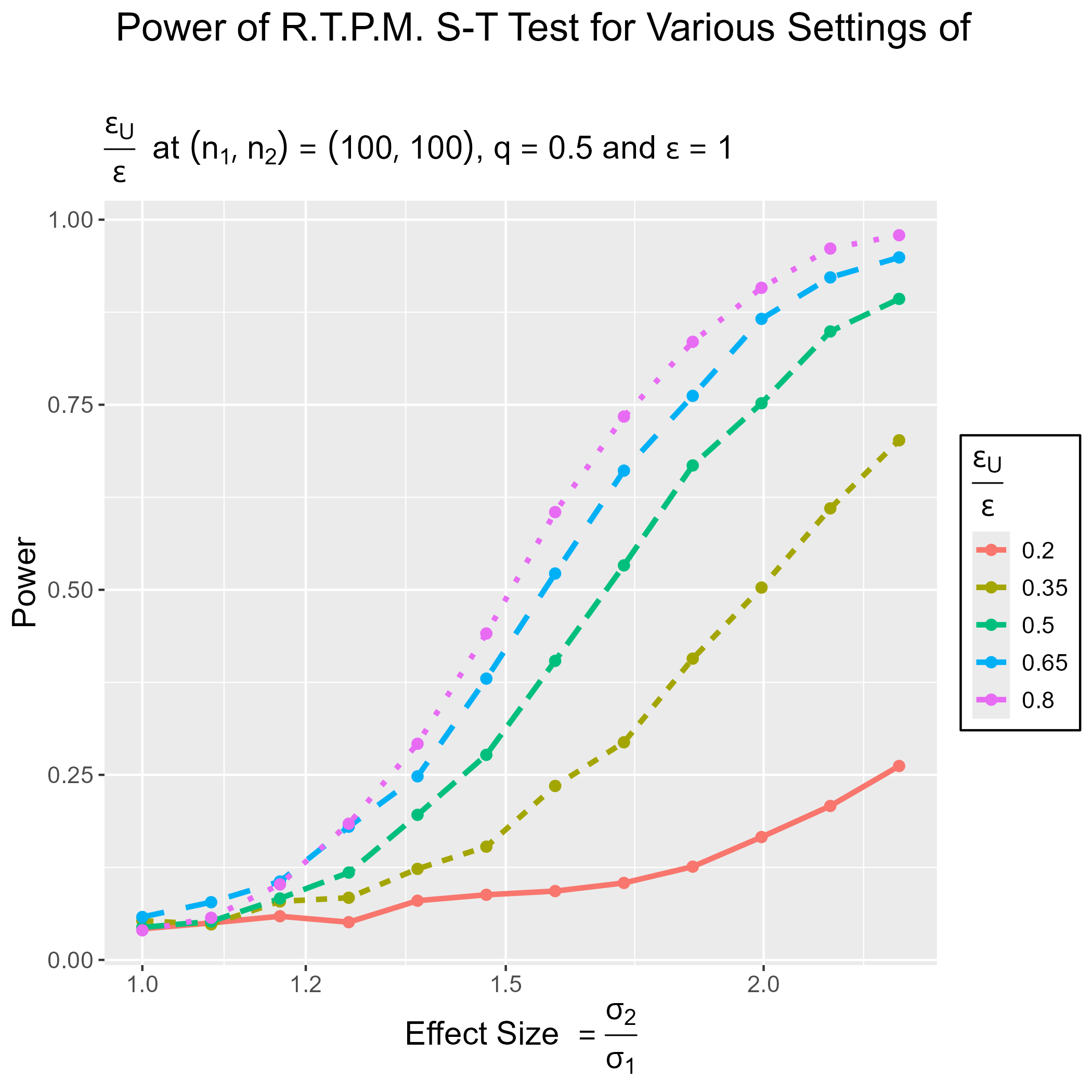}
    \label{fig:BA3}
    \end{minipage}
        \begin{minipage}{.5\textwidth}
    \centering
    \includegraphics[width=\textwidth]{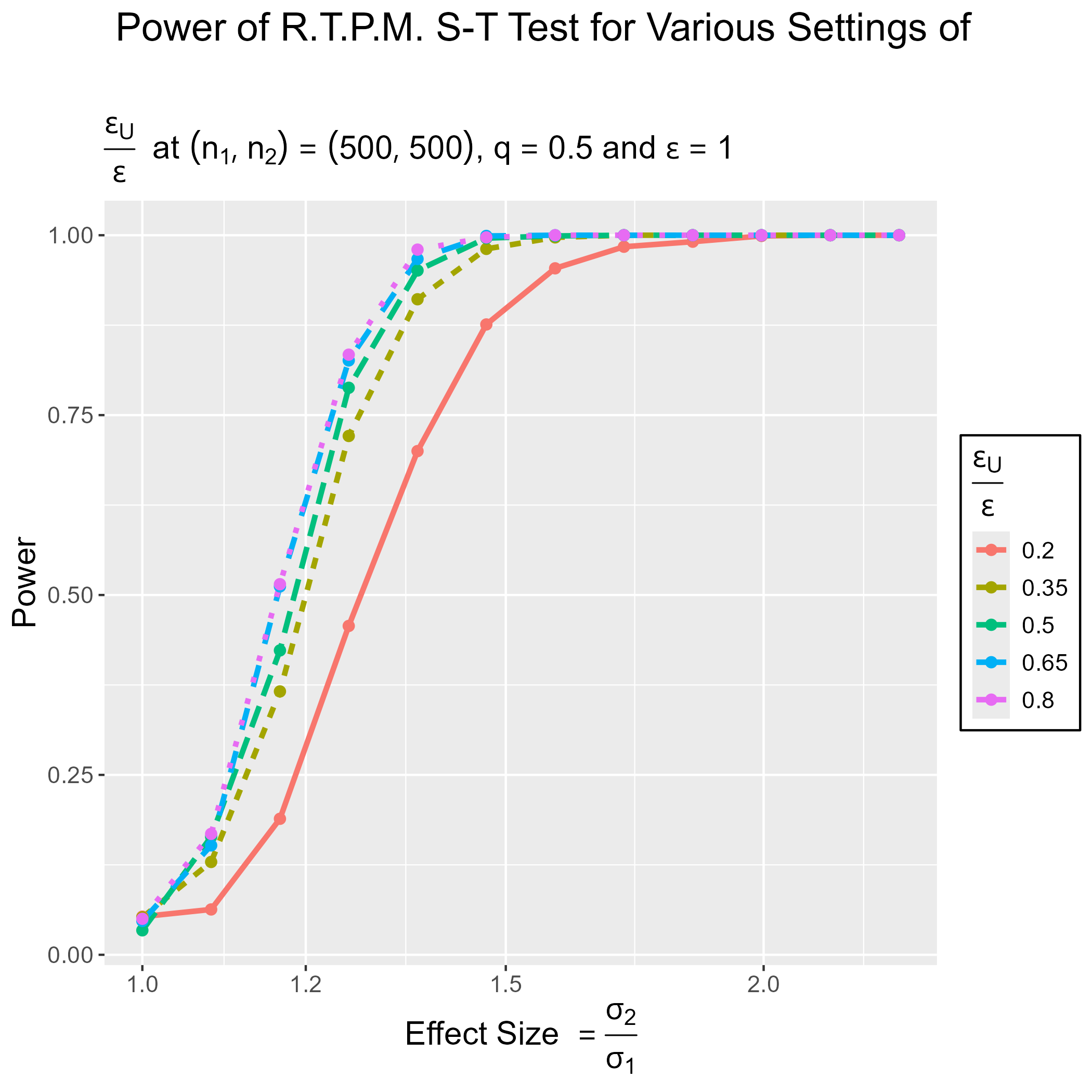}
    \label{fig:BA4}
    \end{minipage}
     \caption{Empirical power of the RPST test with different budget allocations under normally distributed data. Here, $\psi=\tan^{-1}$.}
\end{figure}

\end{document}